\documentclass[fleqn,a4paper]{article}
\usepackage{a4wide}
\usepackage{amsmath}
\usepackage{amssymb}
\usepackage{comment}
\usepackage{hyphenat}
\usepackage{multicol}

\newtheorem{theorem}{Theorem}[section]
\newtheorem{definition}[theorem]{Definition}
\newtheorem{proposition}[theorem]{Proposition}
\newtheorem{lemma}[theorem]{Lemma}
\newtheorem{corollary}[theorem]{Corollary}
\newtheorem{observation}{Observation}[section]

\newcommand{\qed}{\hfill$\Box$}
\newenvironment{proof}[1][]{\addvspace{.2cm} \noindent{\bf Proof: #1}}{\qed\vspace{.3cm}}

\newcommand{\bigO}{\ensuremath{\mathcal{O}}}
\newcommand{\bigOs}{\ensuremath{\mathcal{O}^*}}
\newcommand{\calC}{\ensuremath{\mathcal{C}}}
\newcommand{\NP}{\ensuremath{\mathcal{NP}}}
\newcommand{\sharpP}{\ensuremath{\mathcal{\#P}}}
\newcommand{\N}{\ensuremath{\mathbb{N}}}
\newcommand{\W}{\ensuremath{\mathcal{W}}}
\newcommand{\Z}{\ensuremath{\mathbb{Z}}}

\title{Fast Dynamic Programming on Graph Decompositions\footnote{Preliminary parts of this paper have appeared under the title `Dynamic Programming on Tree Decompositions Using Generalised Fast Subset Convolution' on the 17th Annual European Symposium on Algorithms (ESA 2009), Lecture Notes in Computer Science 5757, pages 566-577, and under the title `Faster Algorithms on Branch and Clique Decompositions' on the 35th International Symposium Mathematical Foundations of Computer Science (MFCS 2010), Lecture Notes in Computer Science 6281, pages 174-185.}}
\author{Johan M. M. van Rooij\footnote{Department of Information and Computing Sciences, Utrecht University, P. O. Box 80.089, NL-3508 TB Utrecht, The Netherlands, \texttt{jmmrooij@cs.uu.nl, hansb@cs.uu.nl}} \and Hans L. Bodlaender\footnotemark[\value{footnote}] \and Erik Jan van Leeuwen\footnote{Department of Informatics, University of Bergen, P. O. Box 7803, NO-5020 Bergen, Norway, \texttt{E.J.van.Leeuwen@ii.uib.no, Martin.Vatshelle@ii.uib.no}} \and Peter Rossmanith\footnote{Department of Computer Science, RWTH Aachen University, DE-52056 Aachen, Germany, \texttt{rossmani@cs.rwth-aachen.de}} \and \addtocounter{footnote}{-1}Martin Vatshelle\footnotemark[\value{footnote}]}

\begin{document}
\maketitle

\abstract{
In this paper, we consider three types of graph decompositions, namely tree decompositions, branch decompositions, and clique decompositions.
We improve the running time of dynamic programming algorithms on these graph decompositions for a large number of problems as a function of the treewidth, branchwidth, or cliquewidth, respectively.
Such algorithms have many practical and theoretical applications.

On tree decompositions of width~$k$, we improve the running time for {\sc Dominating Set} to $\bigOs(3^k)$.
Hereafter, we give a generalisation of this result to $[\rho,\sigma]$-domination problems with finite or cofinite $\rho$~and~$\sigma$.
For these problems, we give $\bigOs(s^k)$-time algorithms.
Here, $s$ is the number of `states' a vertex can have in a standard dynamic programming algorithm for such a problems.
Furthermore, we give an $\bigOs(2^k)$-time algorithm for counting the number of perfect matchings in a graph, and generalise this result to $\bigOs(2^k)$-time algorithms for many clique covering, packing, and partitioning problems.
On branch decompositions of width~$k$, we give an $\bigOs(3^{\frac{\omega}{2}k})$-time algorithm for {\sc Dominating Set}, an $\bigOs(2^{\frac{\omega}{2}k})$-time algorithm for counting the number of perfect matchings, and $\bigOs(s^{\frac{\omega}{2}k})$-time algorithms for $[\rho,\sigma]$-domination problems involving $s$~states with finite or cofinite $\rho$~and~$\sigma$.
Finally, on clique decompositions of width~$k$, we give $\bigOs(4^k)$-time algorithms for {\sc Dominating Set}, {\sc Independent Dominating Set}, and {\sc Total Dominating Set}.

The main techniques used in this paper are a generalisation of fast subset convolution, as introduced by Bj\"orklund et al., now applied in the setting of graph decompositions and augmented such that multiple states and multiple ranks can be used.
In the case of branch decompositions, we combine this approach with fast matrix multiplication, as suggested by Dorn.
Recently, Lokshtanov et al. have shown that some of the algorithms obtained in this paper have running times in which the base in the exponents is optimal, unless the Strong Exponential-Time Hypothesis fails.
}

\section{Introduction}
Width parameters of graphs and their related graph decompositions are important concepts in the theory of graph algorithms.
Many investigations show that problems that are \NP-hard on general graphs become polynomial or even linear-time solvable when restricted to the class of graphs in which a given width parameter is bounded.
However, the constant factors involved in the upper bound on the running times of such algorithms are often large and depend on the parameter.
Therefore, it is often useful to find algorithms where these factors grow as slow as possible as a function of the width parameter~$k$.

In this paper, we consider such algorithms involving three prominent graph-width parameters and their related decompositions: \emph{treewidth} and \emph{tree decompositions}, \emph{branchwidth} and \emph{branch decompositions}, and \emph{cliquewidth} and \emph{$k$-expressions} or \emph{clique decompositions}.
These three graph-width parameters are probably the most commonly used ones in the literature.
However, other parameters such as rankwidth~\cite{OumS06} or booleanwidth~\cite{Bui-XuanTV09} and their related decompositions also exist.

Most algorithms solving combinatorial problems using a graph-width parameter consist of two steps:
\begin{enumerate}
\item Find a graph decomposition of the input graph of small width.
\item Solve the problem by dynamic programming on this graph decomposition.
\end{enumerate}
In this paper, we will focus on the second of these steps and improve the running time of many known algorithms on all three discussed types of graph decompositions as a function of the width parameter.
The results have both theoretical and practical applications, some of which we survey below.

We obtain our results by using variants of the {\em covering product} and the {\em fast subset convolution} algorithm \cite{BjorklundHKK07} in conjunction with known techniques on these graph decompositions.
These two algorithms have been used to speed up other dynamic programming algorithms before, but not in the setting of graph decompositions.
Examples include algorithms for {\sc Steiner Tree}~\cite{BjorklundHKK07,Nederlof09}, graph motif problems~\cite{BetzlerFKN08}, and graph recolouring problems~\cite{PontaHN08}.
An important aspect of our results is an implicit generalisation of the fast subset convolution algorithm that is able to use multiple states.
This contrasts to the set formulation in which the covering product and subset convolution are defined: this formulation is equivalent to using two states (in and out).
Moreover, the fast subset convolution algorithm uses ranked M\"{o}bius transforms, while we obtain our results 
by using transformations that use multiple states and multiple ranks.
It is interesting to note that the state-based convolution technique that we use reminds of the technique used in Strassen's algorithm for fast matrix multiplication~\cite{Strassen69}.

Some of our algorithms also use \emph{fast matrix multiplication} to speed up dynamic programming as introduced by Dorn~\cite{Dorn06}.
To make this work efficiently, we introduce the use of asymmetric vertex states.
We note that matrix multiplication has been used for quite some time as a basic tool for solving combinatorial problems.
It has been used for instance in the {\sc All Pairs Shortest Paths} problem \cite{Seidel95}, in recognising triangle-free graphs~\cite{ItaiR78}, and in computing graph determinants.
One of the results of this paper is that (generalisations of) fast subset convolution and fast matrix multiplication can be combined to obtain faster dynamic programming algorithms for many optimisation problems.

\paragraph{Treewidth-Based Algorithms.}
Tree-decomposition-based algorithms can be used to effectively solve combinatorial problems on graphs of small treewidth both in theory and in practice.
Practical algorithms exist for problems like partial constraint satisfaction~\cite{KosterHK02}.
Furthermore, tree-decomposition-based algorithms are used as subroutines in many areas such as approximation algorithms~\cite{DemaineH08,Eppstein00}, parameterised algorithms~\cite{CyganNPPRW11,DemaineFHT05,MolleRR08,ThilikosSB05}, exponential-time algorithms~\cite{FominGSS09,ScottS07,vanRooijND09}, and subexponential-time algorithms~\cite{BodlaenderR10,CyganNPPRW11,FominT04}.

Many \NP-hard problems can be solved in polynomial time on a graphs whose treewidth is bounded by a constant.
If we assume that a graph~$G$ is given with a tree decomposition~$T$ of~$G$ of width~$k$, then the running time of such an algorithm is typically polynomial in the size of~$G$, but exponential in the treewidth~$k$.
Examples of such algorithms include many kinds of vertex partitioning problems (including the $[\rho,\sigma]$-domination problems) \cite{TelleP97}, edge colouring problems such as {\sc Chromatic Index} \cite{Bodlaender90}, or other problems such as {\sc Steiner Tree}~\cite{KorachS90}.

Concerning the first step of the general two-step approach above, we note that finding a tree decomposition of minimum width is \NP-hard~\cite{ArnborgCP87}.
For fixed~$k$, one can find a tree decomposition of width at most~$k$ in linear time, if such a decomposition exists~\cite{Bodlaender96}.
However, the constant factor involved in the running time of this algorithm is very high.
On the other hand, tree decompositions of small width can be obtained efficiently for special graph classes~\cite{Bodlaender98}, and there are also several good heuristics that often work well in practice~\cite{BodlaenderK10}.

Concerning the second step of this two-step approach, there are several recent results about the running time of algorithms on tree decompositions, with special considerations for the running time as a function of the width of the tree decomposition~$k$.
For several vertex partitioning problems, Telle and Proskurowski showed that there are algorithms that, given a graph with a tree decomposition of width~$k$, solve these problems in $\bigO(c^k n)$ time~\cite{TelleP97}, where~$c$ is a constant that depends only on the problem at hand.
For {\sc Dominating Set}, Alber and Niedermeier gave an improved algorithm that runs in $\bigO(4^k n)$ time~\cite{AlberN02}.
Similar results are given in~\cite{AlberBFKN02} for related problems: {\sc Independent Dominating Set}, {\sc Total Dominating Set}, {\sc Perfect Dominating Set}, {\sc Perfect Code}, {\sc Total Perfect Dominating Set}, {\sc Red-Blue Dominating Set}, and weighted versions of these problems.

If the input graph is planar, then other improvements are possible.
Dorn showed that {\sc Dominating Set} on planar graphs given with a tree decomposition of
width $k$ can be solved in $O^*(3^k)$ time~\cite{Dorn10}; he also gave similar improvements for other problems. 
We obtain the same result without requiring planarity.

In this paper, we show that the number of dominating sets of each given size in a graph can be counted in $\bigOs(3^k)$ time.
After some modifications, this leads to an $\bigOs(nk^23^k)$-time algorithm for {\sc Dominating Set}.
We also show that one can count the number of perfect matchings in a graph in $\bigOs(2^k)$ time, and we generalise these results to the $[\rho,\sigma]$-domination problems, as defined in~\cite{TelleP97}.

For these $[\rho,\sigma]$-domination problems, we show that they can be solved in $\bigOs(s^k)$ time, where~$s$ is the natural number of states required to represent partial solutions.
Here, $\rho$ and $\sigma$ are subsets of the natural numbers, and each choice of these subsets defines a different combinatorial problem.
The only restriction that we impose on these problems is that we require both $\rho$ and $\sigma$ to be either finite or cofinite.
That such an assumption is necessary follows from Chappelle's recent result~\cite{Chapelle10}: he shows that $[\rho,\sigma]$-domination problems are \W[1]-hard when parameterised by the treewidth of the graph if $\sigma$ is allowed to have arbitrarily large gaps between consecutive elements and $\rho$ is cofinite.
The problems to which our results apply include {\sc Strong Stable Set}, {\sc Independent Dominating Set}, {\sc Total Dominating Set}, {\sc Total Perfect Dominating Set}, {\sc Perfect Code}, {\sc Induced $p$-Regular Subgraph}, and many others.
Our results also extend to other similar problems such as {\sc Red-Blue Dominating Set} and {\sc Partition Into Two Total Dominating Sets}.

Finally, we define families of problems that we call $\gamma$-clique covering, $\gamma$-clique packing, or $\gamma$-clique partitioning problems: these families generalise standard problems like {\sc Minimum Clique Partition} in the same way as the $[\rho,\sigma]$-domination problems generalise {\sc Domination Set}.
The resulting families of problems include {\sc Maximum Triangle Packing}, {\sc Partition Into $l$-Cliques} for fixed $l$, the problem to determine the minimum number of odd-size cliques required to cover~$G$, and many others.
For these $\gamma$-clique covering, packing, or partitioning problems, we give $\bigOs(2^k)$-time algorithms, improving the straightforward $\bigOs(3^k)$-time algorithms for these problems.

\paragraph{Branchwidth-Based Algorithms.}
Branch decompositions are closely related to tree decompositions.
Like tree decompositions, branch decompositions have shown to be an effective tool for solving many combinatorial problems with both theoretical and practical applications.
They are used extensively in designing algorithms for planar graphs and for graphs excluding a fixed minor.
In particular, most of the recent results aimed at obtaining faster exact or parameterised algorithms on these graphs rely on branch decompositions~\cite{Dorn06,DornFT08,FominT04,FominT06}.
Practical algorithms using branch decompositions include those for ring routing problems~\cite{CookS93}, and tour merging for the {\sc Travelling Salesman Problem}~\cite{CookS03}.

Concerning the first step of the general two-step approach, we note that finding the branchwidth of a graph is \NP-hard in general~\cite{SeymourT94}.
For fixed~$k$, one can find a branch decomposition of width~$k$ in linear time, if such a decomposition exists, by combining the results from~\cite{Bodlaender96} and~\cite{BodlaenderT97}.
This is similar to tree decompositions, and the constant factors involved in the running time of this algorithms are very large.
However, in contrast to tree decompositions for which the complexity on planar graphs is unknown, there exists a polynomial-time algorithm that computes a branch decomposition of minimal width of a planar graph~\cite{SeymourT94}.
For general graphs several useful heuristics exist~\cite{CookS93,CookS03,Hicks02}.

Concerning the second step of the general two-step approach, Dorn has shown how to use fast matrix multiplication to speed up dynamic programming algorithms on branch decompositions~\cite{Dorn06}.
Among others, he gave an $\bigOs(4^k)$-time algorithm for the {\sc Dominating Set} problem.
On planar graphs, faster algorithms exist using so-called sphere-cut branch decompositions~\cite{DornPBF05}.
On these graphs, {\sc Dominating Set} can be solved in $\bigOs(3^{\frac{\omega}{2}k})$ time, where $\omega$ is the smallest constant such that two $n \times n$ matrices can be multiplied in $\bigO(n^\omega)$ time.
Some of these results can be generalised to graphs that avoid a minor~\cite{DornFT08}.
We obtain the same results without imposing restrictions on the class of graphs to which our algorithms can be applied.

In this paper, we show that one can count the number of dominating sets of each given size in a graph in $\bigOs(3^{\frac{\omega}{2}k})$ time.
We also show that one can count the number of perfect matchings in a graph in $\bigOs(2^{\frac{\omega}{2}k})$ time, and we show that the $[\rho,\sigma]$-domination problems with finite or cofinite $\rho$ and $\sigma$ can be solved in $\bigOs(s^{\frac{\omega}{2}k})$ time, where $s$ is again the natural number of states required to represent partial solutions.

\paragraph{Cliquewidth Based Algorithms.}
The notion of cliquewidth was first studied by Courcelle et al.~\cite{CourcelleER93}.
The graph decomposition associated with cliquewidth is a $k$-expression, which is sometimes also called a clique decomposition.
Similar to treewidth and branchwidth, many well-known problems can be solved in polynomial time on graphs which cliquewidth is bounded by a constant~\cite{CourcelleMR00}.

Whereas the treewidth and branchwidth of any graph are closely related, its cliquewidth can be very different from both of them.
For example, the treewidth of the complete graph on $n$~vertices is equal to $n-1$, while its cliquewidth is equal to~$2$.
However, the cliquewidth of a graph is always bounded by a function of its treewidth~\cite{CourcelleO00}.
This makes cliquewidth an interesting graph parameter to consider on graphs where the treewidth or branchwidth is too high for obtaining efficient algorithms.

Concerning the first step of the general two-step approach, we note that, like the other two parameters, computing the cliquewidth of general graphs is \NP-hard~\cite{FellowsRRS09}.
However, graphs of cliquewidth 1, 2, or 3 can be recognised in polynomial time~\cite{CorneilHLRR00}.
For $k \ge 4$, there is a fixed-parameter-tractable algorithm that, given a graph of cliquewidth~$k$, outputs a $2^{k+1}$ expression.

Concerning the second step, the first singly-exponential-time algorithm for {\sc Dominating Set} on clique decompositions of width~$k$ is an $\bigOs(16^k)$-time algorithm by Kobler and Rotics~\cite{KoblerR03}.
The previously fastest algorithm for this problem has a running time of $\bigOs(8^k)$, obtained by transforming the problem to a problem on boolean decompositions~\cite{Bui-XuanTV09}.
In this paper, we present a direct algorithm that runs in $\bigOs(4^k)$ time.
We also show that one can count the number of dominating sets of each given size at the cost of an extra polynomial factor in the running time.
Furthermore, we show that one can solve {\sc Independent Dominating Set} in $\bigOs(4^k)$ and {\sc Total Dominating Set} in  $\bigOs(4^k)$ time.

We note that there are other width parameters of graphs that potentially have lower values than cliquewidth, for example rankwidth (see~\cite{OumS06}) and booleanwidth (see~\cite{Bui-XuanTV09}).
These width parameters are related since a problem is fixed-parameter tractable parameterised by cliquewidth if and only if it is fixed-parameter tractable parameterised by rankwidth or booleanwidth~\cite{Bui-XuanTV09,OumS06}.
However, for many problems the best known running times for these problems are often much better as a function of the cliquewidth than as a function of the rankwidth or booleanwidth.
For example, dominating set can be solved on rank decompositions of width~$k$ in $\bigOs(2^{\frac{3}{4}k^2+\frac{23}{4}k})$ time~\cite{Bui-XuanTV10,GanianH10} and on boolean decompositions of width~$k$ in $\bigO(8^k)$ time~\cite{Bui-XuanTV09}.

\paragraph{Optimality, Polynomial Factors.}
We note that our results attain, or are very close to, intuitive natural lower bounds for the problems considered, namely a polynomial times the amount of space used by any dynamic programming algorithm for these problems on graph decompositions.
Similarly, it makes sense to think about the number of states necessary to represent partial solutions as the best possible base of the exponent in the running time: this equals the space requirements.
Currently, this is $\bigOs(3^k)$ space for {\sc Dominating Set} on tree decompositions and branch decompositions and $\bigOs(4^k)$ space on clique decompositions.

Very recently, this intuition has been strengthened by a result of Lokshtanov et al.~\cite{LokshtanovMS10}.
They prove that it is impossible to improve the exponential part of the running time for a number of tree-decomposition-based algorithms that we present in this paper, unless the \emph{Strong Exponential-Time Hypothesis} fails.
That is, unless there exist an algorithm for the general {\sc Satisfiability} problem running in $\bigO((2-\epsilon)^n)$ time, for any $\epsilon > 0$.
In particular, this holds for our tree decomposition based algorithms for {\sc Dominating Set} and {\sc Partition Into Triangles}.

We see that our algorithms on tree decompositions and clique decompositions all attain this intuitive lower bound.
On branch decompositions, we are very close.
When the number of states that one would naturally use to represent partial solutions equals~$s$, then our algorithms run in $\bigOs(s^{\frac{\omega}{2}k})$ time, where $\omega/2 < 1.188$.
Under the hypothesis that $\omega = 2$, which could be the true value of~$\omega$, our algorithms do attains this space bound.

Because of these seemingly-optimal exponential factors in the running times of our algorithms, we spend quite some effort to make the polynomial factors involved as small as possible.
In order to improve these polynomial factors, we need to distinguish between different problems based on a technical property for each type of graph decomposition that we call the \emph{de Fluiter property}.
This property is related to the concept of \emph{finite integer index}~\cite{BodlaenderA01,Fluiter97}.

Considering the polynomial factors of the running times of our algorithms sometimes leads to seemingly strange situations when the matrix multiplication constant is involved.
To see this, notice that $\omega$ is defined as the smallest constant such that two $n \times n$ matrices can be multiplied in $\bigO(n^\omega)$ time.
Consequently, any polylogarithmic factor in the running time of the corresponding matrix-multiplication algorithm disappears in an infinitesimal increase of~$\omega$.
These polylogarithmic factors are additional polynomial factors in the running times of our algorithms on branch decompositions.
In our analyses, we pay no extra attention to this, and we only explicitly give the polynomial factors involved that are not related to the time required to multiply matrices.

Also, because many of our algorithms use numbers which require more than a constant amount of bits to represent (often $n$-bit numbers are involved), the time and space required to represent these numbers and perform arithmetic operations on these numbers affects the polynomial factors in the running times of our algorithms.
We will always include these factors and highlight them using a special notation ($i_+(n)$ and $i_\times(n)$).

\paragraph{Model of Computation.}
In this paper, we use the \emph{Random Access Machine} (RAM) model with $\bigO(k)$-bit word size~\cite{FredmanW93} for the analysis of our algorithms.
In this model, memory access can be performed in constant time for memory of size~$\bigO(c^k)$, for any constant~$c$.
We consider addition and multiplication operations on $\bigO(k)$-bit numbers to be unit-time operations.
For an overview of this model, see for example~\cite{Hagerup98}.

We use this computational model because we do not want the table look-up operations to influence the polynomial factors of the running time.
Since the tables have size $\bigOs(s^k)$, for a problem specific integer $s \geq 2$, these operations are constant-time operations in this model.

\paragraph{Paper Organisation}
This paper is organised as follows.
We start with some preliminaries in Section~\ref{sec:prelim}.
In Section~\ref{sec:treewidth}, we present our results on tree decompositions.
This is followed by our results on branch decompositions in Section~\ref{sec:branchwidth} and clique decompositions in Section~\ref{sec:cliquewidth}.
To conclude, we briefly discuss the relation between the de Fluiter properties and finite integer index in Section~\ref{sec:fluiterprop}.
Finally, we give some concluding remarks in Section~\ref{sec:conclusion}.

\section{Preliminaries} \label{sec:prelim}
Let $G=(V,E)$ be an $n$-vertex graph with $m$ edges.
We denote the open neighbourhood of a vertex $v$ by $N(v)$ and the closed neighbourhood of a $v$ by $N[v]$, i.e., $N(v) = \{ u \in V \;|\; \{u,v\} \in E\}$ and $N[v] = N(v) \cup \{v\} $.
For a vertex subset~$U \subseteq V$, we denote by~$G[U]$ the subgraph induced by~$U$, i.e., $G[U] = (U, E \cap (U \times U))$.
We denote the powerset of a set $S$ by $2^S$.

For a decomposition tree~$T$, we often identify~$T$ with the set of nodes in~$T$, and we write~$E(T)$ for the edges of~$T$.
We often associate a table with each node or each edge in a decomposition tree $T$.
Such a table $A$ can be seen as a function, while we write $|A|$ for the size of $A$, that is, the total space required to store all entries of $A$.

We denote the time required to add and multiply $n$-bit numbers by $i_+(n)$ and $i_\times(n)$, respectively.
Currently, $i_\times(n) = n \log(n)2^{\bigO(\log^*(n))}$ due to F\"urer's algorithm \cite{Furer09}, and $i_+(n) = \bigO(n)$.
Note that $i_\times(k) = i_+(k) = \bigO(k)$ due to the chosen model of computation.

\subsection{Combinatorial Problems Studied} \label{sec:problems}
A \emph{dominating set} in a graph $G$ is a set of vertices $D \subseteq V$ such that for every vertex $v \in V \setminus D$ there exists a vertex $d \in D$ with $\{v,d\}\in E$, i.e, $\bigcup_{v \in D} N[v] = V$.
A dominating set $D$ in $G$ is a \emph{minimum dominating set} if it is of minimum cardinality among all dominating sets in $G$.
The classical \NP-hard problem {\sc Dominating Set} asks to find the size of a minimum dominating set in $G$.
Given a (partial) solution $D$ of {\sc Dominating Set}, we say that a vertex $d \in D$ \emph{dominates} a vertex $v$ if $v \in N[d]$, and that a vertex $v$ is \emph{undominated} if $N[v] \cap D = \emptyset$.
Besides the standard minimisation version of the problem, we also consider counting the number of minimum dominating sets, and counting the number of dominating sets of each given size.

A \emph{matching} in $G$ is a set of edges $M \subseteq E$ such that no two edges are incident to the same vertex.
A vertex that is an endpoint of an edge in $M$ is said to be \emph{matched} to the other endpoint of this edge.
A \emph{perfect matching} is a matchings in which every vertex $v \in V$ is matched.
Counting the number of perfect matchings in a graph is a classical \sharpP-complete problem~\cite{Valiant79}.

A \emph{$[\rho, \sigma]$-dominating set} is a generalisation of a dominating set introduced by Telle in \cite{Telle94,TelleP97}.
\begin{definition}[{$[\rho, \sigma]$}-dominating set]
Let $\rho, \sigma \subseteq \N$.
A $[\rho, \sigma]$-dominating set in a graph $G$ is a subset $D \subseteq V$ such that:
\begin{itemize}
\item for every $v \in V \setminus D$: $|N(v) \cap D| \in \rho$;
\item for every $v \in D$: $|N(v) \cap D| \in \sigma$.
\end{itemize}
\end{definition}

The $[\rho, \sigma]$-domination problems are the computational problems of finding $[\rho, \sigma]$-dominating sets; see Table~\ref{tab:rhosigma}.
Of these problems, we consider several variants: the \emph{$[\rho, \sigma]$-existence problems} ask whether a $[\rho, \sigma]$-dominating set exists in a graph $G$; the \emph{$[\rho, \sigma]$-minimisation} and \emph{$[\rho, \sigma]$-maximisation problems} ask for the minimum or maximum cardinality of a $[\rho, \sigma]$-dominating set in a graph $G$; and the \emph{$[\rho, \sigma]$-counting problems} ask for the number of $[\rho, \sigma]$-dominating sets in a graph $G$.
In a $[\rho, \sigma]$-counting problem, we sometimes restrict ourselves to counting $[\rho, \sigma]$-dominating sets of minimum size, maximum size, or of each given size.

Throughout this paper, we assume that $\rho$ and $\sigma$ are either finite or cofinite.

\begin{table}[tb]
	\begin{center}
	\begin{tabular}{|l|l||l|}
	\hline
	$\rho$ & $\sigma$ & Standard problem description \\
	\hline
	$\{0,1,\ldots\}$ & $\{0\}$ & Independent Set\\
	$\{1,2,\ldots\}$ & $\{0,1,\ldots\}$ & Dominating Set\\
	$\{0,1\}$ & $\{0\}$ & Strong Stable Set/2-Packing/ \\
	& & Distance-2 Independent Set\\
	$\{1\}$ & $\{0\}$ & Perfect Code/Efficient Dominating Set\\
	$\{1,2,\ldots\}$ & $\{0\}$ & Independent Dominating Set\\
	$\{1\}$ & $\{0,1,\ldots\}$ & Perfect Dominating Set\\
	$\{1,2,\ldots\}$ & $\{1,2,\ldots\}$ & Total Dominating Set\\
	$\{1\}$ & $\{1\}$ & Total Perfect Dominating Set\\
	$\{0,1\}$ & $\{0,1,\ldots\}$ & Nearly Perfect Set\\
	$\{0,1\}$ & $\{0,1\}$ & Total Nearly Perfect Set\\
	$\{1\}$ & $\{0,1\}$ & Weakly Perfect Dominating Set\\
	$\{0,1,\ldots\}$ & $\{0,1,\ldots,p\}$ & Induced Bounded Degree Subgraph\\
	$\{p,p+1,\ldots\}$ & $\{0,1,\ldots\}$ & $p$-Dominating Set \\
	$\{0,1,\ldots\}$ & $\{p\}$ & Induced $p$-Regular Subgraph \\
	\hline
	\end{tabular}
	\end{center}
	\caption{$[\rho,\sigma]$-domination problems (taken from \cite{Telle94,TelleP97}).}
	\label{tab:rhosigma}
\end{table}

Another type of problems we consider are clique covering, packing, and partitioning problems.
Because we want to give general results applying to many different problems, we will define a class of problems of our own: we define the notion of \emph{$\gamma$-clique covering}, \emph{$\gamma$-clique packing}, and \emph{$\gamma$-clique partitioning problems}.

We start by defining the $\gamma$-clique problems and note that their definitions somewhat resemble the definition of $[\rho,\sigma]$-domination problems.
\begin{definition}[$\gamma$-clique covering, packing, and partitioning] \label{def:cliqueproblems}
Let $\gamma \subseteq \N\setminus\{0\}$, let $G$ be a graph, and let $\calC$ be a collection of cliques from $G$ such that the size of every clique in $\calC$ is contained in $\gamma$.
We define the following notions:
\begin{itemize}
\item $\calC$ is a $\gamma$-clique cover of $G$ if $\calC$ covers the vertices of $G$, i.e, $\bigcup_{C \in \calC} C = V$.
\item $\calC$ is a $\gamma$-clique packing of $G$ if the cliques are disjoint, i.e, for any two $C_1, C_2 \in \calC$: $C_1 \cap C_2 = \emptyset$.
\item $\calC$ is a $\gamma$-clique partitioning of $G$ if it is both a $\gamma$-clique cover and a $\gamma$-clique packing.
\end{itemize}
\end{definition}

The corresponding computational problems are defined in the following way.
The $\gamma$-\emph{clique covering} problems asks for the cardinality of the smallest $\gamma$-clique cover.
The $\gamma$-\emph{clique packing} problems asks for the cardinality of the largest $\gamma$-clique packing.
The $\gamma$-\emph{clique partitioning} problems asks whether a $\gamma$-clique partitioning exists.
For these problems, we also consider their minimisation, maximisation, and counting variants.
See Table~\ref{tab:cliqueproblems} for some concrete example problems.
We note that clique covering problems in the literature often ask to cover all the edges of a graph: here we cover only the vertices.

Throughout this paper, we assume that $\gamma$ is decidable in polynomial time, that is, for every $j \in \N$ we can decide in time polynomial in $j$ whether $j \in \gamma$.

\begin{table}
	\begin{center}
	\begin{tabular}{|l|l||l|}
	\hline
	$\gamma$ & problem type & Standard problem description \\
	\hline
	$\{1,2,\ldots\}$ & partitioning, minimisation & Minimum clique partition \\
	$\{2\}$ & partitioning, counting & Count perfect matchings \\
	$\{3\}$ & covering & Minimum triangle cover of vertices \\
	$\{3\}$ & packing & Maximum triangle packing \\
	$\{3\}$ & partitioning & Partition into triangles \\
	$\{p\}$ & partitioning & Partition into $p$-cliques \\
	$\{1,3,5,7,\ldots\}$ & covering & Minimum cover by odd-cliques \\
	\hline
	\end{tabular}
	\end{center}
\caption{$\gamma$-clique covering, packing and partitioning problems.}
\label{tab:cliqueproblems}
\end{table}

\subsection{Graph Decompositions}
We consider dynamic programming algorithms on three different kinds of graph decompositions, namely tree decompositions, branch decompositions, and clique decompositions.

\subsubsection{Tree Decompositions}
The notions of a tree decomposition and treewidth were introduced by Robertson and Seymour \cite{RobertsonS86} and measure the tree-likeness of a graph.

\begin{definition}[tree decomposition]
A \emph{tree decomposition} of a graph~$G$ consists of a tree~$T$ in which each node~$x \in T$ has an associated set of vertices~$X_x \subseteq V$ (called a \emph{bag}) such that $\bigcup_{x \in T} X_x = V$ and the following properties hold:
\begin{enumerate}
\item for each $\{u,v\} \in E$, there exists a node $x \in T$ such that $\{u,v\} \in X_x$.
\item if $v \in X_x$ and $v \in X_y$, then $v \in X_z$ for all nodes $z$ on the path from node $x$ to node $y$ in~$T$.
\end{enumerate}
\end{definition}

The \emph{width}~$tw(T)$ of a tree decomposition~$T$ is the size of the largest bag of $T$ minus one.
The treewidth of a graph $G$ is the minimum treewidth over all possible tree decompositions of~$G$.
We note that the minus one in the definition exists to set the treewidth of trees to one.
In this paper, we will always assume that tree decompositions of the appropriate width are given.

Dynamic programming algorithms on tree decompositions are often presented on nice tree decompositions, which were introduced by Kloks~\cite{Kloks94}.
We give a slightly different definition of a nice tree decomposition.
\begin{definition}[nice tree decomposition] \label{def:nicetreedecomp}
A \emph{nice tree decomposition} is a tree decomposition with one special node $z$ called the \emph{root} with $X_z = \emptyset$ and in which each node is one of the following types:
\begin{enumerate}
\item \emph{Leaf node}: a leaf $x$ of $T$ with $X_x = \{v\}$ for some vertex $v \in V$.
\item \emph{Introduce node}: an internal node~$x$ of $T$ with one child node~$y$; this type of node has $X_x = X_y \cup \{v\}$, for some $v \notin X_y$. The node is said to \emph{introduce} the vertex~$v$.
\item \emph{Forget node}: an internal node~$x$ of $T$ with on child node~$y$; this type of node has $X_x = X_y \setminus \{v\}$, for some $v \in X_y$. The node is said to \emph{forget} the vertex~$v$.
\item \emph{Join node}: an internal node $x$ with two child nodes~$l$ and~$r$; this type of node has $X_x = X_r = X_l$.
\end{enumerate}
\end{definition}
We note that this definition is slightly different from the usual.
In our definition, we have the extra requirements that a bag $X_x$ associated with a leaf $x$ of $T$ consists of a single vertex $v$ ($X_x = \{v\}$), and that the bag $X_z$ associated with the root node $Z$ is empty ($X_z = \emptyset$).

Given a tree decomposition consisting of $\bigO(n)$ nodes, a nice tree decomposition of equal width and also consisting of $\bigO(n)$ nodes can be found in $\bigO(n)$ time~\cite{Kloks94}.
By adding a series of forget nodes to the old root, and by adding a series of introduce nodes below an old leaf node if its associated bag contains more than one vertex, we can easily modify any nice tree decomposition to have our extra requirements within the same running time.

By fixing the root of $T$, we associate with each node $x$ in a tree decomposition $T$ a vertex set $V_x \subseteq V$: a vertex $v$ belongs to $V_x$ if and only if there exists a bag $y$ with $v \in X_y$ such that either $y=x$ or $y$ is a descendant of $x$ in $T$.
Furthermore, we associate with each node $x$ of $T$ the induced subgraph $G_x = G[V_x]$ of $G$.
I.e., $G_x$ is the following graph:
\[ G_x = G\!\!\left[\bigcup \{ X_y \;|\; \textrm{$y = x$ or $y$ is a descendant of $x$} \} \right] \]

For an overview of tree decompositions and dynamic programming on tree decompositions,  see~\cite{BodlaenderK08,HicksKK05}.

\subsubsection{Branch Decompositions} \label{sec:defbw}
Branch decompositions are related to tree decompositions and also originate from the series of papers on graph minors by Robertson and Seymour \cite{RobertsonS91}.

\begin{definition}[branch decomposition]
A \emph{branch decomposition} of a graph~$G$ is a tree~$T$ in which each internal node has degree three and in which each leaf~$x$ of~$T$ has an assigned edge $e_x \in E$ such that this assignment is a bijection between the leaves of~$T$ and the edges~$E$ of~$G$.
\end{definition}

If we would remove any edge~$e$ from a branch decomposition~$T$ of~$G$, then this cuts~$T$ into two subtrees~$T_1$ and~$T_2$.
In this way, the edge~$e \in E(T)$ partitions the edges of~$G$ into two sets~$E_1$,~$E_2$, where~$E_i$ contains exactly those edges in the leaves of subtree~$T_i$.
The \emph{middle set} $X_e$ associated to the edge~$e \in E(T)$ is defined to be the set of vertices~$X_e \subseteq V$ that are both an endpoint of an edge in the edge partition~$E_1$ and an endpoint of an edge in the edge partition~$E_2$, where $E_1$ and $E_2$ are associated with $e$.
That is, if $V_i = \bigcup E_i$, then $X_e = V_1 \cap V_2$.

The \emph{width}~$bw(T)$ of a branch decomposition~$T$ is the size of the largest middle set associated with the edges of~$T$.
The branchwidth $bw(G)$ of a graph~$G$ is the minimum width over all possible branch decompositions of~$G$.
In this paper, we always assume that a branch decomposition of the appropriate width is given.

Observe that vertices $v$ of degree one in $G$ are not in any middle set of a branch decomposition $T$ of $G$.
Let $u$ be the neighbour of such a vertex $v$.
We include the vertex $v$ in the middle set of the edge $e$ of $T$ incident to the leaf of $T$ that contains $\{u,v\}$.
This raises the branchwidth to $\max\{2,bw(G)\}$. 
Throughout this paper, we ignore this technicality.

The treewidth $tw(G)$ and branchwidth $bw(G)$ of any graph are related in the following way:
\begin{proposition}[\cite{RobertsonS91}] \label{prop:1.5}
For any graph $G$ with branchwidth $bw(G) \geq 2$:
\[ bw(G) \leq tw(G) + 1 \leq \left\lfloor \frac{3}{2} bw(G) \right\rfloor \]
\end{proposition}

To perform dynamic programming on a branch decomposition $T$, we need $T$ to be rooted.
To create a root, we choose any edge $e \in E(T)$ and subdivide it, creating edges $e_1$ and $e_2$ and a new node $y$.
Next, we create another new node $z$, which will be our root, and add it together with the new edge $\{y,z\}$ to $T$.
The middle sets associated with the edges created by the subdivision are set to $X_e$, i.e., $X_{e_1} = X_{e_2} = X_e$.
Furthermore, the middle set of the new edge $\{y,z\}$ is the empty set: $X_{\{y,z\}} = \emptyset$.

We use the following terminology on the edges in a branch decomposition $T$ giving similar names to edges as we would usually do to vertices.
We call any edge of $T$ that is incident to a leaf but not the root a \emph{leaf edge}.
Any other edge is called an \emph{internal edge}.
Let $x$ be the lower endpoint of an internal edge $e$ of $T$ and let $l$, $r$ be the other two edges incident to $x$.
We call the edges $l$ and $r$ the \emph{child edges} of $e$.

\begin{definition}[partitioning of middle sets] \label{def:middlesets}
For a branch decomposition $T$, let $e \in E(T)$ be an edge not incident to a leaf with left child $l \in E(T)$ and right child $r \in E(T)$.
We define the following partitioning of $X_e \cup X_l \cup X_r$:
\begin{enumerate}
\item The \emph{intersection vertices}: $I = X_e \cap X_l \cap X_r$.
\item The \emph{forget vertices}: $F = (X_l \cap X_r) \setminus I$.
\item The \emph{vertices passed from the left}: $L = (X_e \cap X_l) \setminus I$.
\item The \emph{vertices passed from the right}: $R = (X_e \cap X_r) \setminus I$.
\end{enumerate}
\end{definition}
Notice that this is a partitioning because any vertex in at least one of the sets $X_e$, $X_l$, $X_r$ must be in at least two of them by definition of a middle set.

Because each bag has size at most $k$, the partitioning satisfies the properties:
\[ |I| + |L| + |R| \leq k \qquad\qquad |I| + |L| + |F| \leq k \qquad\qquad |I| + |R| + |F| \leq k \]

We associate with each edge $e \in E(T)$ of a branch decomposition $T$ the induced subgraph $G_e = G[V_e]$ of $G$.
A vertex $v \in V$ belongs to $V_e$ in this definition if and only if there is a middle set $f$ with $f=e$ or $f$ below $e$ in $T$ with $v \in X_f$.
That is, $v$ is in $V_e$ if and only if $v$ is an endpoint of an edge associated with a leaf of $T$ that is below $e$ in $T$, i.e.:
\[ G_e = G\!\!\left[\bigcup \{ X_f \;|\; \textrm{$f = e$ or $f$ is below $e$ in $T$} \} \right] \]

For an overview of branch decomposition based techniques, see~\cite{HicksKK05}.

\subsubsection{$k$-Expressions and Cliquewidth} \label{sec:defcw}
Another notion related to the decomposition of graphs is cliquewidth, introduced by Courcelle et al.~\cite{CourcelleER93}.

\begin{definition}[$k$-expression]
A \emph{$k$-expression} is an expression combining any number of the following four operations on labelled graphs with labels $\{1,2,\ldots,k\}$:
\begin{enumerate}
\item \emph{create a new graph}: create a new graph with one vertex having any label,
\item \emph{relabel}: relabel all vertices with label $i$ to $j$ ($i \not= j$),
\item \emph{add edges}: connect all vertices with label $i$ to all vertices with label $j$ ($i \not= j$),
\item \emph{join graphs}: take the disjoint union of two labelled graphs.
\end{enumerate}
\end{definition}

The \emph{cliquewidth} $cw(G)$ of a graph $G$ is defined to be the minimum $k$ for which there exists a $k$-expression that evaluates to a graph isomorphic to $G$.

The definition of a $k$-expression can also be turned into a rooted decomposition tree.
In this decomposition tree $T$, leafs of the tree $T$ correspond to the operations that create new graphs, effectively creating the vertices of $G$, and internal vertices of $T$ correspond to one of the other three above operations described above.
We call this tree a \emph{clique decomposition} of width $k$.
In this paper, we always assume that a given decomposition of the appropriate width is given.

In this paper, we also ssume that any $k$-expression does not contain superfluous operations, e.g., a $k$-expression does apply the operation to add edges between vertices with labels $i$ and $j$ twice in a row without first changing the sets of vertices with the labels $i$ and $j$, and it does not relabel vertices with a given label or add edges between vertices with a given label if the set of vertices with such a label is empty.
Under these conditions, it is not hard to see that any $k$-expressions consists of at most $\bigO(n)$ join operations and $\bigO(nk^2)$ other operations.

More information on solving problems on graphs of bounded cliquewidth can be found in \cite{CourcelleMR00}.

\subsection{Fast Algorithms to Speed Up Dynamic Programming} \label{sec:matrixmultiplic} \label{sec:fastsubsetconv}
In this paper, we will use fast algorithms for two standard problems as subroutines to speed up dynamic programming.
These are fast multiplication of matrices, and fast subset convolution.

\paragraph{Fast Matrix Multiplication.}
In this paper, we let $\omega$ be the smallest constant such that two $n \times n$ matrices can be multiplied in $\bigO(n^\omega)$ time; that is, $\omega$ is the matrix multiplication constant.
Currently, $\omega < 2.376$ due to the algorithm by Coppersmith and Winograd \cite{CoppersmithW90}.

For multiplying an $(n \times p)$ matrix $A$ and a $(p \times n)$ matrix $B$, we differentiate between $p \leq n$ and $p > n$.
Under the assumption that $\omega = 2.376$, an $\bigO(n^{1.85}p^{0.54})$ time algorithm is known if $p \leq n$ \cite{CoppersmithW90}.
Otherwise, the matrices can be multiplied in $\bigO(\frac{p}{n}n^\omega) = \bigO(pn^{\omega-1})$ time by matrix splitting: split the matrices $A$ and $B$ into $\frac{p}{n}$ many $n \times n$ matrices $A_1,\ldots A_\frac{p}{n}$ and $B_1,\ldots B_\frac{p}{n}$, multiply each of the $\frac{p}{n}$ pairs $A_i \times B_i$, and sum up the results.

\paragraph{Fast Subset Convolution.} 
Given a set $U$ and two functions $f,g: 2^{U} \rightarrow \Z$, their \emph{subset convolution} $(f * g)$ is defined as follows:
\[ (f * g)(S) = \sum_{X \subseteq S} f(X) g(S \setminus X) \]
\noindent The fast subset convolution algorithm by Bj\"orklund et al.~can compute this convolution using $\bigO(k^2 2^k)$ arithmetic operations \cite{BjorklundHKK07}.

Similarly, Bj\"orklund et al.~define the \emph{covering product} $(f *_c g)$ and the \emph{packing product} $(f *_p g)$ of $f$ and $g$ in the following way:
\[ (f *_c g)(S) = \mathop{\sum_{X, Y \subseteq S}}_{X \cup Y = S} f(X) g(Y)  \qquad \qquad
   (f *_p g)(S) = \mathop{\sum_{X, Y \subseteq S}}_{X \cap Y = \emptyset} f(X) g(Y) \]
\noindent These products can be computed using $\bigO(k 2^k)$ arithmetic operations \cite{BjorklundHKK07}.

In this paper, we will not directly use the algorithms of Bj\"orklund et al.~as subroutines.
Instead, we present their algorithms based on what we will call state changes.
The result is exactly the same as using the algorithms by Bj\"orklund et al.~as subroutines.
We choose to present our results in this way because it allows us to easily generalise the fast subset convolution algorithm to a more complex setting than functions with domain $2^U$ for some set~$U$.

\section{Dynamic Programming on Tree Decompositions} \label{sec:treewidth}
Algorithms solving \NP-hard problems in polynomial time on graphs of bounded treewidth are often dynamic programming algorithms of the following form.
The tree decomposition $T$ is traversed in a bottom-up manner.
For each node $x \in T$ visited, the algorithm constructs a table with partial solutions on the subgraph $G_x$, that is, the induced subgraph on all vertices that are in a bag $X_y$ where $y = x$ or $y$ is a descendant of $x$ in $T$.
Let an \emph{extension} of such a partial solution be a solution on $G$ that contains the partial solution on $G_x$, and let two such partial solutions $P_1$, $P_2$ have the same \emph{characteristic} if any extension of $P_1$ also is an extension of $P_2$ and vice versa.
The table for a node $x \in T$ does not store all possible partial solutions on $G_x$: it stores a set of solutions such that it contains exactly one partial solution for each possible characteristic.
While traversing the tree $T$, the table for a node $x \in T$ is computed using the tables that had been constructed for the children of $x$ in $T$.

This type of algorithm typically has a running time of the form $\bigO(f(k)poly(n))$ or even $\bigO(f(k)n)$, for an some function $f$ that grows at least exponentially.
This is because the size of the computed tables often is (at least) exponential in the treewidth $k$, but polynomial (or even constant) in the size of the graph $G$.
See Proposition~\ref{prop:simpletwdsalg} for an example algorithm.

In this section, we improve the exponential part of running time for many dynamic programming algorithms on tree decompositions for a large class of problems.
When the number of partial solutions of different characteristics stored in the table is $\bigOs(s^k)$, previous algorithms typically run in time $\bigOs(r^k)$ for some $r > s$.
This is because it is hard for these algorithms to compute a new table for a node in $T$ with multiple children.
In this case, the algorithm often needs to inspect exponentially many combinations of partial solutions from it children per entry of the new table.
We will show that algorithms with a running time of $\bigOs(s^k)$ exist for many problems.

This section is organised as follows.
We start by setting up the framework that we use for dynamic programming on tree decompositions by giving a simple algorithm in Section~\ref{sec:twintro}.
Here, we also define the de Fluiter property for treewidth.
Then, we give our results on {\sc Dominating Set} in Section~\ref{sec:dstw}, our results on counting perfect matchings in Section~\ref{sec:countpmtwalg}, our results on $[\rho,\sigma]$-domination problems in Section~\ref{sec:rhosigmatw}, and finally our results on the $\gamma$-clique covering, packing, and partitioning problems in Section~\ref{sec:cliquetwalg}.

\subsection{General Framework on Tree Decompositions} \label{sec:twintro}
We will now give a simple dynamic programming algorithm for the {\sc Dominating Set} problem.
This algorithm follows from standard techniques for treewidth-based algorithms, and we will give faster algorithms later.

\begin{proposition} \label{prop:simpletwdsalg}
There is an algorithm that, given a tree decomposition of a graph~$G$ of width~$k$, computes the size of a minimum dominating set in~$G$ in $\bigO(n 5^k i_+(\log(n)))$ time.
\end{proposition}
\begin{table}[tb]
	\begin{center}
	\begin{tabular}{r|l}
	state & meaning \\
	\hline \hline
	$1$   & this vertex is in the dominating set. \\
	$0_1$ & this vertex is not in the dominating set and has already been dominated. \\
	$0_0$ & this vertex is not in the dominating set and has not yet been dominated. \\
	$0_?$ & this vertex is not in the dominating set and may or may not be dominated. \\
	\end{tabular}
	\end{center}
	\caption{Vertex states for the {\sc Dominating Set} problem.}
	\label{tab:dsstates}
\end{table}
\begin{proof}
First, we construct a nice tree decomposition~$T$ of~$G$ of width~$k$ from the given tree decomposition in $\bigO(n)$ time.

Similar to Telle and Proskurowski~\cite{TelleP97}, we introduce vertex states $1$, $0_1$, and~$0_0$ that characterise the `state' of a vertex with respect to a vertex set~$D$ that is a partial solution of the {\sc Dominating Set} problem: $v$ has state~$1$ if $v \in D$; $v$ has state~$0_1$ if $v \not\in D$ but~$v$ is dominated by~$D$, i.e., there is a $d \in D$ with $\{v,d\} \in E$; and, $v$ has state~$0_0$ if $v \not\in D$ and~$v$ is not dominated by~$D$; see also Table~\ref{tab:dsstates}.

For each node~$x$ in the nice tree decomposition~$T$, we consider partial solutions $D \subseteq V_x$, such that all vertices in $V_x \setminus X_x$ are dominated by~$D$.
We characterise these sets~$D$ by the states of the vertices in~$X_x$ and the size of~$D$.
More precisely, we will compute a table~$A_x$ with an entry $A_x(c) \in \{0,1,\ldots,n\} \cup \{\infty\}$ for each $c \in \{1,0_1,0_0\}^{|X_x|}$.
We call $c \in \{1,0_1,0_0\}^{|X_x|}$ a \emph{colouring} of the vertices in~$X_x$.
A table entry $A_x(c)$ represents the size of the partial solution~$D$ of {\sc Dominating Set} in the induced subgraph~$G_x$ associated with the node~$x$ of~$T$ that satisfies the requirements defined by the states in the colouring~$c$, or infinity if no such set exists.
That is, the table entry gives the size of the smallest partial solution~$D$ in~$G_x$ that contains all vertices in~$X_x$ with state~$1$ in~$c$ and that dominates all vertices in~$G_x$ except those in~$X_x$ with state~$0_0$ in~$c$, or infinity if no such set exists.
Notice that these $3^{|X_x|}$ colourings correspond to $3^{|X_x|}$ partial solutions with different characteristics, and that it contains a partial solution for each possible characteristic.

We now show how to compute the table~$A_x$ for the next node $x \in T$ while traversing the nice tree decomposition~$T$ in a bottom-up manner.
Depending on the type of the node~$x$ (see Definition~\ref{def:nicetreedecomp}), we do the following:

\smallskip \noindent {\it Leaf node}:
Let~$x$ be a leaf node in~$T$.
The table consists of three entries, one for each possible colouring $c \in \{1,0_1,0_0\}$ of the single vertex~$v$ in~$X_x$.
\[ A_x(\{1\}) = 1 \qquad\qquad A_x(\{0_1\}) = \infty \qquad\qquad A_x(\{0_0\}) = 0\] 
Here, $A_x(c)$ corresponds to the size of the smallest partial solution satisfying the requirements defined by the colouring~$c$ on the single vertex~$v$.

\smallskip \noindent {\it Introduce node}:
Let~$x$ be an introduce node in~$T$ with child node~$y$.
We assume that when the $l$-th coordinate of a colouring of~$X_x$ represents a vertex~$u$, then the same coordinate of a colouring of~$X_y$ also represents~$u$, and that the last coordinate of a colouring of~$X_x$ represents the newly introduced vertex~$v$.
Now, for any colouring $c \in \{1,0_1,0_0\}^{|X_y|}$:
\begin{eqnarray*}
A_x(c \times \{0_1\}) & = & \left\{ \begin{array}{ll} A_y(c) & \textrm{if $v$ has a neighbour with state $1$ in $c$} \\ \infty & \textrm{otherwise} \end{array} \right. \\
A_x(c \times \{0_0\}) & = & \left\{ \begin{array}{ll} A_y(c) & \textrm{if $v$ has no neighbour with state $1$ in $c$} \\ \infty & \textrm{otherwise} \end{array} \right.
\end{eqnarray*}
For colourings with state~$1$ for the introduced vertex, we say that a colouring~$c_x$ of~$X_x$ \emph{matches} a colouring~$c_y$ of~$X_y$ if: 
\begin{itemize}
\item For all $u \in X_y \setminus N(v)$: $c_x(u) = c_y(u)$.
\item For all $u \in X_y \cap N(v)$: either $c_x(u)\!=\!c_y(u)\!=\!1$, or $c_x(u) \!=\! 0_1$ and $c_y(u) \!\in\! \{0_1,0_0\}$.
\end{itemize}
Here, $c(u)$ is the state of the vertex~$u$ in the colouring~$c$.
We compute $A_x(c)$ by the following formula:
\begin{eqnarray*}
A_x(c \times \{1\})	& \!=\! & \left\{ \begin{array}{ll} \infty & \textrm{if $c(u)=0_0$ for some $u \in N(v)$} \\
														1 + \min \{ A_y(c') \;|\; \textrm{$c'$ matches $c$} \} & \textrm{otherwise} \end{array} \right.
\end{eqnarray*}
It is not hard to see that $A_x(c)$ now corresponds to the size of the partial solution satisfying the requirements imposed on~$X_x$ by the colouring~$c$..

\smallskip \noindent {\it Forget node}:
Let~$x$ be a forget node in~$T$ with child node~$y$.
Again, we assume that when the $l$-th coordinate of a colouring of~$X_x$ represents a vertex~$u$, then the same coordinate of a colouring of~$X_y$ also represents~$u$, and that the last coordinate of a colouring of~$X_y$ represents vertex~$v$ that we are forgetting.
\[ A_x(c) = \min\{ A_y(c \times \{1\}), A_y(c \times \{0_1\}) \} \]
Now, $A_x(c)$ corresponds to the size of the smallest partial solution satisfying the requirements imposed on~$X_x$ by the colouring~$c$ as we consider only partial solutions that dominate the forgotten vertex.

\smallskip \noindent {\it Join node}:
Let~$x$ be a join node in~$T$ and let~$l$ and~$r$ be its child nodes.
As $X_x = X_l = X_r$, we can assume that the same coordinates represent the same vertices in a colouring of each of the three bags.

Let $c_x(v)$ be the state that represents the vertex~$v$ in colouring~$c_x$ of~$X_x$.
We say that three colourings~$c_x$, $c_l$, and~$c_r$ of~$X_x$, $X_l,$ and~$X_r$, respectively, \emph{match} if for each vertex $v \in X_x$:
\begin{itemize}
\item either $c_x(v) = c_l(v) = c_r(v) = 1$, 
\item or $c_x(v) = c_l(v) = c_r(v) = 0_0$, 
\item or $c_x(v) = 0_1$ while $c_l(v)$ and $c_r(v)$ are $0_1$ or $0_0$, but not both $0_0$.
\end{itemize}
Notice that three colourings~$c_x$, $c_l$, and~$c_r$ match if for each vertex~$v$ the requirements imposed by the states are correctly combined from the states in the colourings on both child bags~$c_l$ and~$c_r$ to the states in the colourings of the parent bag~$c_x$.
That is, if a vertex is required by~$c_x$ to be in the vertex set of a partial solution, then it is also required to be so in~$c_l$ and~$c_r$; if a vertex is required to be undominated in~$c_x$, then it is also required to be undominated in~$c_l$ and~$c_r$; and, if a vertex is required to be not in the partially constructed dominating set but it is required to be dominated in~$c_x$, then it is required not to be in the vertex sets of the partial solutions in both~$c_l$ and~$c_r$, but it must be dominated in one of both partial solutions.

The new table~$A_x$ can be computed by the following formula:
\[ A_x(c_x) = \min_{c_x, c_l, c_r \,\textrm{\scriptsize match}} A_l(c_l) + A_r(c_r) - \#_1(c_x) \]
Here, $\#_1(c)$ stands for the number of $1$-states in the colouring~$c$.
This number needs to be subtracted from the total size of the partial solution because the corresponding vertices are counted in each entry of $A_l(c_l)$ as well as in each entry of $A_r(c_r)$.
One can easily check that this gives a correct computation of~$A_x$.

\smallskip
After traversing the nice tree decomposition~$T$, we end up in the root node $z \in T$.
As $X_z = \emptyset$ and thus $G_z = G$, we find the size of the minimum dominating set in~$G$ in the single entry of~$A_z$.

\smallskip
It is not hard to see that the algorithm stores the size of the smallest partial solution of {\sc Dominating Set} in~$A_x$ for each possible characteristic on~$X_x$ for every node $x \in T$.
Hence, the algorithm is correct.

For the running time, observe that, for a leaf or forget node, $\bigO(3^{|X_x|}i_+(\log(n)))$ time is required since we work with $\log(n)$-bit numbers.
In an introduce node, we need more time as we need to inspect multiple entries from $A_y$ to compute an entry of $A_x$.
For a vertex~$u$ outside $N(v)$, we have three possible combinations of states, and for a vertex $u \in N(v)$ we have four possible combinations we need to inspect: the table entry with $c_x(u) = c_y(u) = 0_0$, colourings with $c_x(u) = c_y(u) = 1$, and colourings with $c_x(u) = 0_1$ while $c_y(u) = 0_0$ or $c_y(u) = 0_1$.
This leads to a total time of $\bigO(4^{|X_x|}i_+(\log(n)))$ for an introduce node.
In a join node, five combinations of states need to be inspected per vertex requiring $\bigO(5^{|X_x|}i_+(\log(n)))$ time in total.
As the largest bag has size at most $k+1$ and the tree decomposition~$T$ has $\bigO(n)$ nodes, the running time is $\bigO(n 5^k i_+(\log(n)))$.
\end{proof}

Many of the details of the algorithm described in the proof of Proposition~\ref{prop:simpletwdsalg} also apply to other algorithms described in this section.
We will not repeat these details: for the other algorithms we will only specify how to compute the tables $A_x$ for all four kinds of nodes.

We notice that the above algorithm computes only the size of a minimum dominating set in $G$, not the dominating set itself.
To construct a minimum dominating set $D$, the tree decomposition $T$ can be traversed in top-down order (reverse order compared to the algorithm of Proposition~\ref{prop:simpletwdsalg}).
We start by selecting the single entry in the table of the root node, and then, for each child node $y$ of the current node $x$, we select an the entry in $A_y$ which was used to compute the selected entry of $A_x$.
More specifically, we select the entry that was either used to copy into the selected entry of $A_x$, or we select one, or in a join node two, entries that lead to the minimum that was computed for $A_x$.
In this way, we trace back the computation path that computed the size of $D$.
During this process, we construct $D$ by adding each vertex that is not yet in $D$ and that has state $1$ in $c$ to $D$.
As we only use colourings that lead to a minimum dominating set, this process gives us a minimum dominating set in $G$.

Before we give a series of new, fast dynamic programming algorithms for a broad range of problems, we need the following definition.
We use it to improve the polynomial factors involved in the running times of the algorithms in this section.
\begin{definition}[de Fluiter property for treewidth] \label{def:fluiterproptw}
Given a graph optimisation problem~$\Pi$, consider a method to represent the different characteristics of partial solutions used in an algorithm that performs dynamic programming on tree decomposition to solve~$\Pi$.
Such a representation of partial solutions has the \emph{de Fluiter property for treewidth} if the difference between the objective values of any two partial solutions of $\Pi$ that are associated with a different characteristic and can both still be extended to an optimal solution is at most $f(k)$, for some non-negative function $f$ that depends only on the treewidth~$k$.
\end{definition}

This property is named after Babette van Antwerpen-de Fluiter, as this property implicitly play an important role in her work reported in \cite{BodlaenderA01,Fluiter97}.
Note that although we use the value $\infty$ in our dynamic programming tables, we do not consider such entries since they can never be extended to an optimal solution.
Hence, these entries do not influence the de Fluiter property.
Furthermore, we say that a problem has the \emph{linear de Fluiter property for treewidth} if $f$ is a linear function in~$k$.

Consider the representation used in Proposition~\ref{prop:simpletwdsalg} for the {\sc Dominating Set} problem.
This representation has the de Fluiter property for treewidth with $f(k) = k + 1$ because any table entry that is more than $k+1$ larger than the smallest value stored in the table cannot lead to an optimal solution.
This holds because any partial solution of {\sc Dominating Set}~$D$ that is more than $k+1$ larger than the smallest value stored in the table cannot be part of a minimum dominating set.
Namely, we can obtain a partial solution that is smaller than~$D$ and that dominates the same vertices or more by taking the partial solution corresponding to the smallest value stored in the table and adding all vertices in $X_x$ to it.

For a discussion of the de Fluiter properties and their relation to the related property \emph{finite integer index}, see Section~\ref{sec:fluiterprop}.

\subsection{Minimum Dominating Set} \label{sec:dstw}
Alber et al.~showed that one can improve the straightforward result of Proposition~\ref{prop:simpletwdsalg} by choosing a different set of states to represent characteristics of partial solutions~\cite{AlberBFKN02,AlberN02}: they obtained an $\bigOs(4^k)$-time algorithm using the set of states $\{1,0_1,0_?\}$ (see Table~\ref{tab:dsstates}).
We obtain an $\bigOs(3^k)$-time algorithm by using yet another set of states, namely $\{1,0_0,0_?\}$.

Note that $0_?$ represents a vertex $v$ that is not in the vertex set $D$ of a partial solution of {\sc Dominating Set}, while we do not specify whether $v$ is dominated; i.e., given~$D$, vertices with state $0_1$ and with state $0_0$ could also have state $0_?$.
In particular, there is no longer a unique colouring of $X_x$ with states for a specific partial solution: a partial solution can correspond to several such colourings.
Below, we discuss in detail how we can handle this situation and how it can lead to faster algorithms.

Since the state $0_0$ represents an undominated vertex and the state $0_?$ represents a vertex that may or may not be dominated, one may think that it is impossible to guarantee that a vertex is dominated using these states.
We circumvent this problem by not just computing the \emph{size} of a minimum dominating set, but by computing the \emph{number} of dominating sets of each fixed size~$\kappa$ with $0 \leq \kappa \leq n$.
This approach does not store (the size of) a solution per characteristic of the partial solutions, but counts the number of partial solutions of each possible size per characteristic.
We note that the algorithm of Proposition~\ref{prop:simpletwdsalg} can straightforwardly be modified to also count the number of (minimum) dominating sets.

For our next algorithm, we use dynamic programming tables in which an entry $A_x(c,\kappa)$ represents the number of partial solutions of {\sc Dominating Sets} on $G_x$ of size exactly $\kappa$ that satisfy the requirements defined by the states in the colouring $c$.
That is, the table entries give the number of partial solutions in $G_x$ of size $\kappa$ that dominate all vertices in $V_x \setminus X_x$ and all vertices in $X_x$ with state $0_1$, and that do not dominate the vertices in $X_x$ with state $0_0$.
This approach leads to the following result.

\begin{theorem} \label{thrm:countingtwdsalg}
There is an algorithm that, given a tree decomposition of a graph~$G$ of width~$k$, computes the number of dominating sets in~$G$ of each size~$\kappa$, $0 \leq \kappa \leq n$, in $\bigO(n^3 3^k i_\times(n))$ time.
\end{theorem}
\begin{proof}
We will show how to compute the table~$A_x$ for each type of node~$x$ in a nice tree decomposition~$T$.
Recall that an entry $A_x(c,\kappa)$ counts the number of partial solution of {\sc Dominating Set} of size exactly~$\kappa$ in~$G_x$ satisfying the requirements defined by the states in the colouring~$c$.

\smallskip \noindent {\it Leaf node}:
Let $x$ be a leaf node in $T$ with $X_x = \{v\}$.
We compute $A_x$ in the following way:
\begin{eqnarray*}
A_x(\{1\},\kappa) 	& = & \left\{ \begin{array}{ll} 1 & \textrm{if $\kappa = 1$} \\ 0 & \textrm{otherwise} \end{array} \right. \\
A_x(\{0_0\},\kappa) & = & \left\{ \begin{array}{ll} 1 & \textrm{if $\kappa = 0$} \\ 0 & \textrm{otherwise} \end{array} \right. \\
A_x(\{0_?\},\kappa) & = & \left\{ \begin{array}{ll} 1 & \textrm{if $\kappa = 0$} \\ 0 & \textrm{otherwise} \end{array} \right.
\end{eqnarray*}
Notice that this is correct since there is exactly one partial solution of size one that contains~$v$, namely $\{v\}$, and exactly one partial solution of size zero that does not contain~$v$, namely $\emptyset$.

\smallskip \noindent {\it Introduce node}:
Let~$x$ be an introduce node in~$T$ with child node~$y$ that introduces the vertex~$v$, and let $c \in \{1,0_1,0_0\}^{|X_y|}$.
We compute $A_x$ in the following way:
\begin{eqnarray*}
A_x(c \times \{1\}, \kappa) 	& = & \left\{ \begin{array}{ll} 0 & \textrm{if $v$ has a neighbour with state $0_0$ in $c$} \\
																		0 & \textrm{if $\kappa = 0$} \\
																		A_y(c, \kappa-1) & \textrm{otherwise} \end{array} \right. \\
A_x(c \times \{0_0\}, \kappa) & = & \left\{ \begin{array}{ll} 0 & \textrm{if $v$ has a neighbour with state $1$ in $c$} \\ 
																		A_y(c, \kappa) & \textrm{otherwise} \end{array} \right. \\
A_x(c \times \{0_?\}, \kappa) & = & A_y(c, \kappa)																		
\end{eqnarray*}
As the state~$0_?$ is indifferent about domination, we can copy the appropriate value from~$A_y$.
With the other two states, we have to set $A_x(c,\kappa)$ to zero if a vertex with state~$0_0$ can be dominated by a vertex with state~$1$.
Moreover, we have to update the size of the set if~$v$ gets state~$1$.

\smallskip \noindent {\it Forget node}:
Let~$x$ be a forget node in~$T$ with child node~$y$ that forgets the vertex~$v$.
We compute~$A_x$ in the following way:
\[ A_x(c, \kappa) = A_y(c \times \{1\}, \kappa) + A_y(c \times \{0_?\}, \kappa) -  A_y(c \times \{0_0\}, \kappa) \]
The number of partial solutions of size~$\kappa$ in~$G_x$ satisfying the requirements defined by~$c$ equals the number of partial solutions of size~$\kappa$ that contain~$v$ plus the number of partial solutions of size~$\kappa$ that do not contain~$v$ but where~$v$ is dominated.
This last number can be computed by subtracting the number of such solutions in which~$v$ is not dominated (state~$0_0$) from the total number of partial solutions in which~$v$ may be dominated or not (state~$0_?$).
This shows the correctness of the above formula.

The computation in the forget node is a simple illustration of the principle of inclusion/exclusion and the related M\"obius transform; see for example~\cite{BjorklundHK09}.

\smallskip \noindent {\it Join node}:
Let~$x$ be a join node in~$T$ and let~$l$ and~$r$ be its child nodes.
Recall that $X_x = X_l = X_r$.

If we are using the set of states $\{1,0_0,0_?\}$, then we do not have the consider colourings with matching states in order to compute the join.
Namely, we can compute~$A_x$ using the following formula:
\[ A_x(c,\kappa) = \sum_{\kappa_l + \kappa_r = \kappa + \#_1(c)} A_l(c,\kappa_l) \cdot A_r(c,\kappa_r) \]
The fact that this formula does not need to consider multiple matching colourings per colouring~$c$ (see Proposition~\ref{prop:simpletwdsalg}) is the main reason why the algorithm of this theorem is faster than previous results.

To see that the formula is correct, recall that any partial solution of {\sc Dominating Set} on~$G_x$ counted in the table~$A_x$ can be constructed from combining partial solutions~$G_l$ and~$G_r$ that are counted in~$A_l$ and~$A_r$, respectively.
Because an entry in~$A_x$ where a vertex~$v$ that has state~$1$ in a colouring of~$X_x$ counts partial solutions with~$v$ in the vertex set of the partial solution, this entry must count combinations of partial solutions in~$A_l$ and~$A_r$ where this vertex is also in the vertex set of these partial solutions and thus also has state~$1$.
Similarly, if a vertex~$v$ has state~$0_0$, we count partial solutions in which~$v$ is undominated; hence~$v$ must be undominated in both partial solutions we combine and also have state~$0_0$.
And, if a vertex~$v$ has state~$0_?$, we count partial solutions in which~$v$ is not in the vertex set of the partial solution and we are indifferent about domination; hence, we can get all combinations of partial solutions from~$G_l$ and~$G_r$ if we also are indifferent about domination in~$A_l$ and~$A_r$ which is represented by the state~$0_?$.
All in all, if we fix the sizes of the solutions from~$G_l$ and~$G_r$ that we use, then we only need to multiply the number of solutions from~$A_r$ and~$A_l$ of this size which have the same colouring on~$X_x$.
The formula is correct as it combines all possible combinations by summing over all possible sizes of solutions on~$G_l$ and~$G_r$ that lead to a solution on~$G_x$ of size~$\kappa$.
Notice that the term $\#_1(c)$ under the summation sign corrects the double counting of the vertices with state~$1$ in~$c$.

\smallskip
After the execution of this algorithm, the number of dominating sets of~$G$ of size~$\kappa$ can be found in the table entry $A_z(\emptyset,\kappa)$, where~$z$ is the root of~$T$.

\smallskip
For the running time, we observe that in a leaf, introduce, or forget node~$x$, the time required to compute~$A_x$ is linear in the size of the table~$A_x$.
The computations involve $n$-bit numbers because there can be up to~$2^n$ dominating sets in~$G$.
Since $c \in \{1,0_0,0_?\}^{|X_x|}$ and $0 \leq \kappa \leq n$, we can compute each table~$A_x$ in $\bigO(n3^k i_+(n))$ time.
In a join node~$x$, we have to perform $\bigO(n)$ multiplications to compute an entry of~$A_x$.
This gives a total of $\bigO(n^2 3^k i_\times(n))$ time per join node.
As the nice tree decomposition has $\bigO(n)$ nodes, the total running time is $\bigO(n^3 3^k i_\times(n))$.
\end{proof}

The algorithm of Theorem~\ref{thrm:countingtwdsalg} is exponentially faster in the treewidth~$k$ compared to the previous fastest algorithm of Alber et al.~\cite{AlberBFKN02,AlberN02}.
Also, no exponentially faster algorithm exists unless the Strong Exponential-Time Hypothesis fails~\cite{LokshtanovMS10}.
The exponential speed-up comes from the fact that we use a different set of states to represent the characteristics of partial solutions: a set of states that allows us to perform the computations in a join node much faster.
We note that although the algorithm of Theorem~\ref{thrm:countingtwdsalg} uses underlying ideas of the covering product of~\cite{BjorklundHKK07}, no transformations associated with such an algorithm are used directly.

To represent the characteristics of the partial solutions of the {\sc Dominating Set} problem, we can use any of the following three sets of states: $\{1,0_1,0_0\}$, $\{1,0_1,0_?\}$, $\{1,0_0,0_?\}$.
Depending on which set we choose, the number of combinations that we need to inspect in a join node differ.
We give an overview of this in Figure~\ref{fig:jointablesds}: each table represents a join using a different set of states, and each state in an entry of such a table represents a combination of the states in the left and right child nodes that need to be inspected to the create this new state.
The number of non-empty entries now shows how many combinations have to be considered per vertex in a bag of a join node.
Therefore, one can easily see that a table in a join node can be computed in $\bigOs(5^k)$, $\bigOs(4^k)$, and $\bigOs(3^k)$ time, respectively, depending on the set of states used.
These tables correspond to the algorithm of Proposition~\ref{prop:simpletwdsalg}, the algorithm of Alber et al.~\cite{AlberBFKN02,AlberN02}, and the algorithm of Theorem~\ref{thrm:countingtwdsalg}, respectively.

The way in which we obtain the third table in Figure~\ref{fig:jointablesds} from the first one reminds us of Strassen's algorithm for matrix multiplication~\cite{Strassen69}: the speed-up in this algorithm comes from the fact that one multiplication can be omitted by using a series of extra additions and subtractions.
Here, we do something similar by adding up all entries with a $0_1$-state or $0_0$-state together in the $0_?$-state and computing the whole block of four combinations at once.
We then reconstruct the values we need by subtracting to combinations with two $0_0$-states.

\begin{figure}[tb]
	\begin{center}
	\begin{multicols}{3}
	\begin{tabular}{c||c|c|c|} $\times$ & $1$ & $0_1$ & $0_0$ \\ \hline \hline $1$ & $1$ & & \\ \hline $0_1$ & & $0_1$ & $0_1$ \\ \hline $0_0$ & & $0_1$ & $0_0$ \\ \hline
	\end{tabular}
	\begin{tabular}{c||c|c|c|} $\times$ & $1$ & $0_1$ & $0_?$ \\ \hline \hline $1$ & $1$ & & \\ \hline $0_1$ & &  & $0_1$ \\ \hline $0_?$ & & $0_1$ & $0_?$ \\ \hline
	\end{tabular}
	\begin{tabular}{c||c|c|c|} $\times$ & $1$ & $0_?$ & $0_0$ \\ \hline \hline $1$ & $1$ & & \\ \hline $0_?$ & & $0_?$ &  \\ \hline $0_0$ & & & $0_0$ \\ \hline
	\end{tabular}
	\end{multicols}
	\end{center}
	\caption{Join tables for the {\sc Dominating Set} problem. From left to right they correspond to Proposition~\ref{prop:simpletwdsalg}, the algorithm from \cite{AlberBFKN02,AlberN02}, and Theorem~\ref{thrm:countingtwdsalg}.}
	\label{fig:jointablesds}
\end{figure}

The exponential speed-up obtained by the algorithm of Theorem~\ref{thrm:countingtwdsalg} comes at the cost of extra polynomial factors in the running time.
This is $n^2$ times the factor due to the fact that we work with $n$-bit numbers.
Since we compute the number of dominating sets of each size~$\kappa$, $0 \leq \kappa \leq n$, instead of computing a minimum dominating set, some extra polynomial factors in $n$ seem unavoidable.
However, the ideas of Theorem~\ref{thrm:countingtwdsalg} can also be used to count only \emph{minimum} dominating sets
Using that {\sc Dominating Set} has the de Fluiter property for treewidth, this leads to the following result, where the factor $n^2$ is replaced by the much smaller factor $k^2$.

\begin{corollary} \label{cor:countmdstwalg}
There is an algorithm that, given a tree decomposition of a graph~$G$ of width~$k$, computes the number of minimum dominating sets in~$G$ in $\bigO(n k^2 3^k i_\times(n))$ time.
\end{corollary}
\begin{proof}
We notice that the representation of the different characteristics of partial solutions used in Theorem~\ref{thrm:countingtwdsalg} has the linear de Fluiter property when used to count the number of minimum dominating sets.
More explicitly, when counting the number of minimum dominating sets, we need to store only the number of partial solutions of each different characteristic that are at most $k+1$ larger in size than the smallest partial solution with a non-zero entry.
This holds, as larger partial solutions can never lead to a minimum dominating set since taking any set corresponding to this smallest non-zero entry and adding all vertices in~$X_x$ leads to a smaller partial solution that dominates at least the same vertices.

In this way, we can modify the algorithm of Theorem~\ref{thrm:countingtwdsalg} such that, in each node $x \in T$, we store a number~$\xi_x$ representing the size of the smallest partial solution and a table~$A_x$ with the number of partial solutions $A_x(c,\kappa)$ with $\xi_x \leq \kappa \leq \xi_x+k+1$.

In a leaf node~$x$, we simply set $\xi_x = 0$.
In an introduce or forget node~$x$ with child node~$y$, we first compute the entries $A_x(c,\kappa)$ for $\xi_y \leq \kappa \leq \xi_y + k +1$ and then set~$\xi_x$ to the value of~$\kappa$ corresponding to the smallest non-zero entry of~$A_x$.
While computing~$A_x$, the algorithm uses $A_y(c,\kappa) = 0$ for any entry $A_y(c,\kappa)$ that falls outside the given range of~$\kappa$.
Finally, in a join node~$x$ with child nodes~$r$ and~$l$, we do the same as in Theorem~\ref{thrm:countingtwdsalg}, but we compute only the entries with $\kappa$ in the range $\xi_l + \xi_r - (k+1) \leq \kappa \leq \xi_l + \xi_r + (k+1)$.
Furthermore, as all terms of the sum with~$\kappa_l$ or~$\kappa_r$ outside the range of~$A_l$ and~$A_r$ evaluate to zero, we now have to evaluate only $\bigO(k)$ terms of the sum.
It is not hard to see that all relevant combinations of partial solutions from the two child nodes~$l$ and~$r$ fall in this range of~$\kappa$.

The modified algorithm computes $\bigO(n)$ tables of size $\bigO(k3^k)$, and the computation of each entry requires at most $\bigO(k)$ multiplications of $n$-bit numbers.
Therefore, the running time is $\bigO(n k^2 3^k i_\times(n))$.
\end{proof}

A disadvantage of the direct use of the algorithm of Corollary~\ref{cor:countmdstwalg} compared to Proposition~\ref{prop:simpletwdsalg} is that we cannot reconstruct a minimum dominating set in $G$ by directly tracing back the computation path that gave us the size of a minimum domination set.
However, as we show below, we can transform the tables computed by Theorem~\ref{thrm:countingtwdsalg} and Corollary~\ref{cor:countmdstwalg} that use the states $\{1,0_0,0_?\}$ in $\bigOs(3^k)$ time into tables using any of the other sets of states.
These transformations have two applications.
First of all, they allow us to easily construct a minimum dominating set in $G$ from the computation of Corollary~\ref{cor:countmdstwalg} by transforming the computed tables into different tables as used in Proposition~\ref{prop:simpletwdsalg} and thereafter traverse the tree in a top-down fashion as we have discussed earlier.
Secondly, they can be used to switch from using $n$-bit numbers to $\bigO(k)$-bit numbers, further improving the polynomial factors of the running time if we are interested only in solving the {\sc Dominating Set} problem.

\begin{lemma} \label{lem:dsstates}
Let~$x$ be a node of a tree decomposition~$T$ and let~$A_x$ be a table with entries $A_x(c,\kappa)$ representing the number of partial solutions of {\sc Dominating Set} of~$G_x$ of each size~$\kappa$, for some range of~$\kappa$, corresponding to each colouring~$c$ of the bag~$X_x$ with states from one of the following sets:
\[ \{1,0_1,0_0\} \qquad \qquad \{1,0_1,0_?\} \qquad \qquad \{1,0_0,0_?\} \qquad \qquad \textrm{(see Table~\ref{tab:dsstates})} \]
The information represented in the table~$A_x$ does not depend on the choice of the set of states from the options given above.
Moreover, there exist transformations between tables using representations with different sets of states using $\bigO(|X_x||A_x|)$ arithmetic operations.
\end{lemma} 
\begin{proof}
We will transform~$A_x$ such that it represents the same information using a different set of states.
The transformation will be given for fixed~$\kappa$ and can be repeated for each~$\kappa$ in the given range.

The transformations work in~$|X_x|$ steps.
In step~$i$, we assume that the first $i-1$ coordinates of the colouring~$c$ in our table~$A_x$ use the initial set of states, and the last $|X_x|-i$ coordinates use the set of states to which we want to transform.
Using this as an invariant, we change the set of states used for the $i$-th coordinate at step~$i$.

Transforming from $\{1,0_1,0_0\}$ to $\{1,0_0,0_?\}$ can be done using the following formula in which $A_x(c,\kappa)$ represents our table for colouring~$c$, $c_1$ is a subcolouring of size $i-1$ using states $\{1,0_1,0_0\}$, and~$c_2$ is a subcolouring of size $|X_x|-i$ using states $\{1,0_0,0_?\}$.
\[ A_x(c_1 \times \{0_?\} \times c_2, \kappa) = A_x(c_1 \times \{0_1\} \times c_2, \kappa) + A_x(c_1 \times \{0_0\} \times c_2, \kappa) \]
We keep entries with states~$1$ and~$0_0$ on the $i$-th vertex the same, and we remove entries with state~$0_1$ on the $i$-th vertex after computing the new value.
In words, the above formula counts the number partial solutions that do not containing the $i$-th vertex~$v$ in their vertex sets by adding the number of partial solutions that do not contain~$v$ in their vertex sets and dominate it to the number of partial solutions that do not contain~$v$ in the vertex sets and do not dominate it.
This completes the description of the transformation.

To see that the new table contains the same information, we can apply the reverse transformation from the set of states $\{1,0_0,0_?\}$ to the set $\{1,0_1,0_0\}$ by using the same transformation with a different formula to introduce the new state:
\[ A_x(c_1 \times \{0_1\} \times c_2, \kappa) = A_x(c_1 \times \{0_?\} \times c_2, \kappa) - A_x(c_1 \times \{0_0\} \times c_2, \kappa) \]
A similar argument applies here: the number of partial solutions that dominate but do not contain the $i$-th vertex~$v$ in their vertex sets equals the total number of partial solutions that do not contain~$v$ in their vertex sets minus the number of partial solutions in which~$v$ is undominated.

The other four transformations work similarly.
Each transformation keeps the entries of one of the three states $0_1$, $0_0$, and $0_?$ intact, computes the entries for the new state by a coordinate-wise addition or subtraction of the other two states, and removes the entries using the third state from the table.
To compute an entry with the new state, either the above two formula can be used if the new state is~$0_1$ or~$0_?$, or the following formula can be used if the new state is~$0_0$:
\[ A_x(c_1 \times \{0_0\} \times c_2, \kappa) = A_x(c_1 \times \{0_?\} \times c_2, \kappa) - A_x(c_1 \times \{0_1\} \times c_2, \kappa) \]

For the above transformations, we need~$|X_x|$ additions or subtractions for each of the $|A_x|$ table entries.
Hence, a transformation requires $\bigO(|X_x||A_x|)$ arithmetic operations.
\end{proof}

We are now ready to give our final improvement for {\sc Dominating Set}.

\begin{corollary} \label{cor:solvedstwalg}
There is an algorithm that, given a tree decomposition of a graph~$G$ of width~$k$, computes the size of a minimum dominating set in~$G$ in $\bigO(n k^2 3^k)$ time.
\end{corollary}

We could give a slightly shorter proof than the one given below.
This proof would directly combine the algorithm of Proposition~\ref{prop:simpletwdsalg} with the ideas of Theorem~\ref{thrm:countingtwdsalg} using the transformations from Lemma~\ref{lem:dsstates}.
However, combining our ideas with the computations in the introduce and forget nodes in the algorithm of Alber et al.~\cite{AlberBFKN02,AlberN02} gives a more elegant solution, which we prefer to present.

\begin{proof}
On leaf, introduce, and forget nodes, our algorithm is exactly the same as the algorithm of Alber et al.~\cite{AlberBFKN02,AlberN02}, while on a join node it is similar to Corollary~\ref{cor:countmdstwalg}.
We give the full algorithm for completeness.

For each node $x \in T$, we compute a table~$A_x$ with entries $A_x(c)$ containing the size of a smallest partial solution of {\sc Dominating Set} that satisfies the requirements defined by the colouring~$c$ using the set of states $\{1,0_1,0_?\}$.

\smallskip \noindent {\it Leaf node}:
Let~$x$ be a leaf node in~$T$. We compute~$A_x$ in the following way:
\[ A_x(\{1\}) = 1 \qquad \qquad A_x(\{0_1\}) = \infty \qquad \qquad A_x(\{0_?\}) = 0 \]

\smallskip \noindent {\it Introduce node}:
Let~$x$ be an introduce node in~$T$ with child node~$y$ introducing the vertex~$v$.
We compute~$A_x$ in the following way:
\begin{eqnarray*}
A_x(c \times \{0_1\}) & = & \left\{ \begin{array}{ll} A_y(c) & \textrm{if $v$ has a neighbour with state 1 in $c$} \\ \infty & \textrm{otherwise} \end{array} \right. \\
A_x(c \times \{0_?\}) & = & A_y(c) \\
A_x(c \times \{1\}) & = & 1 + A_y(\phi_{N(v):0_1 \rightarrow 0_?}(c))
\end{eqnarray*}
Here, $\phi_{N(v):0_1 \rightarrow 0_?}(c)$ is the colouring~$c$ with every occurrence of the state~$0_1$ on a vertex in $N(v)$ replaced by the state~$0_?$.

\smallskip \noindent {\it Forget node}:
Let~$x$ be a forget node in~$T$ with child node~$y$ forgetting the vertex~$v$.
We compute~$A_x$ in the following way:
\[ A_x(c) = \min\{ A_y(c \times \{1\}), A_y(c \times \{0_1\}) \} \]
Correctness of the operations on a leaf, introduce, and forget node are easy to verify and follow from~\cite{AlberBFKN02,AlberN02}.

\smallskip \noindent {\it Join node}:
Let~$x$ be a join node in~$T$ and let~$l$ and~$r$ be its child nodes.
We first create two tables~$A'_l$ and~$A'_r$.
For $y \in \{l,r\}$, we let $\xi_y = \min\left\{ A_y(c') \;|\; c' \in \{1,0_1,0_?\}^{|X_y|} \right\}$ and let~$A'_y$ have entries $A'_y(c,\kappa)$ for all $c \in \{1,0_1,0_?\}^{|X_y|}$ and $\kappa$ with $\xi_y \leq \kappa \leq \xi_y + k + 1$:
\[ A'_y(c,\kappa) = \left\{ \begin{array}{ll} 1 & \textrm{if $A_y(c)=\kappa$} \\ 0 & \textrm{otherwise} \end{array} \right. \]
After creating the tables~$A'_l$ and~$A'_r$, we use Lemma~\ref{lem:dsstates} to transform the tables~$A'_l$ and~$A'_r$ such that they use colourings~$c$ with states from the set $\{1,0_0,0_?\}$.
The initial tables~$A'_y$ do not contain the actual number of partial solutions; they contain a $1$-entry if a corresponding partial solution exists.
In this case, the tables obtained after the transformation count the number $1$-entries in the tables before the transformation.
In the table~$A'_x$ computed for the join node~$x$, we now count the number of combinations of these $1$-entries.
This suffices since any smallest partial solution in~$G_x$ that is obtained by joining partial solutions from both child nodes must consist of minimum solutions in~$G_l$ and~$G_r$. 

We can compute~$A'_x$ by evaluating the formula for the join node in Theorem~\ref{thrm:countingtwdsalg} for all~$\kappa$ with $\xi_l + \xi_r - (k+1) \leq \kappa \leq \xi_l + \xi_r + (k+1)$ using the tables~$A'_l$ and~$A'_r$.
If we do this in the same way as in Corollary~\ref{cor:countmdstwalg}, then we consider only the $\bigO(k)$ terms of the formula where~$\kappa_l$ and~$\kappa_r$ fall in the specified ranges for~$A_l$ and~$A_r$, respectively, as other terms evaluate to zero.
In this way, we obtain the table~$A'_x$ in which entries are marked by colourings with states from the set $\{1,0_0,0_?\}$.
Finally, we use Lemma~\ref{lem:dsstates} to transform the table~$A'_x$ such that it again uses colourings with states from the set $\{1,0_1,0_?\}$.
This final table gives the number of combinations of $1$-entries in~$A_l$ and~$A_r$ that lead to partial solutions of each size that satisfy the associated colourings.
Since we are interested only in the size of the smallest partial solution of {\sc Dominating Set} of each characteristic, we can extract these values in the following way:
\[ A_x(c) = \min\{ \kappa \;|\; A'_x(c,\kappa) \geq 1; \; \xi_l + \xi_r - (k+1) \leq \kappa \leq \xi_l + \xi_r + (k+1) \} \]

\smallskip
For the running time, we first consider the computations in a join node.
Here, each state transformation requires $\bigO(k^23^k)$ operations by Lemma~\ref{lem:dsstates} since the tables have size $\bigO(k3^k)$.
These operations involve $\bigO(k)$-bit numbers since the number of $1$-entries in $A_l$ and $A_r$ is at most $3^{k+1}$.
Evaluating the formula that computes $A'_x$ from the tables $A'_l$ and $A'_r$ costs $\bigO(k^23^k)$ multiplications.
If we do not store a $\log(n)$-bit number for each entry in the tables $A_x$ in any of the four kinds of nodes of $T$, but store only the smallest entry using a $\log(n)$-bit number and let $A'_x$ contain the difference to this smallest entry, then all entries in any of the $A'_x$ can also be represented using $\bigO(k)$-bit numbers.
Since there are $\bigO(n)$ nodes in~$T$, this gives a running time of $\bigO(n k^2 3^k)$.
Note that the time required to multiply the $\bigO(k)$-bit numbers disappears in the computational model with $\bigO(k)$-bit word size that we use.
\end{proof}

Corollary~\ref{cor:solvedstwalg} gives the currently fastest algorithm for {\sc Dominating Set} on graphs given with a tree decomposition of width~$k$.
Essentially, what the algorithm does is fixing the $1$-states and applying the covering product of Bj\"orklund et al.~\cite{BjorklundHKK07} on the~$0_1$-states and~$0_?$-states, where the~$0_1$-states need to be covered by the same states from both child nodes.
We chose to present our algorithm in a way that does not use the covering product directly, because reasoning with states allows us to generalise our results in Section~\ref{sec:rhosigmatw}.

We conclude by stating that we can directly obtain similar results for similar problems using exactly the same techniques:
\begin{proposition}
For each of the following problems, there is an algorithm that solves them, given a tree decomposition of a graph~$G$ of width~$k$, using the following running times:
\begin{itemize}
\item {\sc Independent Dominating Set} in $\bigO(n^3 3^k)$ time,
\item {\sc Total Dominating Set} in $\bigO(n k^2 4^k)$ time,
\item {\sc Red-Blue Dominating Set} in $\bigO(n k^2 2^k)$ time,
\item {\sc Partition Into Two Total Dominating Sets} in $\bigO(n 6^k)$ time.
\end{itemize}
\end{proposition}
\begin{proof}[(Sketch)]
Use the same techniques as in the rest of this subsection.
We emphasise only the following details.

With {\sc Independent Dominating Set}, the factor $n^3$ comes from the fact that this (minimisation) problem does not have the de Fluiter property for treewidth.
However, we can still use $\bigO(k)$-bit numbers.
This holds because, even though the expanded tables $A'_l$ and $A'_r$ have size at most $n3^k$, they still contain the value one only once for each of the $3^k$ characteristics before applying the state changes.
Therefore, the total sum of the values in the table, and thus also the maximum value of an entry in these tables after the state transformations is $3^k$; these can be represented by $\bigO(k)$-bit numbers.

With {\sc Total Dominating Set}, the running time is linear in $n$ while the extra polynomial factor is $k^2$.
This is because this problem does have the linear de Fluiter property for treewidth.

With {\sc Red-Blue Dominating Set}, an exponential factor of $2^k$ suffices as we can use two states for the red vertices (in the red-blue dominating set or not) and two different states for the blue vertices (dominated or not).

With {\sc Partition Into Two Total Dominating Sets}, we note that we can restrict ourselves to using six states when we modify the tree decomposition such that every vertex always has at least one neighbour and hence is always dominated by at least one of the two partitions.
Furthermore, the polynomial factors are smaller because this is not an optimisation problem and we do not care about the sizes of both partitions.
\end{proof}

\subsection{Counting the Number of Perfect Matchings} \label{sec:countpmtwalg}
The next problem we consider is the problem of computing the number of perfect matchings in a graph.
We give an $\bigOs(2^k)$-time algorithm for this problem.
This requires a slightly more complicated approach than the approach of the previous section.
The main difference is that here every vertex needs to be matched \emph{exactly} once, while previously we needed to dominate every vertex \emph{at least} once.
After introducing state transformations similar to Lemma~\ref{lem:dsstates}, we will introduce some extra counting techniques to overcome this problem.

The obvious tree-decomposition-based dynamic programming algorithm uses the set of states $\{0,1\}$, where $1$ means this vertex is matched and $0$ means that it is not.
It then computes, for every node $x \in T$, a table $A_x$ with entries $A_x(c)$ containing the number of matchings in $G_x$ with the property that the only vertices that are not matched are exactly the vertices in the current bag $X_x$ with state $0$ in $c$.
This algorithm will run in $\bigOs(3^k)$ time; this running time can be derived from the join table in Figure~\ref{fig:jointablepm}.
Similar to Lemma~\ref{lem:dsstates} in the previous section, we will prove that the table $A_x$ contains exactly the same information independent of whether we use the set of states $\{0,1\}$ or $\{0,?\}$, where $?$ represents a vertex for which we do not specify whether it is matched or not.
I.e., for a colouring $c$, we count the number of matchings in $G_x$, where all vertices in $V_x \setminus X_x$ and all vertices in $X_x$ with state $1$ in $c$ are matched, all vertices in $X_x$ with state $0$ in $c$ are unmatched, and all vertices in $X_x$ with state~$?$ can either be matched or not.

\begin{figure}[tb]
	\begin{center}
	\begin{multicols}{2}
	\hspace{2cm}
	\begin{tabular}{c||c|c|} $\times$ & $0$ & $1$ \\ \hline \hline $0$ & $0$ & $1$ \\ \hline $1$ & $1$ & \\ \hline
	\end{tabular}
	
	\begin{tabular}{c||c|c|} $\times$ & $0$ & $?$ \\ \hline \hline $0$ & $0$ & \\ \hline $?$ & & $\not\,?$ \\ \hline
	\end{tabular}
	\hspace{2cm}
	\end{multicols}
	\end{center}
	\caption{Join tables for counting the number of perfect matchings. We used the symbol $\not\,?$ in the last table because the direct combination of two $?$-states can lead to matching a vertex twice.}
	\label{fig:jointablepm}
\end{figure}

\begin{lemma} \label{lem:pmstates}
Let~$x$ be a node of a tree decomposition~$T$ and let~$A_x$ be a table with entries $A_x(c)$ representing the number of matchings in~$G_x$ matching all vertices in $V_x \setminus X_x$ and corresponding to each colouring~$c$ of the bag~$X_x$ with states from one of the following sets:
\[ \{1,0\} \qquad \qquad \qquad \{1,?\} \qquad \qquad \qquad \{0,?\} \]
The information represented in the table~$A_x$ does not depend on the choice of the set of states from the options given above.
Moreover, there exist transformations between tables using representations with different sets of states using $\bigO(|X_x||A_x|)$ arithmetic operations.
\end{lemma} 
If one defines a vertex with state~$1$ or~$?$ to be in a set~$S$, and a vertex with state~$0$ not to be in~$S$, then the state changes essentially are M\"obius transforms and inversions, see~\cite{BjorklundHKK07}.
The transformations in the proof below essentially are the fast evaluation algorithms from~\cite{BjorklundHKK07}.

\begin{proof}
The transformations work almost identical to those in the proof of Lemma~\ref{lem:dsstates}.
In step $1 \leq i \leq |X_x|$, we assume that the first $i-1$ coordinates of the colouring~$c$ in our table use one set of states, and the last $|X_x|-i$ coordinates use the other set of states.
Using this as an invariant, we change the set of states used for the $i$-th coordinate at step~$i$.

Transforming from $\{0,1\}$ to $\{0,?\}$ or $\{1,?\}$ can be done using the following formula.
In this formula, $A_x(c)$ represents our table for colouring~$c$, $c_1$ is a subcolouring of size $i-1$ using states $\{0,1\}$, and~$c_2$ is a subcolouring of size $|X_x|-i$ using states $\{0,?\}$:
\[ A_x(c_1 \times \{?\} \times c_2) = A_x(c_1 \times \{0\} \times c_2) + A_x(c_1 \times \{1\} \times c_2) \]
In words, the number of matchings that may contain some vertex~$v$ equals the sum of the number of matchings that do and the number of matchings that do not contain~$v$.

The following two similar formulas can be used for the other four transformations:
\[ A_x(c_1 \times \{1\} \times c_2) = A_x(c_1 \times \{?\} \times c_2) - A_x(c_1 \times \{0\} \times c_2) \]
\[ A_x(c_1 \times \{0\} \times c_2) = A_x(c_1 \times \{?\} \times c_2) - A_x(c_1 \times \{1\} \times c_2) \]

In these transformations, we need~$|X_x|$ additions or subtractions for each of the $|A_x|$ table entries.
Hence, a transformation requires $\bigO(|X_x||A_x|)$ arithmetic operations.
\end{proof}

Although we can transform our dynamic programming tables such that they use different sets of states, this does not directly help us in obtaining a faster algorithm for counting the number of perfect matchings.
Namely, if we would combine two partial solutions in which a vertex~$v$ has the~$?$-state in a join node, then it is possible that~$v$ is matched twice in the combined solution: once in each child node.
This would lead to incorrect answers, and this is why we put a $\not\,?$ instead of a~$?$ in the join table in Figure~\ref{fig:jointablepm}.
We overcome this problem by using some additional counting tricks that can be found in the proof below.

\begin{theorem} \label{thrm:countingpmtwalg}
There is an algorithm that, given a tree decomposition of a graph~$G$ of width~$k$, computes the number of perfect matchings in $G$ in $\bigO(n k^2 2^k i_\times(k\log(n)))$ time.
\end{theorem}
\begin{proof}
For each node $x \in T$, we compute a table~$A_x$ with entries $A_x(c)$ containing the number of matchings that match all vertices in $V_x \setminus X_x$ and that satisfy the requirements defined by the colouring~$c$ using states $\{1,0\}$.
We use the extra invariant that vertices with state~$1$ are matched only with vertices outside the bag, i.e., vertices that have already been forgotten by the algorithm.
This prevents vertices being matched within the bag and greatly simplifies the presentation of the algorithm.

\smallskip \noindent {\it Leaf node}:
Let~$x$ be a leaf node in~$T$. We compute~$A_x$ in the following way:
\[ A_x(\{1\}) = 0 \qquad \qquad A_x(\{0\}) = 1 \]
The only matching in the single vertex graph is the empty matching.

\smallskip \noindent {\it Introduce node}:
Let~$x$ be an introduce node in~$T$ with child node~$y$ introducing the vertex~$v$.
The invariant on vertices with state~$1$ makes the introduce operation trivial:
\[ A_x(c \times \{1\}) = 0 \qquad \qquad  A_x(c \times \{0\}) = A_y(c) \]

\smallskip \noindent {\it Forget node}:
Let~$x$ be a forget node in~$T$ with child node~$y$ forgetting the vertex~$v$.
If the vertex $v$ is not matched already, then it must be matched to an available neighbour at this point:
\[ A_x(c) = A_y(c \times \{1\}) + \sum_{u \in N(v), c(u)=1} A_y(\phi_{u:1 \rightarrow 0}(c) \times \{0\}) \]
Here, $c(u)$ is the state of~$u$ in~$c$ and $\phi_{u:1 \rightarrow 0}(c)$ is the colouring~$c$ where the state of~$u$ is changed from~$1$ to~$0$.
This formula computes the number of matchings corresponding to~$c$, by adding the number of matchings in which~$v$ is matched already to the number of matchings of all possibly ways of matching~$v$ to one of its neighbours.
We note that, because of our extra invariant, we have to consider only neighbours in the current bag~$X_x$.
Namely, if we would match~$v$ to an already forgotten vertex~$u$, then we could have matched~$v$ to~$u$ in the node where~$u$ was forgotten.

\smallskip \noindent {\it Join node}:
Let~$x$ be a join node in~$T$ and let~$l$ and~$r$ be its child nodes.

The join is the most interesting operation.
As discussed before, we cannot simply change the set of states to $\{0,?\}$ and perform the join similar to {\sc Dominating Set} as suggested by Table~\ref{fig:jointablepm}.
We use the following method: we expand the tables and index them by the number of matched vertices in~$X_l$ or~$X_r$, i.e., the number of vertices with state~$1$.
Let $y \in \{l,r\}$, then we compute tables~$A'_l$ and~$A'_r$ as follows:
\[ A'_y(c,i) = \left\{ \begin{array}{ll} A_y(c) & \textrm{if $\#_1(c) = i$}\\ 0 & \textrm{otherwise} \end{array} \right.  \]
Next, we change the state representation in both tables~$A'_y$ to $\{0,\!?\}$ using Lemma~\ref{lem:pmstates}.
These tables do not use state~$1$, but are still indexed by the number of $1$-states used in the previous representation.
Then, we join the tables by combining all possibilities that arise from~$i$ $1$-states in the previous representation using states $\{0,1\}$ (stored in the index~$i$) using the following formula:
\[ A'_x(c,i) = \sum_{i_l + i_r = i} A'_l(c,i_l) \cdot A'_r(c,i_r) \]
As a result, the entries $A'_x(c,i)$ give us the total number of ways to combine partial solutions from~$G_l$ and~$G_r$ such that the vertices with state~$0$ in~$c$ are unmatched, the vertices with state~$?$ in~$c$ can be matched in zero, one, or both partial solutions used, and the total number of times the vertices with state~$?$ are matched is~$i$.

Next, we change the states in the table~$A'_x$ back to $\{0,1\}$ using Lemma~\ref{lem:pmstates}.
It is important to note that the $1$-state can now represent a vertex that is matched twice because the $?$-state used before this second transformation represented vertices that could be matched twice as well.
However, we can find those entries in which no vertex is matched twice by applying the following observation: the total number of $1$-states in~$c$ should equal the sum of those in its child tables, and this sum is stored in the index~$i$.
Therefore, we can extract the number of perfect matchings for each colouring~$c$ using the following formula:
\[ A_x(c) = A'_x(c,\#_1(c)) \]
In this way, the algorithm correctly computes the tables~$A_x$ for a join node $x \in T$.
This completes the description of the algorithm.

\smallskip
The computations in the join nodes again dominate the running time.
In a join node, the transformations of the states in the tables cost $\bigO(k^22^k)$ arithmetic operations each, and the computations of~$A'_x$ from~$A'_l$ and~$A'_r$ also costs $\bigO(k^22^k)$ arithmetic operations.
We will now show that these arithmetic operations can be implemented using $\bigO(k\log(n))$-bit numbers.
For every vertex, we can say that the vertex is matched to another vertex at the time when it is forgotten in~$T$, or when its matching neighbour is forgotten.
When it is matched at the time that it is forgotten, then it is matched to one of its at most $k+1$ neighbours.
This leads to at most $k+2$ choices per vertex.
As a result, there are at most $\bigO(k^n)$ perfect matchings in~$G$, and the described operations can be implemented using $\bigO(k\log(n))$-bit numbers.

Because a nice tree decomposition has $\bigO(n)$ nodes, the running time of the algorithm is $\bigO(n k^2 2^k i_\times(k\log(n)))$.
\end{proof}

The above theorem gives the currently fastest algorithm for counting the number of perfect matchings in graphs with a given tree decompositions of width $k$.
The algorithm uses ideas from the fast subset convolution algorithm of Bj\"orklund et al.~\cite{BjorklundHKK07} to perform the computations in the join node.

\subsection{$[\rho,\sigma]$-Domination Problems} \label{sec:rhosigmatw}
We have shown how to solve two elementary problems in $\bigOs(s^k)$ time on graphs of treewidth~$k$, where~$s$ is the number of states per vertex used in representations of partial solutions.
In this section, we generalise our result for {\sc Dominating Set} to the $[\rho,\sigma]$-domination problems.
We show that we can solve all $[\rho,\sigma]$-domination problems with finite or cofinite~$\rho$ and~$\sigma$ in $\bigOs(s^k)$ time.
This includes the existence (decision), minimisation, maximisation, and counting variants of these problems.

For the $[\rho,\sigma]$-domination problems, one can also use colourings with states to represent the different characteristics of partial solutions.
Let~$D$ be the vertex set of a partial solution of a $[\rho,\sigma]$-domination problem.
One set of states that we use involves the states~$\rho_j$ and~$\sigma_j$, where~$\rho_j$ and~$\sigma_j$ represent vertices not in~$D$, or in~$D$, that have~$j$ neighbours in~$D$, respectively.
For finite $\rho$, $\sigma$, we let $p = \max \{\rho\}$ and $q = \max \{\sigma\}$.
In this case, we have the following set of states: $\{\rho_0,\rho_1,\ldots,\rho_p,\sigma_0,\sigma_1,\linebreak[0]\ldots,\sigma_q\}$.
If~$\rho$ or~$\sigma$ are cofinite, we let $p = 1 + \max\{ \N \setminus \rho \}$ and $q = 1 + \max\{ \N \setminus \sigma \}$.
In this case, we replace the last state in the given sets by $\rho_{\geq p}$ or $\rho_{\geq q}$, respectively.
This state represents a vertex in the vertex set~$D$ of the partial solution of the $[\rho,\sigma]$-domination problem that has at least~$p$ neighbours in~$D$, or a vertex not in~$D$ with at least~$q$ neighbours in~$D$, respectively.
Let $s=p+q+2$ be the number of states involved.

Dynamic programming tables for the $[\rho,\sigma]$-domination problems can also be represented using different sets of states that contain the same information.
In this section, we will use three different sets of states.
These sets are defined as follows.
\begin{definition} \label{def:rsstates}
Let State Set~I, II, and~III be the following sets of states:
\begin{itemize}
\item State Set I: $\{\rho_0,\rho_1,\rho_2,\ldots,\rho_{p-1}, \rho_p / \rho_{\geq p}, \sigma_0,\sigma_1,\sigma_2,\ldots,\sigma_{q-1},\sigma_q / \sigma_{\geq q} \}$.
\item StateSet II: $\{\rho_0,\rho_{\leq1},\rho_{\leq2},\ldots,\rho_{\leq p-1},\rho_{\leq p} / \rho_{\N}, \sigma_0,\sigma_{\leq 1},\sigma_{\leq 2},\ldots,\sigma_{\leq q-1},\sigma_{\leq q} / \sigma_{\N} \}$.
\item State Set III: $\{\rho_0,\rho_1,\rho_2,\ldots \rho_{p-1}, \rho_p / \rho_{\geq p-1}, \sigma_0,\sigma_1,\sigma_2,\ldots,\sigma_{q-1},\sigma_q / \sigma_{\geq q-1} \}$.
\end{itemize}
\end{definition}
\noindent The meaning of all the states is self-explanatory: $\rho_{condition}$ and $\sigma_{condition}$ consider the number of partial solutions of the $[\rho,\sigma]$-domination problem that do not contain ($\rho$-state) or do contain ($\sigma$-state) this vertex with a number of neighbours in the corresponding vertex sets satisfying the $condition$.
The subscript~$\N$ stands for no condition at all, i.e., $\rho_{\N} = \rho_{\geq 0}$: all possible number of neighbours in~$\N$.
We note that the notation $\rho_p / \rho_{\geq p}$ in Definition~\ref{def:rsstates} is used to indicate that this set uses the state~$\rho_p$ if~$\rho$ is finite and~$\rho_{\geq p}$ if~$\rho$ is cofinite.

\begin{lemma} \label{lem:rsstates}
Let~$x$ be a node of a tree decomposition~$T$ and let~$A_x$ be a table with entries $A_x(c,\kappa)$ representing the number of partial solutions of size~$\kappa$ to the $[\rho,\sigma]$-domination problem in~$G_x$ corresponding to each colouring~$c$ of the bag~$X_x$ with states from any of the three sets from Definition~\ref{def:rsstates}.
The information represented in the table~$A_x$ does not depend on the choice of the set of states from the options given in Definition~\ref{def:rsstates}.
Moreover, there exist transformations between tables using representations with different sets of states using $\bigO(s|X_x||A_x|)$ arithmetic operations.
\end{lemma}
\begin{proof}
We apply transformations that work in $|X_x|$ steps and are similar to those in the proofs of Lemmas~\ref{lem:dsstates} and~\ref{lem:pmstates}.
In the $i$-th step, we replace the states at the $i$-th coordinate of~$c$.
We use the following formulas to create entries with a new state.

We will give only the formulas for the $\rho$-states.
The formulas for the $\sigma$-states are identical, but with~$\rho$ replaced by~$\sigma$ and~$p$ replaced by~$q$.
We note that we slightly abuse notation below since we use that $\rho_{\leq 0} = \rho_0$.

To obtain states from State Set~I not present in State Set~II or~III, we can use:
\begin{eqnarray*}
A_x(c_1 \times \{\rho_j \} \times c_2,\kappa) & = & A_x(c_1 \times \{\rho_{\leq j}\} \times c_2, \kappa) - A_x(c_1 \times \{\rho_{\leq j-1}\} \times c_2, \kappa) \\
A_x(c_1 \times \{\rho_{\geq p} \} \times c_2,\kappa) & = & A_x(c_1 \times \{\rho_{\N}\} \times c_2, \kappa) - A_x(c_1 \times \{\rho_{\leq p-1}\} \times c_2, \kappa) \\
A_x(c_1 \times \{\rho_{\geq p} \} \times c_2,\kappa) & = & A_x(c_1 \times \{\rho_{\geq p-1}\} \times c_2, \kappa) - A_x(c_1 \times \{\rho_{p-1}\} \times c_2, \kappa)
\end{eqnarray*}

To obtain states from State Set~II not present in State Set~I or~III, we can use:
\begin{eqnarray*}
A_x(c_1 \times \{\rho_{\leq j} \} \times c_2,\kappa) & = & \sum_{l=0}^j A_x(c_1 \times \{\rho_l\} \times c_2, \kappa) \\
A_x(c_1 \times \{\rho_{\N} \} \times c_2,\kappa) & = & A_x(c_1 \times \{\rho_{\geq p}\} \times c_2, \kappa) + \sum_{l=0}^{p-1} A_x(c_1 \times \{\rho_l\} \times c_2, \kappa) \\
A_x(c_1 \times \{\rho_{\N} \} \times c_2,\kappa) & = & A_x(c_1 \times \{\rho_{\geq p-1}\} \times c_2, \kappa) + \sum_{l=0}^{p-2} A_x(c_1 \times \{\rho_l\} \times c_2, \kappa)
\end{eqnarray*}

To obtain states from State Set~III not present in State Set~I or~II, we can use the same formulas used to obtain states from State Set~I in combination with the following formulas:
\begin{eqnarray*}
A_x(c_1 \times \{\rho_{\geq p-1} \} \times c_2,\kappa) & = & A_x(c_1 \times \{\rho_{\geq p}\} \times c_2, \kappa) + A_x(c_1 \times \{\rho_{p-1}\} \times c_2, \kappa) \\
A_x(c_1 \times \{\rho_{\geq p-1} \} \times c_2,\kappa) & = & A_x(c_1 \times \{\rho_{\N}\} \times c_2, \kappa) - A_x(c_1 \times \{\rho_{\leq p-2}\} \times c_2, \kappa)
\end{eqnarray*}

As the transformations use~$|X_x|$ steps in which each entry is computed by evaluating a sum of less than~$s$ terms, the transformations require $\bigO(|X_x||A_x|)$ arithmetic operations.
\end{proof}

We note that similar transformations can also be used to transform a table into a new table that uses different sets of states on different vertices in a bag~$X_x$.
For example, we can use State Set~I on the first two vertices (assuming some ordering) and State Set~III on the other $|X_x|-2$ vertices.
We will use a transformation of this type in the proof of Theorem~\ref{thrm:rstwalg}.
\smallskip

To prove our main result for the $[\rho,\sigma]$-domination problems, we will also need more involved state transformations than those given above.
We need to generalise the ideas of the proof of Theorem~\ref{thrm:countingpmtwalg}.
In this proof, we expanded the tables~$A_l$ and~$A_r$ of the two child nodes~$l$ and~$r$ such that they contain entries $A_l(c,i)$ and $A_r(c,i)$, where~$i$ was an index indicating the number of $1$-states used to create the $?$-states in~$c$.
We will generalise this to the states used for the $[\rho,\sigma]$-domination problems.

Below, we often say that a colouring~$c$ of a bag~$X_x$ using State Set~I from Definition~\ref{def:rsstates} is \emph{counted} in a colouring~$c'$ of~$X_x$ using State Set~II.
We let this be the case when, all partial solutions counted in the entry with colouring~$c$ in a table using State Set~I are also counted in the entry with colouring~$c'$ in the same table when transformed such that it uses State Set~II.
I.e., when, for each vertex $v \in X_x$, $c(v)$ and $c'(v)$ are both $\sigma$-states or both $\rho$-states, and if $c(v) = \rho_i$ or $c(v) = \sigma_i$, then $c'(v) = \rho_{\leq j}$ or $c'(v) = \sigma_{\leq j}$ for some $j \geq i$.

Consider the case where~$\rho$ and~$\sigma$ are finite.
We introduce an \emph{index vector} $\vec{i} = (i_{\rho1}, i_{\rho2}, \ldots, i_{\rho p},i_{\sigma1},i_{\sigma2},\ldots,i_{\sigma q})$ that is used in combination with states from State Set~II from Definition~\ref{def:rsstates}.
In this index vector, $i_{\rho j}$ and~$i_{\sigma j}$ represent the sum over all vertices with state~$\rho_{\leq j}$ and~$\sigma_{\leq j}$ of the number of neighbours of the vertex in~$D$, respectively.
We say that a solution corresponding to a colouring~$c$ using State Set~I from Definition~\ref{def:rsstates} \emph{satisfies} a combination of a colouring~$c'$ using State Set~II and an index vector~$\vec{i}$ if: $c$ is counted in~$c'$, and for each~$i_{\rho j}$ or $i_{\sigma j}$, the sum over all vertices with state~$\rho_{\leq j}$ and~$\sigma_{\leq j}$ in~$c'$ of the number of neighbours of the vertex in~$D$ equals~$i_{\rho j}$ or~$i_{\sigma j}$, respectively.

We clarify this with an example.
Suppose that we have a bag of size three and a dynamic programming table indexed by colourings using the set of states $\{\rho_0,\rho_1,\rho_2,\sigma_0\}$ (State Set~I) that we want to transform to one using the set states $\{\rho_0,\rho_{\leq1},\rho_{\leq2},\sigma_0\}$ (State Set~II): thus $\vec{i} = (i_{\rho1},i_{\rho2})$.
Notice that a partial solution corresponding to the colouring $c=(\rho_0,\rho_1,\rho_2)$ will be counted in both $c'_1 = (\rho_{0},\rho_{\leq 2},\rho_{\leq 2})$ and $c'_2 = (\rho_{\leq 1},\rho_{\leq 1},\rho_{\leq 2})$.
In this case, $c$ satisfies the combination $(c'_1, \vec{i} = (0,3))$ since the sum of the subscripts of the states in~$c$ of the vertices with state~$\rho_{\leq 1}$~in $c'_1$ equals zero and this sum for the vertices with state~$\rho_{\leq 2}$ in~$c'_1$ equals three.
Also, $c$ satisfies no combination of $c'_1$ with an other index vector.
Similarly, $c$ satisfies the combination $(c'_2, \vec{i}=(1,2))$ and no other combination involving~$c'_2$.

In the case where~$\rho$ or~$\sigma$ are cofinite, the index vectors are one shorter: we do not count the sum of the number of neighbours in~$D$ of the vertices with state~$\rho_{\N}$ and~$\sigma_{\N}$.

What we will need is a table containing, for each possible combination of a colouring using State Set~II with an index vector, the number of partial solutions that satisfy these.
We can construct such a table using the following lemma.
\begin{lemma} \label{lem:rsstates2}
Let~$x$ be a node of a tree decomposition~$T$ of width~$k$.
There exists an algorithm that, given a table~$A_x$ with entries $A_x(c,\kappa)$ containing the number of partial solutions of size~$\kappa$ to the $[\rho,\sigma]$-domination problem corresponding to the colouring~$c$ on the bag~$X_x$ using State Set~I from Definition~\ref{def:rsstates}, computes in $\bigO(n(sk)^{s-1} s^{k+1} i_+(n))$ time a table~$A'_x$ with entries $A'_x(c,\kappa,\vec{i})$ containing the number partial solutions of size $\kappa$ to the $[\rho,\sigma]$-domination problem satisfying the combination of a colouring using State Set~II and the index vector $\vec{i}$. 
\end{lemma}
\begin{proof}
We start with the following table~$A'_x$ using State Set~I:
\[ A'_x(c,\kappa,\vec{i}) = \left\{ \begin{array}{ll} A_x(c,\kappa) & \textrm{if $\vec{i}$ is the all-0 vector} \\ 0 & \textrm{otherwise} \end{array} \right. \]
Since there are no colourings with states~$\rho_{\leq j}$ and~$\sigma_{\leq j}$ yet, the sum of the number of neighbours in the vertex set~$D$ of the partial solutions of vertices with these states is zero.

Next, we change the states of the $j$-th coordinate at step~$j$ similar to Lemma~\ref{lem:rsstates}, but now we also updates the index vector~$\vec{i}$:
\begin{eqnarray*}
A'_x(c_1 \times \{\rho_{\leq j} \} \times c_2,\kappa,\vec{i} ) & = & \sum_{l=0}^j A'_x(c_1 \times \{\rho_l\} \times c_2, \kappa, \vec{i}_{i_{\rho j} \rightarrow (i_{\rho j}-l)}) \\
A'_x(c_1 \times \{\sigma_{\leq j} \} \times c_2,\kappa,\vec{i} ) & = & \sum_{l=0}^j A'_x(c_1 \times \{\sigma_l\} \times c_2, \kappa, \vec{i}_{i_{\sigma j} \rightarrow (i_{\sigma j}-l)})
\end{eqnarray*}
Here, $\vec{i}_{i_{\rho j} \rightarrow (i_{\rho j}-l)}$ denotes the index vector~$\vec{i}$ with the value of~$i_{\rho j}$ set to $i_{\rho j}-l$.

If~$\rho$ or~$\sigma$ are cofinite, we simply use the formula in Lemma~\ref{lem:rsstates} for every fixed index vector~$\vec{i}$ for the~$\rho_{\N}$-states and~$\sigma_{\N}$-states.
We do so because we do not need to keep track of any index vectors for these states.

For the running time, note that each index~$i_{\rho j}$, $i_{\sigma j}$ can have only values between zero and~$sk$ because there can be at most~$k$ vertices in~$X_x$ that each have at most~$s$ neighbours in~$D$ when considered for a state of the form~$\rho_{\leq j}$ or~$\sigma_{\leq j}$, as $j < p$ or $j < q$, respectively.
The new table has $\bigO(n (sk)^{s-2} s^{k+1})$ entries since we have $s^{k+1}$ colourings, $n + 1$ sizes~$\kappa$, and $s-2$ indices that range over~$sk$ values.
Since the algorithm uses at most $k+1$ steps in which it computes a sum with less than~$s$ terms for each entry using $n$-bit numbers, this gives a running time of $\bigO(n(sk)^{s-1} s^{k+1} i_+(n))$.
\end{proof}

We are now ready to prove our main result of this section.
\begin{theorem} \label{thrm:rstwalg}
Let $\rho, \sigma \subseteq \N$ be finite or cofinite, and let~$p$, $q$ and~$s$ be the values associated with the corresponding $[\rho,\sigma]$-domination problem.
There is an algorithm that, given a tree decomposition of a graph~$G$ of width~$k$, computes the number of $[\rho,\sigma]$-dominating sets in~$G$ of each size $\kappa$, $0 \leq \kappa \leq n$, in $\bigO(n^3 (sk)^{2(s-2)} s^{k+1} i_\times(n))$ time.
\end{theorem}
Notice that, for any given $[\rho,\sigma]$-domination problem, $s$ is a fixed constant.
Hence, Theorem~\ref{thrm:rstwalg} gives us $\bigOs(s^k)$-time algorithms for these problems.

\begin{proof}
Before we give the computations involved for each type of node in a nice tree decomposition~$T$, we slightly change the meaning of the subscript of the states $\rho_{condition}$ and $\sigma_{condition}$.
In our algorithm, we let the subscripts of these states count only the number of neighbours in the vertex sets~$D$ of the partial solution of the $[\rho,\sigma]$-domination problem that have already been forgotten by the algorithm.
This prevents us from having to keep track of any adjacencies within a bag during a join operation.
We will update these subscripts in the forget nodes.
This modification is similar to the approach for counting perfect matchings in the proof of Theorem~\ref{thrm:countingpmtwalg}, where we matched vertices in a forget node to make sure that we did not have to deal with vertices that are matched within a bag when computing the table for a join node.

We will now give the computations for each type of node in a nice tree decomposition~$T$.
For each node $x \in T$, we will compute a table $A_x(c,\kappa)$ containing the number of partial solutions of size~$\kappa$ in~$G_x$ corresponding to the colouring~$c$ on~$X_x$ for all colourings~$c$ using State Set~I from Definition~\ref{def:rsstates} and all $0 \leq \kappa \leq n$.
During this computation, we will transform to different sets of states using Lemmas~\ref{lem:rsstates} and~\ref{lem:rsstates2} when necessary.

\smallskip \noindent {\it Leaf node}:
Let~$x$ be a leaf node in~$T$. 

Because the subscripts of the states count only neighbours in the vertex set of the partial solutions that have already been forgotten, we use only the states~$\rho_0$ and~$\sigma_0$ on a leaf.
Furthermore, the number of $\sigma$-states must equal~$\kappa$.
As a result, we can compute~$A_x$ in the following way:
\[ A_x(c,\kappa) = \left\{ \begin{array}{ll} 1 & \textrm{if $c = \{\rho_0\}$ and $\kappa = 0$} \\ 1 & \textrm{if $c = \{\sigma_0\}$ and $\kappa = 1$} \\ 0 & \textrm{otherwise} \end{array} \right. \]

\smallskip \noindent {\it Introduce node}:
Let~$x$ be an introduce node in~$T$ with child node~$y$ introducing the vertex~$v$.

Again, the entries where~$v$ has the states~$\rho_j$ or~$\sigma_j$, for $j \geq 1$, will be zero due to the definition of the (subscripts of) the states.
Also, we must again keep track of the size~$\kappa$.
Let~$\varsigma$ be the state of the introduced vertex.
We compute $A_x$ in the following way:
\[ A_x(c \times \{\varsigma\},\kappa) = \left\{ \begin{array}{ll} A_y(c,\kappa) & \textrm{if $\varsigma=\rho_0$} \\ A_y(c,\kappa-1) & \textrm{if $\varsigma=\sigma_0$ and $\kappa \geq 1$} \\ 0 & \textrm{otherwise}    \end{array} \right. \]

\smallskip \noindent {\it Forget node}:
Let~$x$ be a forget node in~$T$ with child node~$y$ forgetting the vertex~$v$.

The operations performed in the forget node are quite complicated.
Here, we must update the states such that they are correct after forgetting the vertex~$v$, and we must select those solutions that satisfy the constraints imposed on~$v$ by the specific $[\rho,\sigma]$-domination problem.
We will do this in three steps: we compute intermediate tables~$A_1$, $A_2$ in the first two steps and finally~$A_x$ in step three.
Let $c(N(v))$ be the subcolouring of~$c$ restricted to vertices in $N(v)$.

\emph{Step 1}: We update the states used on the vertex~$v$.
We do so to include the neighbours in~$D$ that the vertex~$v$ has inside the bag~$X_x$ in the states used to represent the different characteristics.
Notice that after including these neighbours, the subscripts of the states on~$v$ represent the total number of neighbours that~$v$ has in~$D$.
The result will be the table~$A_1$, which we compute using the following formulas where $\#_\sigma(c)$ stands for the number of~$\sigma$-states in the colouring~$c$:
\begin{eqnarray*}
A_1(c \times \{ \rho_j \}, \kappa) & = & \left\{ \begin{array}{ll} A_y(c \times \{ \rho_{j - \#_\sigma(c(N(v)))} \} , \kappa) & \textrm{if $j \geq \#_\sigma(c(N(v)))$} \\ 0 & \textrm{otherwise} \end{array} \right. \\
A_1(c \times \{ \sigma_j \}, \kappa) & = & \left\{ \begin{array}{ll} A_y(c \times \{ \sigma_{j - \#_\sigma(c(N(v)))} \} , \kappa) & \textrm{if $j \geq \#_\sigma(c(N(v)))$} \\ 0 & \textrm{otherwise} \end{array} \right.
\end{eqnarray*}
If~$\rho$ or~$\sigma$ are cofinite, we also need the following formulas:
\begin{eqnarray*}
A_1(c \times \{ \rho_{\geq p} \}, \kappa) & = & A_y(c \times \{ \rho_{\geq p} \}, \kappa) + \sum_{i = p - \#_\sigma(c(N(v)))}^{p-1} A_y(c \times \{ \rho_i \},\kappa) \\
A_1(c \times \{ \sigma_{\geq q} \}, \kappa) & = & A_y(c \times \{ \sigma_{\geq q} \}, \kappa) + \sum_{i = q - \#_\sigma(c(N(v)))}^{q-1} A_y(c \times \{ \sigma_i \},\kappa)
\end{eqnarray*}
Correctness of these formulas is easy to verify.

\emph{Step 2}: We update the states representing the neighbours of~$v$ such that they are according to their definitions after forgetting~$v$.
All the required information to do this can again be read from the colouring~$c$.

We apply Lemma~\ref{lem:rsstates} and change the state representation for the vertices in $N(v)$ to State Set~III (Definition~\ref{def:rsstates}) obtaining the table $A_1'(c,\kappa)$; we do not change the representation of other vertices in the bag.
That is, if~$\rho$ or~$\sigma$ are cofinite, we replace the last state~$\rho_{\geq p}$ or~$\sigma_{\geq q}$ by~$\rho_{\geq p-1}$ or~$\sigma_{\geq q-1}$, respectively, on vertices in $X_y \cap N(v)$.
We can do so as discussed below the proof of \ref{lem:rsstates}.

This state change allows us to extract the required values for the table~$A_2$, as we will show next.
We introduce the function~$\phi$ that will send a colouring using State Set~I to a colouring that uses State Set~I on the vertices in $X_y \setminus N(v)$ and State Set~III on the vertices in $X_y \cap N(v)$.
This function updates the states used on~$N(v)$ assuming that we would put~$v$ in the vertex set~$D$ of the partial solution.
We define $\phi$ in the following way: it maps a colouring~$c$ to a new colouring with the same states on vertices in $X_y \setminus N(v)$ while it applies the following replacement rules on the states on vertices in $X_y \cap N(v)$: $\rho_1 \mapsto \rho_0$, $\rho_2 \mapsto \rho_1$, \ldots, $\rho_p \mapsto \rho_{p-1}$, $\rho_{\geq p} \mapsto \rho_{\geq p-1}$, $\sigma_1 \mapsto \sigma_0$, $\sigma_2 \mapsto \sigma_1$, \ldots, $\sigma_q \mapsto \sigma_{q-1}$, $\sigma_{\geq q} \mapsto \sigma_{\geq q-1}$.
Thus, $\phi$ lowers the counters in the conditions that index the states by one for states representing vertices in $N(v)$.
We note that~$\phi(c)$ is defined only if $\rho_0, \sigma_0 \not\in c$.

Using this function, we can easily update our states as required:
\begin{eqnarray*}
A_2( c \times \{ \sigma_j \}, \kappa) & = & \left\{ \begin{array}{ll} A'_1(\phi(c) \times \{\sigma_j\},\kappa) & \textrm{if $\rho_0,\sigma_0 \not\in c(N(v))$} \\ 0 & \textrm{otherwise} \end{array} \right. \\
A_2( c \times \{ \rho_j \}, \kappa) & = & A'_1(c \times \{ \rho_j \}, \kappa)
\end{eqnarray*}
In words, for partial solutions on which the vertex~$v$ that we will forget has a~$\sigma$-state, we update the states for vertices in $X_y \cap N(v)$ such that the vertex~$v$ is counted in the subscript of the states.
Entries in~$A_2$ are set to 0 if the states count no neighbours in~$D$ while~$v$ has a~$\sigma$-state in~$c$ and thus a neighbour in~$D$ in this partial solution.

Notice that after updating the states using the above formula the colourings~$c$ in~$A_2$ again uses State Set~I from Definition~\ref{def:rsstates}.

\emph{Step 3}: We select the solutions that satisfy the constraints of the specific $[\rho,\sigma]$-domination problem on~$v$ and forget~$v$.
\[ A_x(c, \kappa) = \left( \sum_{i \in \rho} A_2(c \times \{ \rho_i \}, \kappa)  \right) + \left(  \sum_{i \in \sigma} A_2(c \times \{ \sigma_i \}, \kappa) \right) \]
We slightly abuse our notation here when~$\rho$ or~$\sigma$ are cofinite.
Following the discussion of the construction of the table~$A_x$, we conclude that this correctly computes the required values.

\smallskip \noindent {\it Join node}:
Let~$x$ be a join node in~$T$ and let~$l$ and~$r$ be its child nodes.
Computing the table~$A_x$ for the join node~$x$ is the most interesting operation.

First, we transform the tables~$A_l$ and~$A_r$ of the child nodes such that they use State Set~II (Definition~\ref{def:rsstates}) and are indexed by index vectors using Lemma~\ref{lem:rsstates2}.
As a result, we obtain tables~$A'_l$ and~$A'_r$ with entries $A'_l(c,\kappa,\vec{g})$ and $A'_r(c,\kappa,\vec{h})$.
These entries count the number of partial solutions of size~$\kappa$ corresponding to the colouring~$c$ such that the sum of the number of neighbours in~$D$ of the set of vertices with each state equals the value that the index vectors~$\vec{g}$ and~$\vec{h}$ indicate.
Here, $D$ is again the vertex set of the partial solution involved.
See the example above the statement of Lemma~\ref{lem:rsstates2}.

Then, we compute the table $A_x(c,\kappa,\vec{i})$ by combining identical states from~$A'_l$ and~$A'_r$ using the formula below.
In this formula, we sum over all ways of obtaining a partial solution of size~$\kappa$ by combining the sizes in the tables of the child nodes and all ways of obtaining index vector $\vec{i}$ from $\vec{i} = \vec{g} + \vec{h}$.
\[ A'_x(c,\kappa,\vec{i}) = \sum_{\kappa_l + \kappa_r = \kappa + \#_\sigma(c)} \left( \sum_{i_{\rho1}=g_{\rho1}+h_{\rho1}} \!\!\cdots\!\! \sum_{i_{\sigma q}=g_{\sigma q}+h_{\sigma q}} A'_l(c,\kappa_l,\vec{g}) \cdot A'_r(c,\kappa_r,\vec{h}) \right) \]

We observe the following: a partial solution~$D$ in~$A'_x$ that is a combination of partial solutions from~$A'_l$ and~$A'_r$ is counted in an entry in $A'_x(c,\kappa,\vec{i})$ if and only if it satisfies the following three conditions.
\begin{enumerate}
\item The sum over all vertices with state~$\rho_{\leq j}$ and~$\sigma_{\leq j}$ of the number of neighbours of the vertex in~$D$ of this combined partial solution equals $i_{\rho j}$ or $i_{\sigma j}$, respectively.
\item The number of neighbours in~$D$ of each vertex with state~$\rho_{\leq j}$ or~$\sigma_{\leq j}$ of both partial solutions used to create this combined solution is at most~$j$.
\item The total number of vertices in~$D$ in this joined solution is~$\kappa$.
\end{enumerate}

Let $\Sigma_\rho^l(c)$, $\Sigma_\sigma^l(c)$ be the weighted sums of the number of~$\rho_j$-states and~$\sigma_j$-states with $0 \leq j \leq l$ in~$c$, respectively, defined by:
\[ \Sigma_\rho^l(c) = \sum_{j=1}^l j \cdot \#_{\rho_j}(c) \qquad \qquad \Sigma_\sigma^l(c) = \sum_{j=1}^l j \cdot \#_{\sigma_j}(c) \]
We note that $\Sigma_\rho^1(c) = \#_{\rho_1}(c)$ and $\Sigma_\sigma^1(c) = \#_{\sigma_1}(c)$.

Now, using Lemma~\ref{lem:rsstates}, we change the states used in the table $A'_x$ back to State Set~I.
If~$\rho$ and~$\sigma$ are finite, we extract the values computed for the final table~$A_x$ in the following way:
\[ A_x(c,\kappa) = A'_x \left( c, \; \kappa, \; (\Sigma_\rho^1(c),\Sigma_\rho^2(c),\ldots,\Sigma_\rho^p(c),\Sigma_\sigma^1(c),\Sigma_\sigma^2(c),\ldots,\Sigma_\sigma^q(c)) \; \right) \]
If~$\rho$ or~$\sigma$ are cofinite, we use the same formula but omit the components $\Sigma_\rho^p(c)$ or $\Sigma_\sigma^q(c)$ from the index vector of the extracted entries, respectively.

Below, we will prove that the entries in~$A_x$ are exactly the values that we want to compute.
We first give some intuition.
In essence, the proof is a generalisation of how we performed the join operation for counting the number of perfect matchings in the proof of Theorem~\ref{thrm:countingpmtwalg}.
State Set~II has the role of the $?$-states in the proof of Theorem~\ref{thrm:countingpmtwalg}.
These states are used to count possible combinations of partial solutions from~$A_l$ and~$A_r$.
These combinations include incorrect combinations in the sense that a vertex can have more neighbours in~$D$ than it should have; this is analogous to counting the number of perfect matchings, where combinations were incorrect if a vertex is matched twice.
The values $\Sigma_\rho^l(c)$ and $\Sigma_\sigma^l(c)$ represent the total number of neighbours in~$D$ of the vertices with a~$\rho_j$-states or~$\sigma_j$-states with $0 \leq j \leq l$ in~$c$, respectively.
The above formula uses these $\Sigma_\rho^l(c)$ and $\Sigma_\sigma^l(c)$ to extract exactly those values from the table $A'_x$ that correspond to correct combinations.
That is, in this case, correct combinations for which the number of neighbours of a vertex in~$D$ is also correctly represented by the new states.

We will now prove that the computation of the entries in~$A_x$ gives the correct values.
An entry in $A_x(c,\kappa)$ with $c \in \{\rho_0,\sigma_0\}^k$ is correct: these states are unaffected by the state changes and the index vector is not used.
The values of these entries follow from combinations of partial solutions from both child nodes corresponding to the same states on the vertices.

Now consider an entry in $A_x(c,\kappa)$ with $c \in \{\rho_0,\rho_1,\sigma_0\}^k$.
Each~$\rho_1$-state comes from a~$\rho_{\leq 1}$-state in $A_x'(c,\kappa,\vec{i})$ and is a combination of partial solutions from~$A_l$ and~$A_r$ with the following combinations of states on this vertex: $(\rho_0,\rho_0)$, $(\rho_0,\rho_1)$, $(\rho_1,\rho_0)$, $(\rho_1,\rho_1)$.
Because we have changed states back to State Set~I, each $(\rho_0,\rho_0)$ combination is counted in the~$\rho_0$-state on this vertex, and thus subtracted from the combinations used to form state~$\rho_1$: the other three combinations remain counted in the~$\rho_1$-state.
Since we consider only those solutions with index vector $i_{\rho_1} = \Sigma_\rho^1(c)$, the total number of~$\rho_1$-states used to form this joined solution equals $\Sigma_\rho^1(c) = \#_{\rho_1}(c)$.
Therefore, no $(\rho_1,\rho_1)$ combination could have been used, and each partial solution counted in $A(c,\kappa)$ has exactly one neighbour in~$D$ on each of the~$\rho_1$-states, as required.

We can now inductively repeat this argument for the other states.
For $c \in \{\rho_0,\rho_1,\rho_2,\sigma_0\}^k$, we know that the entries with only~$\rho_0$-states and~$\rho_1$-states are correct.
Thus, when a~$\rho_2$-state is formed from a~$\rho_{\leq 2}$-state during the state transformation of Lemma~\ref{lem:rsstates}, all nine possibilities of getting the state~$\rho_{\leq 2}$ from the states~$\rho_0$, $\rho_1$, and~$\rho_2$ in the child bags are counted, and from this number all three combinations that should lead to a~$\rho_0$ and~$\rho_1$ in the join are subtracted.
What remains are the combinations $(\rho_0,\rho_2)$, $(\rho_1,\rho_2)$, $(\rho_2,\rho_2)$, $(\rho_1,\rho_1)$, $(\rho_1,\rho_2)$, $(\rho_2,\rho_0)$.
Because of the index vector of the specific the entry we extracted from~$A'_x$, the total sum of the number of neighbours in~$D$ of these vertices equals $\Sigma_\rho^2$, and hence only the combinations $(\rho_0,\rho_2)$, $(\rho_1,\rho_1)$, and $(\rho_2,\rho_0)$ could have been used.
Any other combination would raise the component~$i_{\rho2}$ of~$\vec{i}$ to a number larger than $\Sigma_\rho^2$.

If we repeat this argument for all states involved, we conclude that the above computation correctly computes~$A_x$ if~$\rho$ and~$\sigma$ are finite.
If~$\rho$ or~$\sigma$ are cofinite, then the argument can also be used with one small difference.
Namely, the index vectors are one component shorter and keep no index for the states~$\rho_{\N}$ and~$\sigma_{\N}$.
That is, at the point in the algorithm where we introduce these index vectors and transform to State Set~II using Lemma~\ref{lem:rsstates2}, we have no index corresponding to the sum of the number of neighbours in the vertex set~$D$ of the partial solution of the vertices with states~$\rho_{\N}$ and~$\sigma_{\N}$.
However, we do not need to select entries corresponding to having~$p$ or~$q$ neighbours in~$D$ for the states~$\rho_{\geq p}$ and~$\sigma_{\geq q}$ since these correspond to all possibilities of getting at least~$p$ or~$q$ neighbours in~$D$.
When we transform the states back to State Set~I just before extracting the values for~$A_x$ from~$A'_x$, entries that have the state~$\rho_{\geq p}$ or~$\sigma_{\geq q}$ after the transformation count all possible combinations of partial solutions except those counted in any of the other states.
This is exactly what we need since all combinations with less than~$p$ (or~$q$) neighbours are present in the other states.

\smallskip
After traversing the whole decomposition tree $T$, one can find the number of $[\rho,\sigma]$-dominating sets of size $\kappa$ in the table computed for the root node $z$ of $T$ in $A_z(\emptyset,\kappa)$.

\smallskip
We conclude with an analysis of the running time.
The most time-consuming computations are again those involved in computing the table~$A_x$ for a join node~$x$.
Here, we need $\bigO(n(sk)^{s-1} s^{k+1} i_+(n))$ time for the transformations of Lemma~\ref{lem:rsstates2} that introduce the index vectors since $\max\{ |X_x| \;|\; x \in T\} = k+1$.
However, this is still dominated by the time required to compute the table~$A'_x$: this table contains at most $s^{k+1} n (sk)^{s-2}$ entries $A'_x(c,\kappa,\vec{i})$, each of which is computed by an $n (sk)^{s-2}$-term sum.
This gives a total time of $\bigO(n^2 (sk)^{2(s-2)} s^{k+1} i_\times(n))$ since we use $n$-bit numbers.
Because the nice tree decomposition has $\bigO(n)$ nodes, we conclude that the algorithm runs in $\bigO(n^3 (sk)^{2(s-2)} s^{k+1} i_\times(n))$ time in total.
\end{proof}

This proof generalises ideas from the fast subset convolution algorithm~\cite{BjorklundHKK07}.
While convolutions use ranked M\"{o}bius transforms~\cite{BjorklundHKK07}, we use transformations with multiple states and multiple ranks in our index vectors.

The polynomial factors in the proof of Theorem~\ref{thrm:rstwalg} can be improved in several ways.
Some improvements we give are for $[\rho,\sigma]$-domination problems in general, and others apply only to specific problems.
Similar to $s = p + q + 2$, we define the value~$r$ associated with a $[\rho, \sigma]$-domination problems as follows:
\[ r = \left\{ \begin{array}{ll} \max\{p-1,q-1\} & \textrm{if $\rho$ and $\sigma$ are cofinite} \\ \max\{p,q-1\} & \textrm{if $\rho$ is finite and $\sigma$ is cofinite} \\ \max\{p-1,q\} & \textrm{if $\rho$ is confinite and $\sigma$ is finite} \\ \max\{p,q\} & \textrm{if $\rho$ and $\sigma$ are finite} \end{array} \right. \]

\begin{corollary}[General {$\boldsymbol{[\rho, \sigma]}$}-Domination Problems] \label{cor:generalrstwalg}
Let $\rho, \sigma \subseteq \N$ be finite or cofinite, and let~$p$, $q$, $r$, and~$s$ be the values associated with the corresponding $[\rho,\sigma]$-domination problem.
There is an algorithm that, given a tree decomposition of a graph~$G$ of width~$k$, computes the number of $[\rho,\sigma]$-dominating sets in~$G$ of each size~$\kappa$, $0 \leq \kappa \leq n$, in $\bigO(n^3 (rk)^{2r} s^{k+1} i_\times(n))$ time.
Moreover, there is an algorithm that decides whether there exist a $[\rho,\sigma]$-dominating set of size~$\kappa$, for each individual value of $\kappa$, $0 \leq \kappa \leq n$, in $\bigO(n^3 (rk)^{2r} s^{k+1} i_\times(log(n)+k\log(r)))$ time.
\end{corollary}
\begin{proof}
We improve the polynomial factor $(sk)^{2(s-2)}$ to $(rk)^{2r}$ by making the following observation.
We never combine partial solutions corresponding to a~$\rho$-state in one child node with a partial solution corresponding to a~$\sigma$-state on the same vertex in the other child node.
Therefore, we can combine the components of the index vector related to the states~$\rho_j$ and~$\sigma_j$ for each fixed~$j$ in a single index.
For example consider the~$\rho_1$-states and~$\sigma_1$-states.
For these states, this means the following: if we index the number of vertices used to create a~$\rho_1$-state and~$\sigma_1$-state in~$i_1$ and we have~$i_1$ vertices on which a partial solution is formed by considering the combinations $(\rho_0,\rho_1)$, $(\rho_1,\rho_0)$, $(\rho_1,\rho_1)$, $(\sigma_0,\sigma_1)$, $(\sigma_1,\sigma_0)$, or $(\sigma_1,\sigma_1)$, then non of the combinations $(\rho_1,\rho_1)$ and $(\sigma_1,\sigma_1)$ could have been used.
Since the new components of the index vector range between~$0$ and $rk$, this proves the first running time in the statement of the corollary.

The second running time follows from reasoning similar to that in Corollary~\ref{cor:solvedstwalg}.
In this case, we can stop counting the number of partial solutions of each size and instead keep track of the existence of a partial solution of each size.
The state transformations then count the number of $1$-entries in the initial tables instead of the number of solutions.
After computing the table for a join node, we have to reset all entries~$e$ of~$A_x$ to $\min\{1,e\}$.
For these computations, we can use $\bigO(\log(n)+k\log(r))$-bit numbers.
This is because of the following reasoning.
For a fixed colouring~$c$ using State Set~II, each of the at most $r^{k+1}$ colourings using State Set~I that can be counted in~$c$ occur with at most one index vector in the tables~$A'_l$ and~$A'_r$.
Note that these are $r^{k+1}$ colourings, not $s^{k+1}$ colourings, since $\rho$-states are never counted in a colouring~$c$ where the vertex has a $\sigma$-state and vice versa.
Therefore, the result of the large summation over all index vectors~$\vec{g}$ and~$\vec{h}$ with $\vec{i}=\vec{g}+\vec{h}$ can be bounded from above by $(r^k)^2$.
Since we sum over $n$ possible combinations of sizes, the maximum is $nr^{2k}$ allowing us to use $\bigO(\log(n)+k\log(r))$-bit numbers.
\end{proof}

As a result, we can, for example, compute the size of a minimum-cardinality perfect code in $\bigO(n^3k^23^k i_\times(\log(n)))$ time.
Note that the time bound follows because the problem is fixed and we use a computational model with $\bigO(k)$-bit word size.

\begin{corollary}[{$\boldsymbol{[\rho, \sigma]}$}-Optimisation Problems with the de Fluiter Property] \label{cor:defluiterrstwalg}
Let $\rho, \sigma \subseteq \N$ be finite or cofinite, and let~$p$, $q$, $r$, and~$s$ be the values associated with the corresponding $[\rho,\sigma]$-domination problem.
If the standard representation using State Set~I of the minimisation (or maximisation) variant of this $[\rho,\sigma]$-domination problem has the de Fluiter property for treewidth with function~$f$, then there is an algorithm that, given a tree decomposition of a graph~$G$ of width~$k$, computes the number of minimum (or maximum) $[\rho,\sigma]$-dominating sets in~$G$ in $\bigO(n (f(k))^2 (rk)^{2r} s^{k+1} i_\times(n))$ time.
Moreover, there is an algorithm that computes the minimum (or maximum) size of such a $[\rho,\sigma]$-dominating set in $\bigO(n (f(k))^2 (rk)^{2r} s^{k+1} i_\times(log(n)+k\log(r)))$ time.
\end{corollary}
\begin{proof}
The difference with the proof of Corollary~\ref{cor:generalrstwalg} is that, similar to the proof of Corollary~\ref{cor:countmdstwalg}, we can keep track of the minimum or maximum size of a partial solution in each node of the tree decomposition and consider only other partial solutions whose size differs at most $f(k)$ of this minimum or maximum size.
As a result, both factors~$n$ (the factor~$n$ due to the size of the tables, and the factor~$n$ due to the summation over the sizes of partial solutions) are replaced by a factor $f(k)$.
\end{proof}

As an application of Corollary~\ref{cor:defluiterrstwalg}, it follows for example that {\sc 2-Dominating Set} can be solved in $\bigO(n k^6 4^k i_\times(\log(n)))$ time.

\begin{corollary}[{$\boldsymbol{[\rho, \sigma]}$}-Decision Problems] \label{cor:decisionrstwalg}
Let $\rho, \sigma \subseteq \N$ be finite or cofinite, and let~$p$, $q$, $r$, and~$s$ be the values associated with the corresponding $[\rho,\sigma]$-domination problem.
There is an algorithm that, given a tree decomposition of a graph~$G$ of width~$k$, counts the number of $[\rho,\sigma]$-dominating sets in~$G$ in $\bigO(n (rk)^{2r} s^{k+1} i_\times(n))$ time.
Moreover, there is an algorithm that decides whether there exists a $[\rho,\!\sigma]$-dominating set in $\bigO(n(rk)^{2r} s^{k+1} i_\times(log(n)+k\log(r)))$ time.
\end{corollary}
\begin{proof}
This result follows similarly as Corollary~\ref{cor:defluiterrstwalg}. 
In this case, we can omit the size parameter from our tables, and we can remove the sum over the sizes in the computation of entries of~$A'_x$ completely.
\end{proof}

As an application of Corollary~\ref{cor:decisionrstwalg}, it follows for example that we can compute the number of strong stable sets (distance-2 independent sets) in $\bigO(nk^23^k i_\times(n))$ time.

\subsection{Clique Covering, Packing and Partitioning Problems} \label{sec:cliquetwalg}
The final class of problems that we consider for our tree decomposition-based-algorithms are the clique covering, packing, and partitioning problems.
To give a general result, we defined the $\gamma$-clique covering, $\gamma$-clique packing, and $\gamma$-clique partitioning problems in Section~\ref{sec:problems}; see Definition~\ref{def:cliqueproblems}.
For these $\gamma$-clique problems, we obtain $\bigOs(2^k)$ algorithms.

Although any natural problem seems to satisfy this restriction, we remind the reader that we restrict ourselves to polynomial-time decidable $\gamma$, that is, given an integer $j$, we can decide in time polynomial in $j$ whether $j \in \gamma$ or not.
This allows us to precompute $\gamma \cap \{1,2,\ldots,k+1\}$ in time polynomial in $k$, after which we can decide in constant time whether a clique of size $l$ is allowed to be used in an associated covering, packing, or partitioning.

We start by giving algorithms for the $\gamma$-clique packing and partitioning problems.
\begin{theorem} \label{thrm:cliqueparttwalg}
Let $\gamma \subseteq \N \setminus \{0\}$ be polynomial-time decidable.
There is an algorithm that, given a tree decomposition of a graph~$G$ of width~$k$, computes the number of  $\gamma$-clique packings or $\gamma$-clique partitionings of~$G$ using~$\kappa$, $0 \leq \kappa \leq n$, cliques in $\bigO(n^3 k^2 2^k i_\times(nk + n\log(n)))$ time.
\end{theorem}
\begin{proof}
Before we start dynamic programming on the tree decomposition~$T$, we first compute the set $\gamma \cap \{1,2,\dots,k+1\}$.

We use states~$0$ and~$1$ for the colourings~$c$, where~$1$ means that a vertex is already in a clique in the partial solution, and~$0$ means that the vertex is not in a clique in the partial solution.
For each node $x \in T$, we compute a table~$A_x$ with entries $A_x(c,\kappa)$ containing the number of $\gamma$-clique packings or partitionings of~$G_x$ consisting of exactly~$\kappa$ cliques that satisfy the requirements defined by the colouring $c \in \{1,0\}^{|X_x|}$, for all $0 \leq \kappa \leq n$.

The algorithm uses the well-known property of tree decompositions that for every clique~$C$ in the graph~$G$, there exists a node $x \in T$ such that~$C$ is contained in the bag~$X_x$ (a nice proof of this property can be found in \cite{BodlaenderM93}).
As every vertex in~$G$ is forgotten in exactly one forget node in~$T$, we can implicitly assign a unique forget node~$x_C$ to every clique~$C$, namely the first forget node that forgets a vertex from~$C$. 
In this forget node~$x_C$, we will update the dynamic programming tables such that they take the choice of whether to pick~$C$ in a solution into account.

\smallskip \noindent {\it Leaf node}:
Let~$x$ be a leaf node in~$T$. We compute~$A_x$ in the following way:
\[ A_x(\{0\},\kappa) = \left\{ \begin{array}{ll} 1 & \textrm{if $\kappa = 0$} \\ 0 & \textrm{otherwise} \end{array} \right. \qquad \qquad A_x(\{1\},\kappa) = 0 \]
Since we decide to take cliques in a partial solution only in the forget nodes, the only partial solution we count in~$A_x$ is the empty solution.

\smallskip \noindent {\it Introduce node}:
Let~$x$ be an introduce node in~$T$ with child node~$y$ introducing the vertex~$v$.
Deciding whether to take a clique in a solution in the corresponding forget nodes makes the introduce operation trivial since the introduced vertex must have state~$0$:
\[ A_x(c \times \{1\},\kappa) = 0 \qquad \qquad  A_x(c \times \{0\},\kappa) = A_y(c,\kappa) \]

\smallskip \noindent {\it Join node}:
In contrast to previous algorithms, we will first present the computations in the join nodes.
We do so because we will use this operation as a subroutine in the forget nodes.
\begin{figure}[tb]
	\begin{center}
	\begin{multicols}{2}
	\hspace{2cm}
	\begin{tabular}{c||c|c|} $\times$ & $0$ & $1$ \\ \hline \hline $0$ & $0$ & $1$ \\ \hline $1$ & $1$ & \\ \hline 
	\end{tabular} \\	
	\begin{tabular}{c||c|c|} $\times$ & $0$ & $1$ \\ \hline \hline $0$ & $0$ & $1$ \\ \hline $1$ & $1$ & $1$ \\ \hline 
	\end{tabular}
	\hspace{2cm}
	\end{multicols}
	\end{center}
	\caption{Join tables for $\gamma$-clique problems: the left table corresponds to partitioning and packing problems and the right table corresponds to covering problems.}
	\label{fig:jointableclique}
\end{figure}

Let~$x$ be a join node in~$T$ and let~$l$ and~$r$ be its child nodes.
For the $\gamma$-clique partitioning and packing problems, the join is very similar to the join in the algorithm for counting the number of perfect matchings (Theorem~\ref{thrm:countingpmtwalg}).
This can be seen from the corresponding join table; see Figure~\ref{fig:jointableclique}.
The only difference is that we now also have the size parameter~$\kappa$.
Hence, for $y \in \{l,r\}$, we first create the tables~$A'_y$ with entries $A'_y(c,\kappa,i)$, where~$i$ indexes the number of $1$-states in~$c$.
Then, we transform the set of states used for these tables~$A'_y$ from $\{1,0\}$ to $\{0,?\}$ using Lemma~\ref{lem:pmstates}, and compute the table~$A'_x$, now with the extra size parameter~$\kappa$, using the following formula:
\[ A'_x(c,\kappa,i) = \sum_{\kappa_l + \kappa_r = \kappa} \sum_{i = i_l + i_r} A'_l(c,\kappa_l,i_l) \cdot A'_r(c,\kappa_r,i_r) \]
Finally, the states in~$A'_x$ are transformed back to the set $\{0,1\}$, after which the entries of~$A_x$ can be extracted that correspond to the correct number of 1-states in~$c$.
Because the approach described above is a simple extension of the join operation in the proof of Theorem~\ref{thrm:countingpmtwalg} which was also used in the proof of Theorem~\ref{thrm:rstwalg}, we omit further details.

\smallskip \noindent {\it Forget node}:
Let~$x$ be a forget node in~$T$ with child node~$y$ forgetting the vertex~$v$.
Here, we first update the table~$A_y$ such that it takes into account the choice of taking any clique in~$X_y$ that contains~$y$ in a solution or not.

Let~$M$ be a table with all the (non-empty) cliques~$C$ in~$X_y$ that contain the vertex~$v$ and such that $|C| \in \gamma$, i.e., $M$ contains all the cliques that we need to consider before forgetting the vertex~$v$.
We notice that the operation of updating~$A_y$ such that it takes into account all possible ways of choosing the cliques in~$M$ is identical to letting the new~$A_y$ be the result of the join operation on~$A_y$ and the following table~$A_M$:
\begin{eqnarray*}
A_M(c \!\times\! \{1\},\kappa) & \!\!\!=\!\!\! & \left\{ \begin{array}{ll} 1 & \textrm{if all the 1-states in $c$ form a clique with $v$ in $M$ and $\kappa \!=\! 1$} \\ 1 & \textrm{if $c$ is the colouring with only 0-states and $\kappa \!=\! 0$} \\ 0 & \textrm{otherwise} \end{array} \right. \\
A_M(c \!\times\! \{0\},\kappa) & \!\!\!=\!\!\! & 0
\end{eqnarray*}
It is not hard to see that this updates~$A_y$ as required since $A_M(c,\kappa)$ is non-zero only when a clique in~$M$ is used with size $\kappa = 1$, or if no clique is used and $\kappa = 0$.

If we consider a partitioning problem, then $A_x(c,\kappa) = A_y(c \times \{1\},\kappa)$ since~$v$ must be contained in a clique.
If we consider a packing problem, then $A_x(c,\kappa) = A_y(c \times \{1\},\kappa) + A_y(c \times \{0\},\kappa)$ since~$v$ can but does not need to be in a clique.
Clearly, this correctly computes~$A_y$.

\smallskip
After computing~$A_z$ for the root node~$z$ of~$T$, the number of $\gamma$-clique packings or partitionings of each size~$\kappa$ can be found in $A_z(\emptyset,\kappa)$.

\smallskip
For the running time, we first observe that there are at most $\bigO(n2^k)$ cliques in~$G$ since~$T$ has $\bigO(n)$ nodes that each contain at most $k+1$ vertices.
Hence, there are at most $\bigO((n2^k)^n)$ ways to pick at most~$n$ cliques, and we can work with $\bigO(nk + n\log(n))$-bit numbers.
As a join and a forget operation require $\bigO(n^2 k^2 2^k)$ arithmetical operations, the running time is $\bigO(n^3 k^2 2^k i_\times(nk + n\log(n)))$.
\end{proof}

For the $\gamma$-clique covering problems, the situation is different.
We cannot count the number of $\gamma$-clique covers of all possible sizes, as the size of such a cover can be arbitrarily large.
Even if we restrict ourselves to counting covers that contain each clique at most once, then we need numbers with an exponential number of bits.
To see this, notice that the number of cliques in a graph of treewidth~$k$ is at most $\bigOs(2^k)$ since there are at most $\bigO(2^k)$ different cliques in each bag.
Hence, there are at most $2^{\bigOs(2^k)}$ different clique covers, and these can be counted using only $\bigOs(2^k)$-bit numbers.
Therefore, we will restrict ourselves to counting covers of size at most~$n$ because minimum covers will never be larger than~$n$.

A second difference is that, in a forget node, we now need to consider covering the forgotten vertex multiple times.
This requires a slightly different approach.
\begin{theorem} \label{thrm:cliquecovertwalg}
Let $\gamma \subseteq \N \setminus \{0\}$ be polynomial-time decidable.
There is an algorithm that, given a tree decomposition of a graph~$G$ of width~$k$, computes the size and number of minimum $\gamma$-clique covers of $G$ in $\bigO(n^3 \log(k) 2^k i_\times(nk + n\log(n)))$ time.
\end{theorem}
\begin{proof}
The dynamic programming algorithm for counting the number of minimum $\gamma$-clique covers is similar to the algorithm of Theorem~\ref{thrm:cliqueparttwalg}.
It uses the same tables $A_x$ for every $x \in T$ with entries $A_x(c,\kappa)$ for all $c \in \{0,1\}^{|X_x|}$ and $0 \leq \kappa \leq n$.
And, the computations of these tables in a leaf or introduce node of~$T$ are the same.

\smallskip \noindent {\it Join node}:
Let~$x$ be a join node in~$T$ and let~$l$ and~$r$ be its child nodes.
The join operation is different from the join operation in the algorithm of Theorem~\ref{thrm:cliqueparttwalg} as can be seen from Figure~\ref{fig:jointableclique}.
Here, the join operation is similar to our method of handling the~$0_1$-states and~$0_0$-states for the {\sc Dominating Set} problem in the algorithm of Theorem~\ref{thrm:countingtwdsalg}.
We simply transform the states in~$A_l$ and~$A_r$ to $\{0,?\}$ and compute~$A_x$ using these same states by summing over identical entries with different size parameters:
\[ A_x(c,\kappa) = \sum_{\kappa_l + \kappa_r = \kappa} A_l(c,\kappa_l) \cdot A_r(c,\kappa_r) \]
Then, we obtain the required result by transforming~$A_x$ back to using the set of states $\{1,0\}$.
We omit further details because this works analogously to Theorem~\ref{thrm:countingtwdsalg}.
The only difference with before is that we use the value zero for any $A_l(c,\kappa_l)$ or $A_r(c,\kappa_r)$ with $\kappa_l, \kappa_r < 0$ or $\kappa_l, \kappa_r > n$ as these never contribute to minimum clique covers.

\smallskip \noindent {\it Forget node}:
Let~$x$ be a forget node in~$T$ with child node~$y$ forgetting the vertex~$v$.
In contrast to Theorem~\ref{thrm:cliqueparttwalg}, we now have to consider covering~$v$ with multiple cliques.
In a minimum cover, $v$ can be covered at most~$k$ times because there are no vertices in~$X_x$ left to cover after using~$k$ cliques from~$X_x$.

Let~$A_M$ be as in Theorem~\ref{thrm:cliqueparttwalg}.
What we need is a table that contains more than just all cliques that can be used to cover~$v$: it needs to count all combinations of cliques that we can pick to cover~$v$ at most~$k$ times indexed by the number of cliques used.
To create this new table, we let $A_M^0 = A_M$, and let~$A_M^j$ be the result of the join operation applied to the table~$A_M^{j-1}$ with itself.
Then, the table~$A_M^j$ counts all ways of picking a series of~$2^j$ sets $C_1,C_2,\ldots,C_{2^j}$, where each set is either the empty set or a clique from~$M$.
To see that this holds, compare the definition of the join operation for this problem to the result of executing these operations repeatedly.
The algorithm computes $A_M^{\lceil\log(k)\rceil}$.
Because we want to know the number of clique covers that we can choose, and not the number of series of $2^{\lceil\log(k)\rceil}$ sets $C_1,C_2,\ldots,C_{2^{\lceil\log(k)\rceil}}$, we have to compensate for the fact that most covers are counted more than once.
Clearly, each cover consisting of~$\kappa$ cliques corresponds to a series in which $2^{\lceil\log(k)\rceil} - \kappa$ empty sets are picked: there are $\binom{2^{\lceil\log(k)\rceil}}{\kappa}$ possibilities of picking the empty sets and~$\kappa!$ permutations of picking each of the~$\kappa$ cliques in any order.
Hence, we divide each entry $A_M(c,\kappa)$ by $\kappa! \, \binom{2^{\lceil\log(k)\rceil}}{\kappa}$.
Now, $A_M^{\lceil\log(k)\rceil}$ contains the numbers we need for a join with~$A_y$.

After performing the join operation with~$A_y$ and $A_M^{\lceil\log(k)\rceil}$ obtaining a new table~$A_y$, we select the entries of~$A_y$ that cover~$v$: $A_x(c,\kappa) = A_y(c \times \{1\},\kappa)$.

\smallskip
If we have computed~$A_z$ for the root node~$z$ of~$T$, the size of the minimum $\gamma$-clique cover equals the smallest~$\kappa$ for which $A_z(\emptyset,\kappa)$ is non-zero, and this entry contains the number of such sets.

\smallskip
For the running time, we find that in order to compute $A_M^{\lceil\log(k)\rceil}$, we need $\bigO(\log(k))$ join operations.
The running time then follows from the same analysis as in Theorem~\ref{thrm:cliqueparttwalg}.
\end{proof}

Similar to previous results, we can improve the polynomial factors involved.
\begin{corollary} \label{cor:cliqueparttwalg}
Let $\gamma \subseteq \N \setminus \{0\}$ be polynomial-time decidable.
There are algorithms that, given a tree decomposition of a graph~$G$ of width~$k$:
\begin{enumerate}
\item decide whether there exists a $\gamma$-clique partition of~$G$ in $\bigO(n k^2 2^k)$ time.
\item count the number of $\gamma$-clique packings in~$G$ or the number of $\gamma$-clique partitionings in~$G$ in $\bigO(n k^2 2^k i_\times(nk + n\log(n)))$ time.
\item compute the size of a maximum $\gamma$-clique packing in~$G$, maximum $\gamma$-clique partitioning in~$G$, or minimum $\gamma$-clique partitioning in~$G$ of a problem with the de Fluiter property for treewidth in $\bigO(n k^4 2^k)$ time.
\item compute the size of a minimum $\gamma$-clique cover in~$G$ of a problem with the de Fluiter property for treewidth in $\bigO(n k^2 \log(k) 2^k)$ time, or in in $\bigO(n k^2 2^k)$ time if $|\gamma|$ is a constant.
\item compute the number of maximum $\gamma$-clique packings in~$G$, maximum $\gamma$-clique partitionings in~$G$, or minimum $\gamma$-clique partitionings in~$G$ of a problem with the de Fluiter property for treewidth in $\bigO(n k^4 2^k i_\times(nk + n\log(n)))$ time.
\item compute the number of minimum $\gamma$-clique covers in~$G$ of a problem with the de Fluiter property for treewidth in $\bigO(n k^2 \log(k) 2^k i_\times(nk + n\log(n)))$ time, or in $\bigO(n k^2 2^k i_\times(nk + n\log(n)))$ time if $|\gamma|$ is a constant.
\end{enumerate}
\end{corollary}
\begin{proof}[Proof (Sketch)]
Similar to before.
Either use the de Fluiter property to replace a factor~$n^2$ by~$k^2$, or omit the size parameter to completely remove this factor $n^2$ if possible.
Moreover, we can use $\bigO(k)$-bit numbers instead of $\bigO(nk + n\log(n))$-bit numbers if we are not counting the number of solutions.
In this case, we omit the time required for the arithmetic operations because of the computational model that we use with $\bigO(k)$-bit word size.
For the $\gamma$-clique cover problems where $|\gamma|$ is a constant, we note that we can use $A_M^p$ for some constant $p$ because, in a forget node, we need only a constant number of repetitions of the join operation on $A_M^0$ instead of $\log(k)$ repetitions.
\end{proof}

By this result, {\sc Partition Into Triangles} can be solved in $\bigO(n k^2 2^k)$ time.
For this problem, Lokshtanov et al.~proved that the given exponential factor in the running time is optimal, unless the Strong Exponential-Time Hypothesis fails~\cite{LokshtanovMS10}.

We note that in Corollary~\ref{cor:cliqueparttwalg} the problem of deciding whether there exists a $\gamma$-clique cover is omitted.
This is because this problem can easily be solved without dynamic programming on the tree decomposition by considering each vertex and testing whether it is contained in a clique whose size is a member of~$\gamma$.
This requires $\bigO(nk2^k)$ time in general, and polynomial time if $|\gamma|$ is a constant.

\section{Dynamic Programming on Branch Decompositions} \label{sec:branchwidth}
Dynamic programming algorithms on branch decompositions work similar to those on tree decompositions.
The tree is traversed in a bottom-up manner while computing tables $A_e$ with partial solutions on $G_e$ for every \emph{edge} $e$ of $T$ (see the definitions in Section~\ref{sec:defbw}).
Again, the table $A_e$ contains partial solutions of each possible \emph{characteristic}, where two partial solutions $P_1$ and $P_2$ have the same characteristic if any extension of $P_1$ to a solution on $G$ also is an extension of $P_2$ to a solution on $G$.
After computing a table for every edge $e \in E(T)$, we find a solution for the problem on $G$ in the single entry of the table $A_{\{y,z\}}$, where $z$ is the root of $T$ and $y$ is its only child node.
Because the size of the tables is often (at least) exponential in the branchwidth~$k$, such an algorithm typically runs in $\bigO(f(k)poly(n))$ time, for some function $f$ that grows at least exponentially.
See Proposition~\ref{prop:simplebwdsalg} for an example algorithm.

In this section, we improve the exponential part of the running time for many dynamic programming algorithms on branch decompositions.
A difference to our results on tree decompositions is that when the number of partial solutions stored in a table is $\bigOs(s^k)$, then our algorithms will run in $\bigOs(s^{\frac{\omega}{2}k})$ time.
This difference in the running time is due to the fact that the structure of a branch decomposition is different to the structure of a tree decomposition.
A tree decomposition can be transformed into a nice tree decomposition, such that every join node $x$ with children $l$, $r$ has $X_x = X_r = X_l$.
But a branch decomposition does not have such a property: here we need to consider combining partial solutions from both tables of the child edges while forgetting and introducing new vertices at the same time.

This section is organised as follows.
We start by setting up the framework that we use for dynamic programming on branch decompositions by giving a simple algorithm in Section~\ref{sec:bwframework}.
Hereafter, we give our results on {\sc Dominating Set} in Section~\ref{sec:bwds}, our results on counting perfect matchings in Section~\ref{sec:bwcpm}, and our results on the $[\rho,\sigma]$-domination problems in Section~\ref{sec:bwrhosigma}.

\subsection{General Framework on Branch Decompositions} \label{sec:bwframework}
We will first give a simple dynamic programming algorithm for the {\sc Dominating Set} problem.
This algorithm follows from standard techniques on branchwidth-based algorithms and will be improved later.

\begin{proposition} \label{prop:simplebwdsalg}
There is an algorithm that, given a branch decomposition of a graph~$G$ of width~$k$, counts the number of dominating sets in~$G$ of each size~$\kappa$, $0 \leq \kappa \leq n$, in $\bigO(m n^2 6^k i_\times(n))$ time.
\end{proposition}
\begin{proof}
Let~$T$ be a branch decomposition of~$G$ rooted at a vertex~$z$.
For each edge $e \in E(T)$, we will compute a table~$A_e$ with entries $A_e(c,\kappa)$ for all $c \in \{1,0_1,0_0\}^{X_e}$ and all $0 \leq \kappa \leq n$.
Here, $c$ is a colouring with states~$1$, $0_1$, and~$0_0$ that have the same meaning as in the tree-decomposition-based algorithms: see Table~\ref{tab:dsstates}.
In the table~$A_e$, an entry $A_e(c,\kappa)$ equals the number of partial solutions of {\sc Dominating Set} of size~$\kappa$ in~$G_e$ that satisfy the requirements defined by the colouring~$c$ on the vertices in~$X_e$.
That is, the number of vertex sets $D \subseteq V_e$ of size~$\kappa$ that dominate all vertices in~$V_e$ except for those with state~$0_0$ in colouring~$c$ of~$X_e$, and that contain all vertices in~$X_e$ with state~$1$ in~$c$.

The described tables~$A_e$ are computed by traversing the decomposition tree~$T$ in a bottom-up manner.
A branch decompositions has only two kinds of edges for which we need to compute such a table: leaf edges, and internal edges which have two child edges.

\smallskip \noindent {\it Leaf edges}:
Let~$e$ be an edge of~$T$ incident to a leaf of~$T$ that is not the root.
Then, $G_e = G[X_e]$ is a two-vertex graph with $X_e = \{u,v\}$.
Note that $\{u,v\} \in E$.

We compute~$A_e$ in the following way:
\[ A_e(c,\kappa) = \left\{ \begin{array}{ll} 1 & \textrm{if $\kappa = 2$ and $c=(1,1)$} \\ 1 & \textrm{if $\kappa = 1$ and either $c = (1,0_1)$ or $c = (0_1,1)$} \\ 1 & \textrm{if $\kappa = 0$ and $c=(0_0,0_0)$} \\ 0 & \textrm{otherwise} \end{array} \right. \]
The entries in this table are zero unless the colouring~$c$ represents one of the four possible partial solutions of {\sc Dominating Set} on~$G_e$ and the size of this solution is~$\kappa$.
In these non-zero entries, the single partial solution represented by~$c$ is counted.

\smallskip \noindent {\it Internal edges}:
Let~$e$ be an internal edge of~$T$ with child edges~$l$ and~$r$.
Recall the definition of the sets~$I$, $L$, $R$, $F$ induced by~$X_e$, $X_l$, and~$X_r$ (Definition~\ref{def:middlesets}).

Given a colouring~$c$, let $c(I)$ denote the colouring of the vertices of~$I$ induced by~$c$. 
We define $c(L)$, $c(R)$, and $c(F)$ in the same way.
Given a colouring~$c_e$ of~$X_e$, a colouring~$c_l$ of $X_l$, and a colouring~$c_r$ of~$X_r$, we say that these colourings \emph{match} if they correspond to a correct combination of two partial solutions with the colourings~$c_l$ and~$c_r$ on~$X_l$ and~$X_r$ which result is a partial solution that corresponds to the colouring~$c_e$ on~$X_e$.
For a vertex in each of the four partitions~$I$, $L$, $R$, and~$F$ of $X_e \cup X_l \cup X_r$, this means something different:
\begin{itemize}
\item For any $v \in I$: either $c_e(v) = c_l(v) = c_r(v) \in \{1,0_0\}$, or $c_e(v) = 0_1$ while $c_l(v), c_r(v) \in \{0_0,0_1\}$ and not $c_l(v) = c_r(v) = 0_0$. (5 possibilities)
\item For any $v \in F$: either $c_l(v) = c_r(v) = 1$, or $c_l(v), c_r(v) \in \{0_0,0_1\}$ while not $c_l(v) = c_r(v) = 0_0$. (4 possibilities)
\item For any $v \in L$: $c_e(v) = c_l(v) \in \{1,0_1,0_0\}$. (3 possibilities)
\item For any $v \in R$: $c_e(v) = c_r(v) \in \{1,0_1,0_0\}$. (3 possibilities)
\end{itemize}
That is, for vertices in~$L$ or~$R$, the properties defined by the colourings are copied from~$A_l$ and~$A_r$ to~$A_e$.
For vertices in~$I$, the properties defined by the colouring~$c_e$ is a combination of the properties defined by~$c_l$ and~$c_r$ in the same way as it is for tree decompositions (as in Proposition~\ref{prop:simpletwdsalg}).
For vertices in~$F$, the properties defined by the colourings are such that they form correct combinations in which the vertices may be forgotten, i.e., such a vertex is in the vertex set of both partial solutions, or it is not in the vertex set of both partial solutions while it is dominated.

Let $\kappa_{\#_1} = \#_{1}(c_r(I \cup F))$ be the number of vertices that are assigned state~$1$ on $I \cup F$ in any matching triple~$c_e$, $c_l$, $c_r$.
We can count the number of partial solutions on~$G_e$ satisfying the requirements defined by each colouring~$c_e$ on~$X_e$ using the following formula:
\[ A_e(c_e,\kappa) = \sum_{c_e, c_l, c_r \,\textrm{\scriptsize match}} \; \sum_{\kappa_l + \kappa_r = \kappa + \kappa_{\#_1}} A_l(c_l,\kappa_l) \cdot A_r(c_r,\kappa_r) \]
Notice that this formula correctly counts all possible partial solutions on~$G_e$ per corresponding colouring~$c_e$ on~$X_e$ by counting all valid combinations of partial solutions on~$G_l$ corresponding to a colouring~$c_l$ on~$X_l$ and partial solutions~$G_r$ corresponding to a colouring~$c_r$ on~$X_r$.

\smallskip
Let $\{y,z\}$ be the edge incident to the root~$z$ of~$T$.
From the definition of $A_{\{y,z\}}$, $A_{\{y,z\}}(\emptyset,\kappa)$ contains the number of dominating sets of size~$\kappa$ in $G_{\{y,z\}} = G$.

\smallskip
For the running time, we observe that we can compute $A_{e}$ in $\bigO(n)$ time for all leaf edges~$e$ of~$T$.
For the internal edges, we have to compute the $\bigO(n 3^k)$ values of~$A_e$, each of which requires $\bigO(n)$ terms of the above sum per set of matchings states.
Since each vertex in~$I$ has 5 possible matching states, each vertex in~$F$ has 4 possible matching states, and each vertex in~$L$ or~$R$ has 3 possible matching states, we compute each~$A_e$ in $\bigO(n^25^{|I|}4^{|F|}3^{|L|+|R|} i_\times(n))$ time.

Under the constraint that $|I|+|L|+|R|, |I|+|L|+|F|, |I|+|R|+|F| \leq k$, the running time is maximal if $|I| = 0$, $|L|=|R|=|F|=\frac{1}{2}k$.
As~$T$ has $\bigO(m)$ edges and we work with $n$-bit numbers, this leads to a running time of $\bigO(m n^2 4^{\frac{1}{2}k} 3^k i_\times(n)) = \bigO(m n^2 6^k i_\times(n))$.
\end{proof}

The above algorithm gives the framework that we use in all of our dynamic programming algorithms on branch decompositions.
In later algorithms, we will specify only how to compute the tables $A_e$ for both kinds of edges.

\paragraph{De Fluiter Propery for Branchwidth.}
We conclude this subsection with a discussion on a de Fluiter property for branchwidth.
We will see below that such a property for branchwidth is identical to the de Fluiter property for treewidth.

One could define a de Fluiter property for branchwidth by replacing the words treewidth and tree decomposition in Definition~\ref{def:fluiterproptw} by branchwidth and branch decomposition.
However, the result would be a property equivalent to the de Fluiter property for treewidth.
This is not hard to see, namely, consider any edge $e$ of a branch decomposition with middle set $X_e$.
A representation of the different characteristics on $X_e$ of partial solutions on $G_e$ used on branch decompositions can also be used as a representation of the different characteristics on $X_x$ of partial solutions on $G_x$ on tree decompositions, if $X_e = X_x$ and $G_e = G_x$.
Clearly, an extension of a partial solution on $G_e$ with some characteristic is equivalent to an extension of the same partial solution on $G_x = G_e$, and hence the representations can be used on both decompositions.
This equivalence of both de Fluiter properties follows directly.

As a result, we will use the de Fluiter property for treewidth in this section.
In Section~\ref{sec:cliquewidth}, we will also define a de Fluiter property for cliquewidth; this property will be different from the other two.

\subsection{Minimum Dominating Set} \label{sec:bwds}
We start by improving the above algorithm of Proposition~\ref{prop:simplebwdsalg}.
This improvement will be presented in two steps.
First, we use state changes similar to what we did in Section~\ref{sec:dstw}.
Thereafter, we will further improve the result by using fast matrix multiplication in the same way as proposed by Dorn in~\cite{Dorn06}.
As a result, we will obtain an $\bigOs(3^{\frac{\omega}{2}k})$ algorithm for {\sc Dominating Set} for graphs given with \emph{any} branch decomposition of width $k$.

Similar to tree decompositions, it is more efficient to transform the problem to one using states $1$, $0_0$, and $0_?$ if we want to combine partial solutions from different dynamic programming tables.
However, there is a big difference between dynamic programming on tree decompositions and on branch decompositions.
On tree decompositions, we can deal with forget vertices separately, while this is not possible on branch decompositions.
This makes the situation more complicated.
On branch decompositions, vertices in $F$ must be dealt with simultaneously with the computation of $A_e$ from the two tables $A_l$ and $A_r$ for the child edges $l$ and $r$ of $e$.
We will overcome this problem by using different sets of states simultaneously.
The set of states used depends on whether a vertex is in $L$, $R$, $I$ or $F$. 
Moreover, we do this asymmetrically as different states can be used on the same vertices in a different table $A_e$, $A_l$, $A_r$.
This use of \emph{asymmetrical vertex states} will later allow us to easily combine the use of state changes with fast matrix multiplication and obtain significant improvements in the running time.

We state this use of different states on different vertices formally.
We note that this construction has already been used in the proof of Theorem~\ref{thrm:rstwalg}.
\begin{lemma} \label{lem:asymdsstates}
Let~$e$ be an edge of a branch decomposition~$T$ with corresponding middle set~$|X_e|$, and let~$A_e$ be a table with entries $A_e(c,\kappa)$ representing the number of partial solutions of {\sc Dominating Set} in~$G_e$ of each size~$\kappa$, for some range of~$\kappa$, corresponding to all colourings of the middle set~$X_e$ with states such that for every individual vertex in~$X_e$ one of the following fixed sets of states is used:
\[ \{1,0_1,0_0\} \qquad \{1,0_1,0_?\} \qquad \{1,0_0,0_?\} \qquad \textrm{(see Table~\ref{tab:dsstates})} \]
The information represented in the table~$A_e$ does not depend on the choice of the set of states from the options given above.
Moreover, there exist transformations between tables using representations with different sets of states on each vertex using $\bigO(|X_x||A_x|)$ arithmetic operations.
\end{lemma}
\begin{proof}
Use the same $|X_e|$-step transformation as in the proof of Lemma~\ref{lem:dsstates} with the difference that we can choose a different formula to change the states at each coordinate of the colouring~$c$ of~$X_e$.
At coordinate~$i$ of the colouring~$c$, we use the formula that corresponds to the set of states that we want to use on the corresponding vertex in~$X_e$.
\end{proof}

We are now ready to give the first improvement of Proposition~\ref{prop:simplebwdsalg}.
\begin{proposition} \label{prop:secondbwdsalg}
There is an algorithm that, given a branch decomposition of a graph~$G$ of width~$k$, counts the number of dominating sets in~$G$ of each size $\kappa$, $0 \leq \kappa \leq n$, in $\bigO(m n^2 3^{\frac{3}{2}k} i_\times(n))$ time.
\end{proposition}
\begin{proof}
The algorithm is similar to the algorithm of Proposition~\ref{prop:simplebwdsalg}, only we employ a different method to compute~$A_e$ for an internal edge~$e$ of~$T$.

\smallskip \noindent {\it Internal edges}:
Let~$e$ be an internal edge of~$T$ with child edges~$l$ and~$r$.

We start by applying Lemma~\ref{lem:asymdsstates} to~$A_l$ and~$A_r$ and change the sets of states used for each individual vertex in the following way.
We let~$A_l$ use the set of states $\{1,0_?,0_0\}$ on vertices in~$I$ and the set of states $\{1,0_1,0_0\}$ on vertices in~$L$ and~$F$.
We let~$A_r$ use the set of states $\{1,0_?,0_0\}$ on vertices in~$I$, the set of states $\{1,0_1,0_0\}$ on vertices in~$R$, and the set of states $\{1,0_1,0_?\}$ on vertices in~$F$.
Finally, we let~$A_e$ use the set of states $\{1,0_?,0_0\}$ on vertices in~$I$ and the set of states $\{1,0_1,0_0\}$ on vertices in~$L$ and~$R$.
Notice that different colourings use the same sets of states on the same vertices with the exception of the set of states used for vertices in~$F$; here, $A_l$~and~$A_r$ use different sets of states.

Now, three colourings~$c_e$, $c_l$~and~$c_r$ \emph{match} if:
\begin{itemize}
\item For any $v \in I$: $c_e(v) = c_l(v) = c_r(v) \in \{1,0_?,0_0\}$. (3 possibilities)
\item For any $v \in F$: either $c_l(v) = c_r(v) = 1$, or $c_l(v) = 0_0$ and $c_r(v) = 0_1$, or $c_l(v) = 0_1$ and $c_r(v) = 0_?$. (3 possibilities)
\item For any $v \in L$: $c_e(v) = c_l(v) \in \{1,0_1,0_0\}$. (3 possibilities)
\item For any $v \in R$: $c_e(v) = c_r(v) \in \{1,0_1,0_0\}$. (3 possibilities)
\end{itemize}
For the vertices on~$I$, these matching combinations are the same as used on tree decompositions in Theorem~\ref{thrm:countingtwdsalg}, namely the combinations with states from the set $\{1,0_?,0_0\}$ where all states are the same.
For the vertices on~$L$ and~$R$, we do exactly the same as in the proof of Proposition~\ref{prop:simplebwdsalg}.

For the vertices in~$F$, a more complicated method has to be used.
Here, we can use only combinations that make sure that these vertices will be dominated: combinations with vertices that are in vertex set of the partial solution, or combinations in which the vertices are not in this vertex set, but in which they will be dominated.
Moreover, by using different states for~$A_l$ and~$A_r$, every combination of partial solutions is counted exactly once.
To see this, consider each of the three combinations on~$F$ used in Proposition~\ref{prop:simplebwdsalg} with the set of states $\{1,0_1,0_0\}$.
The combination with $c_l(v) = 0_0$ and $c_r(v) = 0_1$ is counted using the same combination, while the other two combinations ($c_l(v) = 0_1$ and $c_r(v) = 0_0$ or $c_r(v) = 0_1$) are counted when combining~$0_1$ with~$0_?$.

In this way, we can compute the entries in the table~$A_e$ using the following formula:
\[ A_e(c_e, \kappa) = \sum_{c_e, c_l, c_r \,\textrm{\scriptsize match}} \; \sum_{\kappa_l + \kappa_r = \kappa + \kappa_{\#_1}} A_l(c_l,\kappa_l) \cdot A_r(c_r, \kappa_r ) \]
Here, $\kappa_{\#_1} = \#_{1}(c_r(I \cup F))$ again is the number of vertices that are assigned state~$1$ on $I \cup F$ in any matching triple~$c_e$, $c_l$, $c_r$.

After having obtained~$A_e$ in this way, we can transform the set of states used back to $\{1,0_1,0_0\}$ using Lemma~\ref{lem:asymdsstates}.

\smallskip
For the running time, we observe that the combination of the different sets of states that we are using allows us to evaluate the above formula in $\bigO(n^2 3^{|I|+|L|+|R|+|F|} i_\times(n))$ time.
As each state transformation requires $\bigO(n 3^k i_+(n))$ time, the improved algorithm has a running time of $\bigO(m n^2 3^{|I|+|L|+|R|+|F|} i_\times(n))$.
Under the constraint that $|I|+|L|+|R|, |I|+|L|+|F|, |I|+|R|+|F| \leq k$, the running time is maximal if $|I| = 0$, $|L|=|R|=|F|=\frac{1}{2}k$.
This gives a total running time of $\bigO(m n^2 3^{\frac{3}{2}k} i_{\times}(n))$.
\end{proof}

We will now give our faster algorithm for counting the number of dominating sets of each size~$\kappa$, $0 \leq \kappa \leq n$ on branch decompositions.
This algorithm uses fast matrix multiplication to speed up the algorithm of Proposition~\ref{prop:secondbwdsalg}.
This use of fast matrix multiplication in dynamic programming algorithms on branch decompositions was first proposed by Dorn in~\cite{Dorn06}.

\begin{theorem} \label{thrm:countdsbwalg}
There is an algorithm that, given a branch decomposition of a graph~$G$ of width~$k$, counts the number of dominating sets in~$G$ of each size~$\kappa$, $0 \leq \kappa \leq n$, in $\bigO(m n^2 3^{\frac{\omega}{2}k} i_\times(n))$ time.
\end{theorem}
\begin{proof}
Consider the algorithm of Proposition~\ref{prop:secondbwdsalg}.
We will make one modification to this algorithm.
Namely, when computing the table~$A_e$ for an internal edge $e \in E(T)$, we will show how to evaluate the formula for $A_e(c,\kappa)$ for a number of colourings~$c$ simultaneously using fast matrix multiplication.
We give the details below.
Here, we assume that the states in the tables~$A_l$ and~$A_r$ are transformed such that the given formula for $A_e(c,\kappa)$ in Proposition~\ref{prop:secondbwdsalg} applies. 

We do the following.
First, we fix the two numbers~$\kappa$ and~$\kappa_l$, and we fix a colouring of~$I$.
Note that this is well-defined because all three tables use the same set of states for colours on~$I$.
Second, we construct a $3^{|L|} \times 3^{|F|}$ matrix~$M_l$ where each row corresponds to a colouring of~$L$ and each column corresponds to a colouring of~$F$ where both colourings use the states used by the corresponding vertices in~$A_l$.
We let the entries of~$M_l$ be the values of $A_l(c_l,\kappa_l)$ for the~$c_l$ corresponding to the colourings of~$L$ and~$F$ of the given row and column of~$M_l$, and corresponding to the fixed colouring on~$I$ and the fixed number~$\kappa_l$.
We also construct a similar $3^{|F|}\times 3^{|R|}$ matrix~$M_r$ with entries from~$A_r$ such that its rows correspond to different colourings of~$F$ and its columns correspond to different colourings of~$R$ where both colourings use the states used by the corresponding vertices in~$A_r$.
The entries of~$M_r$ are the values of $A_r(c_r,\kappa - \kappa_l - \kappa_{\#_1})$ where~$c_r$ corresponds to the colouring of~$R$ and~$F$ of the given row and column of~$M_r$, and corresponding to the fixed colouring on~$I$ and the fixed numbers~$\kappa$ and~$\kappa_l$.
Here, the value of $\kappa_{\#_1} = \#_{1}(c_r(I \cup F))$ depends on the colouring~$c_r$ in the same way as in Proposition~\ref{prop:secondbwdsalg}.
Third, we permute the rows of~$M_r$ such that column~$i$ of~$M_l$ and row~$i$ of~$M_r$ correspond to matching colourings on~$F$.

Now, we can evaluate the formula for~$A_e$ for all entries corresponding to the fixed colouring on~$I$ and the fixed values of~$\kappa$ and~$\kappa_l$ simultaneously by computing $M_e = M_l \cdot M_r$.
Clearly, $M_e$ is a $3^{|L|} \times 3^{|R|}$ matrix where each row corresponds to a colouring of~$L$ and each column corresponds to a colouring of~$R$.
If one works out the matrix product $M_l \cdot M_r$, one can see that each entry of~$M_e$ contains the sum of the terms of the formula for $A_e(c_e,\kappa)$ such that the colouring~$c_e$ corresponds to the given row and column of~$M_e$ and the given fixed colouring on~$I$ and such that $\kappa_l + \kappa_r = \kappa + \kappa_{\#_1}$ corresponding to the fixed~$\kappa$ and~$\kappa_l$.
That is, each entry in~$M_e$ equals the sum over all possible allowed matching combinations of the colouring on~$F$ for the fixed values of~$\kappa$ and~$\kappa_l$, where the~$\kappa_r$ involved are adjusted such that the number of $1$-states used on~$F$ is taken into account.

In this way, we can compute the function~$A_e$ by repeating the above matrix-multiplication-based process for every colouring on~$I$ and every value of~$\kappa$ and~$\kappa_l$ in the range from~$0$ to~$n$.
As a result, we can compute the function~$A_e$ by a series of $n^2\,3^{|I|}$ matrix multiplications.

The time required to compute~$A_e$ in this way depends on~$|I|$, $|L|$, $|R|$ and~$|F|$.
Under the constraint that $|I| + |L| + |F|, |I|+|R|+|F|, |I|+|L|+|R| \leq k$ and using the matrix-multiplication algorithms for square and non-square matrices as described in Section~\ref{sec:matrixmultiplic}, the worst case arises when $|I|=0$ and $|L|=|R|=|F| = \frac{k}{2}$.
In this case, we compute each table~$A_e$ in $\bigO(n^2 (3^{\frac{k}{2}})^\omega i_{\times}(n))$ time.
This gives a total running time of $\bigO(m n^2 3^{\frac{\omega}{2}k} i_\times(n))$.
\end{proof}

Using the fact that {\sc Dominating Set} has the de Fluiter property for treewidth, and using the same tricks as in Corollaries~\ref{cor:countmdstwalg} and~\ref{cor:solvedstwalg}, we also obtain the following results.

\begin{corollary} \label{cor:countmdsbwalg}
There is an algorithm that, given a branch decomposition of a graph~$G$ of width~$k$, counts the number of minimum dominating sets in~$G$ in $\bigO(m k^2 3^{\frac{\omega}{2}k} i_\times(n))$ time.
\end{corollary}
\begin{corollary} \label{cor:solvedsbwalg}
There is an algorithm that, given a branch decomposition of a graph~$G$ of width~$k$, computes the size of a minimum dominating set in~$G$ in $\bigO(m k^2 3^{\frac{\omega}{2}k} i_\times(k))$ time.
\end{corollary}

\subsection{Counting the Number of Perfect Matchings} \label{sec:bwcpm}
The next problem we consider is the problem of counting the number of perfect matchings in a graph.
To give a fast algorithm for this problem, we use both the ideas introduced in Theorem~\ref{thrm:countingpmtwalg} to count the number of perfect matchings on graphs given with a tree decomposition and the idea of using fast matrix multiplication of Dorn~\cite{Dorn06} found in Theorem~\ref{thrm:countdsbwalg}.
We note that~\cite{Dorn06,Dorn07} did not consider counting perfect matchings.
The result will be an $\bigOs(2^{\frac{\omega}{2}k})$ algorithm.

From the algorithm to count the number of dominating sets of each given size in a graph of bounded branchwidth, it is clear that vertices in~$I$ and~$F$ need special attention when developing a dynamic programming algorithm over branch decompositions.
This is no different when we consider counting the number of perfect matchings.
For the vertices in~$I$, we will use state changes and an index similar to Theorem~\ref{thrm:countingpmtwalg}, but for the vertices in~$F$ we will require only that all these vertices are matched.
In contrast to the approach on tree decompositions, we will not take into account the fact that we can pick edges in the matching at the point in the algorithm where we forget the first endpoint of the edge.
We represent this choice directly in the tables~$A_e$ of the leaf edges~$e$: this is possible because every edge of~$G$ is uniquely assigned to a leaf of~$T$.

Our algorithm will again be based on state changes, where we will again use different sets of states on vertices with different roles in the computation.
\begin{lemma} \label{lem:asympmstates}
Let~$e$ be an edge of a branch decomposition~$T$ with corresponding middle set~$|X_e|$, and let~$A_e$ be a table with entries $A_e(c,\kappa)$ representing the number of matchings in~$H_e$ matching all vertices in $V_e \setminus X_e$ and corresponding to all colourings of the middle set~$X_e$ with states such that for every individual vertex in~$X_e$ one of the following fixed sets of states is used:
\[ \{1,0\} \qquad \{1,?\} \qquad \{0,?\} \qquad \textrm{(meaning of the states as in Lemma~\ref{lem:pmstates})} \] 
The information represented in the table~$A_e$ does not depend on the choice of the set of states from the options given above.
Moreover, there exist transformations between tables using representations with different sets of states on each vertex using $\bigO(|X_x||A_x|)$ arithmetic operations.
\end{lemma}
\begin{proof}
The proof is identical to that of Lemma~\ref{lem:asymdsstates} while using the formulas from the proof of Lemma~\ref{lem:pmstates}.
\end{proof}

Now, we are ready to proof the result for counting perfect matchings.
\begin{theorem} \label{thrm:countpmbwalg}
There is an algorithm that, given a branch decomposition of a graph $G$ of width~$k$, counts the number of perfect matchings in~$G$ in $\bigO(m k^2 2^{\frac{\omega}{2}k} i_{\times}(k\log(n)))$ time.
\end{theorem}
\begin{proof}
Let~$T$ be a branch decomposition of~$G$ of branchwidth~$k$ rooted at a vertex~$z$.

For each edge $e \in E(T)$, we will compute a table~$A_e$ with entries $A_e(c)$ for all $c \in \{1,0\}^{X_e}$ where the states have the same meaning as in Theorem~\ref{thrm:countingpmtwalg}.
In this table, an entry $A_e(c)$ equals the number of matchings in the graph~$H_e$ matching all vertices in $V_e \setminus X_e$ and satisfying the requirements defined by the colouring~$c$ on the vertices in~$X_e$.
These entries do not count matchings in~$G_e$ but in its subgraph~$H_e$ that has the same vertices as~$G_e$ but contains only the edges of~$G_e$ that are in the leaves below~$e$ in~$T$.

\smallskip \noindent {\it Leaf edges}:
Let~$e$ be an edge of~$T$ incident to a leaf of~$T$ that is not the root.
Now, $H_e = G_e = G[X_e]$ is a two vertex graph with $X_e = \{u,v\}$ and with an edge between~$u$ and~$v$.

We compute~$A_e$ in the following way:
\[ A_e(c) = \left\{ \begin{array}{ll} 1 & \textrm{if $c=(1,1)$ or $c=(0,0)$} \\ 0 & \textrm{otherwise} \end{array} \right. \]
The only non-zero entries are the empty matching and the matching consisting of the unique edge in~$H_e$.
This is clearly correct.

\smallskip \noindent {\it Internal edges}:
Let~$e$ be an internal edge of~$T$ with child edges~$l$ and~$r$.

Similar to the proof of Theorem~\ref{thrm:countingpmtwalg}, we start by indexing the tables~$A_l$ and~$A_r$ by the number of 1-states used for later use.
However, we now count only the number of 1-states used on vertices in~$I$ in the index.
We compute indexed tables~$A'_l$ and~$A'_r$ with entries $A'_l(c_l,i_l)$ and $A'_r(c_r,i_r)$ using the following formula with $y \in \{l,r\}$:
\[ A'_y(c_y,i_y) = \left\{ \begin{array}{ll} A_y(c_y) & \textrm{if $\#_{1}(c_y(I))) = i_y$} \\ 0 & \textrm{otherwise} \end{array} \right. \]
Here, $\#_{1}(c_y(I))$ is the number of $1$-entries in the colouring~$c_y$ on the vertices in~$I$.

Next, we apply state changes by using Lemma~\ref{lem:asympmstates}.
In this case, we change the states used for the colourings in~$A'_r$ and~$A'_l$ such that they use the set of states $\{0,1\}$ on~$L$, $R$, and~$F$, and the set of states $\{0,?\}$ on~$I$.
Notice that the number of $1$-states used to create the $?$-states is now stored in the index~$i_l$ of $A'_l(c,i_l)$ and~$i_r$ of $A'_r(c,i_r)$.

We say that three colourings~$c_e$ of~$X_e$, $c_l$ of~$X_l$ and~$c_r$ of~$X_r$ using these sets of states on the different partitions of $X_e \cup X_l \cup X_r$ \emph{match} if:
\begin{itemize}
\item For any $v \in I$: $c_e(v) = c_l(v) = c_r(v) \in \{0,?\}$. (2 possibilities)
\item For any $v \in F$: either $c_l(v) = 0$ and $c_r(v) = 1$, or $c_l(v) = 1$ and $c_r(v) = 0$. (2~possibilities)
\item For any $v \in L$: $c_e(v) = c_l(v) \in \{1,0\}$. (2 possibilities)
\item For any $v \in R$: $c_e(v) = c_r(v) \in \{1,0\}$. (2 possibilities)
\end{itemize}
Now, we can compute the indexed table~$A'_e$ for the edge~$e$ of~$T$ using the following formula:
\[ A'_e(c_e, i_e) = \sum_{c_e, c_l, c_r \,\textrm{\scriptsize match}} \; \sum_{i_l + i_r = i_e} A'_l(c_l,i_l) \cdot A'_r(c_r, i_r ) \]
Notice that we can compute~$A'_e$ efficiently by using a series of matrix multiplications in the same way as done in the proof of Theorem~\ref{thrm:countdsbwalg}.
However, the index~$i$ should be treated slightly differently from the parameter~$\kappa$ in the proof of Theorem~\ref{thrm:countdsbwalg}.
After fixing a colouring on~$I$ and the two values of~$i_e$ and~$i_l$, we still create the two matrices~$M_l$ and~$M_r$.
In~$M_l$ each row corresponds to a colouring of~$L$ and each column corresponds to a colouring of~$F$, and in~$M_r$ each rows again corresponds to a colourings of~$F$ and each column corresponds to a colouring of~$R$.
The difference is that we fill~$M_l$ with the corresponding entries $A'_l(c_l,i_l)$ and~$M_r$ with the corresponding entries $A'_r(c_r,i_r)$.
That is, we do not adjust the value of~$i_r$ for the selected $A'_r(c_r,i_r)$ depending on the states used on~$F$.
This is not necessary here since the index counts the total number of $1$-states hidden in the $?$-states and no double counting can take place.
This in contrast to the parameter~$\kappa$ in the proof of Theorem~\ref{thrm:countdsbwalg}; this parameter counted the number of vertices in a solution, which we had to correct to avoid double counting of the vertices in~$F$.

After computing~$A'_e$ in this way, we again change the states such that the set of states $\{1,0\}$ is used on all vertices in the colourings used in~$A'_e$.
We then extract the values of~$A'_e$ in which no two $1$-states hidden in a $?$-state are combined to a new $1$-state on a vertex in~$I$.
We do so using the indices in the same way as in Theorem~\ref{thrm:countingpmtwalg} but with the counting restricted to~$I$:
\[ A_e(c) = A'_e(c,\#_1(c(I))) \]

\smallskip
After computing the~$A_e$ for all $e \in E(T)$, we can find the number of perfect matchings in $G = G_{\{y,z\}}$ in the single entry in $A_{\{y,z\}}$ where~$z$ is the root of~$T$ and~$y$ is its only child.

\smallskip
Because the treewidth and branchwidth of a graph differ by at most a factor $\frac{3}{2}$ (see Proposition~\ref{prop:1.5}), we can conclude that the computations can be done using $\bigO(k\log(n))$-bit numbers using the same reasoning as in the proof of Theorem~\ref{thrm:countingpmtwalg}.
For the running time, we observe that we can compute each~$A_e$ using a series of $k^22^{|I|}$ matrix multiplications.
The worst case arises when $|I|=0$ and $|L|=|R|=|F|=\frac{k}{2}$.
Then the matrix multiplications require $\bigO(k^2 2^{\frac{\omega}{2}k})$ time.
Since~$T$ has $\bigO(m)$ edges, this gives a running time of $\bigO(m k^2 2^{\frac{\omega}{2}k} i_{\times}(k\log(n)))$ time.
\end{proof}

\subsection{$[\rho,\sigma]$-Domination Problems} \label{sec:bwrhosigma}
We have shown how to solve two fundamental graph problems in $\bigOs(s^{\frac{\omega}{2}k})$ time on branch decompositions of width~$k$, where~$s$ is the natural number of states involved in a dynamic programming algorithm on branch decompositions for these problem.
Similar to the results on tree decompositions, we generalize this and show that one can solve all $[\rho,\sigma]$-domination problems with finite or cofinite~$\rho$ and~$\sigma$ in $\bigOs(s^{\frac{\omega}{2}k})$ time.

For the $[\rho,\sigma]$-domination problems, we use states~$\rho_j$ and~$\sigma_j$, where~$\rho_j$ and~$\sigma_j$ represent that a vertex is not in or in the vertex set~$D$ of the partial solution of the $[\sigma,\rho]$-domination problem, respectively, and has~$j$ neighbours in~$D$.
This is similar to Section~\ref{sec:rhosigmatw}.
Note that the number of states used equals $s = p + q + 2$.

On branch decompositions, we have to use a different approach than on tree decompositions, since we have to deal with vertices in~$L$, $R$, $I$, and~$F$ simultaneously.
It is, however, possible to reuse part of the algorithm of Theorem~\ref{thrm:rstwalg}.
Observe that joining two children in a tree decomposition is similar to joining two children in a branch decomposition if $L=R=F=\emptyset$.
Since we have demonstrated in the algorithms earlier in this section that one can have distinct states and perform different computations on~$I$, $L$, $R$, and~$F$, we can essentially use the approach of Theorem~\ref{thrm:rstwalg} for the vertices in~$I$.

\begin{theorem} \label{thrm:rsbwalg}
Let $\rho,\sigma \subseteq \N$ be finite or cofinite.
There is an algorithm that, given a branch decomposition of a graph~$G$ of width~$k$, counts the number of
$[\rho,\sigma]$-dominating sets of~$G$ of each size~$\kappa$, $0 \leq \kappa \leq n$, of a fixed $[\rho,\sigma]$-domination problem involving~$s$ states in $\bigO(m n^2 (sk)^{2(s-2)} s^{\frac{\omega}{2} k} i_\times(n))$ time.
\end{theorem}
\begin{proof}
Let~$T$ be the branch decomposition of~$G$ of width~$k$ rooted at the vertex~$z$.

Recall the definitions of State Sets~I and~II defined in Definition~\ref{def:rsstates}.
Similar to the proof of Theorem~\ref{thrm:rstwalg}, we will use different sets of states to prove this theorem.
In this proof, we mostly use State Set~I while we let the subscripts of the states count only neighbours in~$D$ outside the current middle set.
That is, we use states~$\rho_j$ and~$\sigma_j$ for our tables~$A_e$, $A_f$, and~$A_g$ such that the subscripts~$j$ represent the number of neighbours in the vertex set~$D$ of each partial solution of the $[\sigma,\rho]$-domination problem outside of the vertex sets~$X_e$, $X_f$ and~$X_g$, respectively.
Using these states for colourings~$c$, we compute the table~$A_e$ for each edge $e \in E(T)$ such that the entry $A_e(c,\kappa)$ contains the number of partial solutions of the $[\rho,\sigma]$-domination problem on~$G_e$ consisting of~$\kappa$ vertices that satisfy the requirements defined by~$c$.

\smallskip \noindent {\it Leaf edges}:
Let~$e$ be an edge of~$T$ incident to a leaf of~$T$ that is not the root.
Now, $G_e = G[X_e]$ is a two vertex graph.

We compute~$A_e$ in the following way:
\[ A_e(c, \kappa) = \left\{ \begin{array}{ll} 1 & \textrm{if $c=(\rho_0,\rho_0)$ and $\kappa = 0$} \\ 1 & \textrm{if $c=(\rho_0,\sigma_0)$ or $c=(\sigma_0,\rho_0)$, and $\kappa = 1$} \\ 1 & \textrm{if $c=(\sigma_0,\sigma_0)$ and $\kappa = 2$} \\ 0 & \textrm{otherwise} \end{array} \right. \]
Since the subscripts of the states count only vertices in the vertex set of a partial solutions of the $[\rho,\sigma]$-domination problem on~$G_e$ that are outside the middle set~$X_e$, we only count partial solutions in which the subscripts are zero.
Moreover, the size parameter~$\kappa$ must equal the number of $\sigma$-states since these represent vertices in the vertex set of the partial solutions.

\smallskip \noindent {\it Internal edges}:
Let~$e$ be an internal edge of~$T$ with child edges~$l$ and~$r$.

The process of computing the table~$A_e$ by combining the information in the two tables~$A_l$ and~$A_r$ is quite technical.
This is mainly due to the fact that we need to do different things on the different vertex sets~$I$, $L$, $R$, and~$F$.
We will give a three-step proof.

{\it Step 1}: As a form of preprocessing, we will update the entries in~$A_l$ and~$A_r$ such that the subscripts will not count only the vertices in vertex sets of the partial solutions outside of~$X_l$ and~$X_r$, but also some specific vertices in the vertex sets of the partial solutions in the middle sets.
Later, we will combine the information from~$A_l$ and~$A_r$ to create the table~$A_e$ according to the following general rule: combining~$\rho_i$ and~$\rho_j$ gives~$\rho_{i+j}$, and~$\sigma_i$ and~$\sigma_j$ gives~$\sigma_{i+j}$.
In this context, the preprocessing makes sure that the subscripts of the states in the result in~$A_e$ correctly count the number of vertices in the vertex sets of the partial solutions of the $[\rho,\sigma]$-domination problem.

Recall that for an edge~$e$ of the branch decomposition~$T$ the vertex set~$V_e$ is defined to be the vertex set of the graph~$G_e$, that is, the union of the middle set of~$e$ and all middle sets below~$e$ in~$T$.
We update the tables~$A_l$ and~$A_r$ such that the subscripts of the states~$\rho_j$ and~$\sigma_j$ count the number of neighbours in the vertex sets of the partial solutions with the following properties:
\begin{itemize}
\item States used in~$A_l$ on vertices in~$L$ or~$I$ count neighbours in $(V_l \setminus X_l) \cup F$.
\item States used in~$A_l$ on vertices in~$F$ count neighbours in $(V_l \setminus X_l) \cup L \cup I \cup F$.
\item States used in~$A_r$ on vertices in~$I$ count neighbours in $(V_r \setminus X_r)$ (nothing changes here).
\item States used in~$A_r$ on vertices in~$R$ count neighbours in $(V_r \setminus X_r) \cup F$.
\item States used in~$A_r$ on vertices in~$F$ count neighbours in $(V_r \setminus X_r) \cup R$.
\end{itemize}
If we now combine partial solutions with state~$\rho_i$ in~$A_l$ and state~$\rho_j$ in~$A_r$ for a vertex in~$I$, then the state~$\rho_{i+j}$ corresponding to the combined solution in~$A_e$ correctly counts the number of neighbours in the partial solution in $V_e \setminus X_e$.
Also, states for vertices in~$L$ and~$R$ in~$A_e$ count their neighbours in the partial solution in $V_e \setminus X_e$.
And, if we combine solutions with a state~$\rho_i$ in~$A_l$ and a state~$\rho_j$ in~$A_r$ for a vertex in~$F$, then this vertex will have exactly $i+j$ neighbours in the combined partial solution.

Although one must be careful which vertices to count and which not to count, the actual updating of the tables~$A_l$ and~$A_r$ is simple because one can see which of the counted vertices are in the vertex set of a partial solution ($\sigma$-state) and which are not ($\rho$-state).

Let~$A^*_y$ be the table before the updating process with $y \in \{l,r\}$.
We compute the updated table $A_y$ in the following way:
\[ A_y(c,\kappa) = \left\{ \begin{array}{ll} 0 & \textrm{if $\phi(c)$ is not a correct colouring of $X_y$} \\ A^*_y(\phi(c),\kappa) & \textrm{otherwise} \end{array} \right. \]
Here, $\phi$ is the inverse of the function that updates the subscripts of the states, e.g., if $y = l$ and we consider a vertex in~$I$ with exactly one neighbour with a $\sigma$-state on a vertex in~$F$ in~$c$, then it changes~$\rho_2$ into~$\rho_1$.
The result of this updating is not a correct colouring of~$X_y$ if the inverse does not exist, i.e., if the strict application of subtracting the right number of neighbours results in a negative number.
For example, this happens if~$c$ contains a~$\rho_0$- or~$\sigma_0$-state while it has neighbours that should be counted in the subscripts.

{\it Step 2}: Next, we will change the states used for the tables~$A_l$ and~$A_r$, and we will add index vectors to these tables that allow us to use the ideas of Theorem~\ref{thrm:rstwalg} on the vertices in~$I$.

We will not change the states for vertices in~$L$ in the table~$A_l$, nor for vertices in~$R$ in the table~$A_r$.
But, we will change the states for the vertices in~$I$ in both~$A_l$ and~$A_r$ and on the vertices in~$F$ in~$A_r$.
On~$F$, simple state changes suffice, while, for vertices on~$I$, we need to change the states and introduce index vectors at the same time.

We will start by changing the states for the vertices in~$F$.
On the vertices in~$F$, we will not change the states in~$A_l$, but introduce a new set of states to use for~$A_r$.
We define the states~$\bar{\rho_j}$ and~$\bar{\sigma_j}$.
A table entry with state~$\bar{\rho_j}$ on a vertex~$v$ requires that the vertex has an allowed number of neighbours in the vertex set of a partial solution when combined with a partial solution from~$A_l$ with state~$\rho_j$.
That is, a partial solution that corresponds to the state~$\rho_i$ on~$v$ is counted in the entry with state~$\bar{\rho_j}$ on~$v$ if $i + j \in \rho$.
The definition of $\bar{\sigma_j}$ is similar.

Let~$A^*_r$ be the result of the table for the right child~$r$ of~$e$ obtained by Step~1.
We can obtain the table~$A_r$ with the states on~$F$ transformed as described by a coordinate-wise application of the following formula on the vertices in~$F$.
The details are identical to the state changes in the proofs of Lemmas~\ref{lem:asymdsstates} and~\ref{lem:asympmstates}.
\begin{eqnarray*}
A_r(c_1 \times \{\bar{\rho_j}\} \times c_2) & = & \sum_{i+j \in \rho} A^*_r(c_1 \times \{\rho_i\} \times c_2) \\
A_r(c_1 \times \{\bar{\sigma_j}\} \times c_2) & = & \sum_{i+j \in \sigma} A^*_r(c_1 \times \{\sigma_i\} \times c_2)
\end{eqnarray*}
Notice that if we combine an entry with state~$\rho_j$ in~$A_l$ with an entry with state~$\bar{\rho_j}$ from~$A_r$, then we can count all valid combinations in which this vertex is not in the vertex set of a partial solution of the $[\rho,\sigma]$-domination problem.
The same is true for a combination with state~$\sigma_j$ in~$A_l$ with state~$\bar{\sigma_j}$ in~$A_r$ for vertices in the vertex set of the partial solutions.

As a final part of Step~2, we now change the states in~$A_l$ and~$A_r$ on the vertices in~$I$ and introduce the index vectors $\vec{i} = (i_{\rho1}, i_{\rho2}, \ldots, i_{\rho p},i_{\sigma1},i_{\sigma2},\ldots,i_{\sigma q})$, where~$i_{\rho j}$ and~$i_{\sigma j}$ index the sum of the number of neighbours in the vertex set of a partial solution of the $[\rho,\sigma]$-domination problem over the vertices with state $\rho_{\leq j}$ and $\sigma_{\leq j}$, respectively.
That is, we change the states used in~$A_l$ and~$A_r$ on vertices in~$I$ to State Set~II of Definition~\ref{def:rsstates} and introduce index vectors in exactly the same way as in the proof of Lemma~\ref{lem:rsstates2}, but only on the coordinates of the vertices in~$I$, similar to what we did in the proofs of Lemmas~\ref{lem:asymdsstates} and~\ref{lem:asympmstates}.
Because the states~$\rho_{\leq j}$ and~$\sigma_{\leq j}$ are used only on~$I$, we note that that the component~$i_{\rho_j}$ of the index vector~$\vec{i}$ count the total number of neighbours in the vertex sets of the partial solutions of the $[\rho,\sigma]$-domination problem of vertices with state~$\rho_{\leq j}$ on~$I$.
As a result, we obtain tables~$A'_l$ and~$A'_r$ with entries $A'_l(c_l,\kappa_l,\vec{g})$ and $A'_r(c_r,\kappa_r,\vec{h})$ with index vectors~$\vec{g}$ and~$\vec{h}$, where these entries have the same meaning as in Theorem~\ref{thrm:rstwalg}.
We note that the components~$i_{\rho p}$ and~$i_{\sigma q}$ of the index vector are omitted if~$\rho$ or~$\sigma$ is cofinite, respectively.

We have now performed all relevant preprocessing and are ready for the final step.

{\it Step 3}: Now, we construct the table~$A_e$ by computing the number of valid combinations from~$A'_l$ and~$A'_r$ using fast matrix multiplication.

We first define when three colourings~$c_e$, $c_l$, and~$c_r$ match.
They \emph{match} if:
\begin{itemize}
\item For any $v \in I$: $c_e(v) = c_l(v) = c_r(v)$ with State Set~II. ($s$ possibilities)
\item For any $v \in F$: either $c_l(v) = \rho_j$ and $c_r(v) = \bar{\rho_j}$, or $c_l(v) = \sigma_j$ and $c_r(v) = \bar{\sigma_j}$, with State Set~I used for $c_l$ and the new states used for $c_r$. ($s$ possibilities)
\item For any $v \in L$: $c_e(v) = c_l(v)$ with State Set~I. ($s$ possibilities)
\item For any $v \in R$: $c_e(v) = c_r(v)$ with State Set~I. ($s$ possibilities)
\end{itemize}
State Set~I and State Set~II are as defined in Definition~\ref{def:rsstates}.
That is, colourings match if they forget valid combinations on $F$, and have identical states on~$I$, $L$, and~$R$.

Using this definition, the following formula computes the table~$A'_e$.
The function of this table is identical to the same table in the proof of Theorem~\ref{thrm:rstwalg}: the table gives all valid combinations of entries corresponding to the colouring~$c$ that lead to a partial solution of size~$\kappa$ with the given values of the index vector~$\vec{i}$.
The index vectors allow us to extract the values we need afterwards.
\[ A'_e(c_e,\kappa,\vec{i}) = \!\!\!
\mathop{\sum_{c_e, c_l, c_r}}_{\textrm{\scriptsize match}} \;\;
\sum_{\kappa_l + \kappa_r = \kappa + \#_\sigma(c)} \!\!
\left( \sum_{i_{\rho1}=g_{\rho1}+h_{\rho1}} \!\!\cdots\!\! \sum_{i_{\sigma q}=g_{\sigma q}+h_{\sigma q}} \!\!\!\!
A'_l(c_l,\kappa_l,\vec{g}) \cdot A'_r(c_r,\kappa_r,\vec{h})
\right) \]
Here, $\#_\sigma = \#_\sigma(c_r(I \cup F))$ is the number of vertices that are assigned a $\sigma$-state on $I \cup F$ in any matching triple~$c_e$, $c_l$, $c_r$.

We will now argue what kind of entries the table~$A'_e$ contains by giving a series of observations.
\begin{observation} \label{obs:1}
For a combination of a partial solutions on~$G_l$ counted in~$A'_l$ and a partial solution on~$G_r$ counted in~$A'_r$ to be counted in the summation for $A'_e(c,\kappa,\vec{i})$, it is required that both partial solutions contains the same vertices on $X_l \cap X_r$ ($= I \cap F$).
\end{observation}
\begin{proof}
This holds because sets of matching colourings have a $\sigma$-state on a vertex if and only if the other colourings in which the same vertex is included also have a $\sigma$-state on this vertex.
\end{proof}

\begin{observation} \label{obs:2}
For a combination of a partial solutions on~$G_l$ counted in~$A'_l$ and a partial solution on~$G_r$ counted in~$A'_r$ to be counted in the summation for $A'_e(c,\kappa,\vec{i})$, it is required that the total number of vertices that are part of the combined partial solution is~$\kappa$.
\end{observation}
\begin{proof}
This holds because we demand that~$\kappa$ equals the sum of the sizes of the partial solutions on~$G_l$ and~$G_r$ used for the combination minus the number of vertices in these partial solutions that are counted in both sides, namely, the vertices with a $\sigma$-state on~$I$ or~$F$.
\end{proof}

\begin{observation} \label{obs:3}
For a combination of a partial solutions on~$G_l$ counted in~$A'_l$ and a partial solution on~$G_r$ counted in~$A'_r$ to be counted in the summation for $A'_e(c,\kappa,\vec{i})$, it is required that the subscripts~$j$ of the states~$\rho_j$ and~$\sigma_j$ used in~$c$ on vertices in~$L$ and~$R$ correctly count the number of neighbours of this vertex in $V_e \setminus X_e$ in the combined partial solution.
\end{observation}
\begin{proof}
This holds because of the preprocessing we performed in Step~1.
\end{proof}

\begin{observation} \label{obs:4}
For a combination of a partial solutions on~$G_l$ counted in~$A'_l$ and a partial solution on~$G_r$ counted in~$A'_r$ to be counted in the summation for $A'_e(c,\kappa,\vec{i})$, it is required that the forgotten vertices in a combined partial solution satisfy the requirements imposed by the specific $[\rho,\sigma]$-domination problem.
I.e., if such a vertex is not in the vertex set~$D$ of the combined partial solutions, then it has a number of neighbours in~$D$ that is a member of~$\rho$, and if such a vertex is in the vertex set~$D$ of the combined partial solution, then it has a number of neighbours in~$D$ that is a member of~$\sigma$.
Moreover, all such combinations are considered.
\end{observation}
\begin{proof}
This holds because we combine only entries with the states~$\rho_j$ and~$\bar{\rho_j}$ or with the states~$\sigma_j$ and~$\bar{\sigma_j}$ for vertices in~$F$.
These are all required combinations by definition of the states~$\bar{\rho_j}$ and~$\bar{\sigma_j}$.
\end{proof}

\begin{observation} \label{obs:5}
For a combination of a partial solutions on~$G_l$ counted in~$A'_l$ and a partial solution on~$G_r$ counted in~$A'_r$ to be counted in the summation for $A'_e(c,\kappa,\vec{i})$, it is required that the total sum of the number of neighbours outside~$X_e$ of the vertices with state~$\rho_{\leq j}$ or~$\sigma_{\leq j}$ in a combined partial solution equals~$i_{\rho j}$ or~$i_{\sigma j}$, respectively.
\end{observation}
\begin{proof}
This holds because of the following.
First the subscripts of the states are updated such that every relevant vertex is counted exactly once in Step~1.
Then, these numbers are stored in the index vectors at Step~2.
Finally, the entries of~$A'_e$ corresponding to a given index vector combine only partial solutions which index vectors sum to the given index vector~$\vec{i}$.
\end{proof}

\begin{observation} \label{obs:6}
Let $D_l$ and $D_r$ be the vertex set of a partial solution counted in~$A_l$ and~$A_r$ that are used to create a combined partial solution with vertex set~$D$, respectively. 
After the preprocessing of Step~1, the vertices with state~$\rho_{\leq j}$ or~$\sigma_{\leq j}$ have at most~$j$ neighbours that we count in the vertex sets~$D_l$ and~$D_r$, respectively.
And, if a vertex in the partial solution from~$A_l$ has~$i$ such counted neighbours in~$D_l$, and the same vertex in the partial solution from~$A_r$ has~$j$ such counted neighbours in~$D_r$, then the combined partial solution has a total of $i+j$ neighbours in~$D$ outside of~$X_e$.
\end{observation}
\begin{proof}
The last statement holds because we count each relevant neighbour of a vertex either in the states used in~$A_l$ or in the states used in~$A_r$ by the preprocessing of Step~1.
The first part of the statement follows from the definition of the states~$\rho_{\leq j}$ or~$\sigma_{\leq j}$: here, only partial solutions that previously had a state~$\rho_i$ and~$\sigma_i$ with $i \leq j$ are counted.
\end{proof}

We will now use Observations~\ref{obs:1}-\ref{obs:6} to show that we can compute the required values for $A_e$ in the following way.
This works very similar to Theorem~\ref{thrm:rstwalg}.
First, we change the states in the table~$A'_e$ back to State Set~I (as defined in Definition~\ref{def:rsstates}).
We can do so similar as in Lemma~\ref{lem:rsstates} and Lemmas~\ref{lem:asymdsstates} and~\ref{lem:asympmstates}.
Then, we extract the entries required for the table~$A_e$ using the following formula:
\[ A_e(c,\kappa) = A'_e \left( c, \; \kappa, \; (\Sigma_\rho^1(c),\Sigma_\rho^2(c),\ldots,\Sigma_\rho^p(c),\Sigma_\sigma^1(c),\Sigma_\sigma^2(c),\ldots,\Sigma_\sigma^q(c)) \; \right) \]
Here, $\Sigma_\rho^l(c)$ and $\Sigma_\sigma^l(c)$ are defined as in the proof of Theorem~\ref{thrm:rstwalg}: the weighted sums of the number of~$\rho_j$- and~$\sigma_j$-states with $0 \leq j \leq l$, respectively.

If~$\rho$ or~$\sigma$ is cofinite, we use the same formula but omit the components $\Sigma_\rho^p(c)$ or $\Sigma_\sigma^q(c)$ from the index vector of the extracted entries, respectively.

That the values of these entries equal the values we want to compute follows from the following reasoning.
First of all, any combination considered leads to a new partial solution since it uses the same vertices (Observation~\ref{obs:1}) and forgets vertices that satisfy the constraints of the fixed $[\rho,\sigma]$-domination problem (Observation~\ref{obs:4}).
Secondly, the combinations lead to combined partial solutions of the correct size (Observation~\ref{obs:2}).
Thirdly, the subscripts of the states used in~$A_e$ correctly count the number of neighbours of these vertices in the vertex set of the partial solution in $V_e \setminus X_e$.
For vertices in~$L$ and~$R$, this directly follows from Observation~\ref{obs:3} and the fact that for any three matching colourings the states used on each vertex in~$L$ and~$R$ are the same.
For vertices in~$I$, this follows from exactly the same arguments as in the last part of the proof of Theorem~\ref{thrm:rstwalg} using Observation~\ref{obs:5} and Observation~\ref{obs:6}.
This is the argument where we first argue that any entry which colouring uses only the states~$\rho_0$ and~$\sigma_0$ is correct, and thereafter inductively proceed to~$\rho_j$ and~$\sigma_j$ for $j > 0$ by using correctness for $\rho_{j-1}$ and $\sigma_{j-1}$ and fact that we use the entries corresponding to the chosen values of the index vectors.

All in all, we see that this procedure correctly computes the required table~$A_e$.

\smallskip
After computing~$A_e$ in the above way for all $e \in E(T)$, we can find the number of $[\rho,\sigma]$-dominating sets of each size in the table $A_{\{y,z\}}$, where~$z$ is the root of~$T$ and~$y$ its only child because $G = G_{\{y,z\}}$ and $X_{\{y,z\}} = \emptyset$.

\smallskip
For the running time, we note that we have to compute the tables~$A_e$ for the $\bigO(m)$ edges $e \in E(T)$.
For each table~$A_e$, the running time is dominated by evaluating the formula for the intermediate table~$A'_e$ with entries $A'_e(c,\kappa,\vec{i})$.
We can evaluate each summand of the formula for~$A'_e$ for all combinations of matching states by $s^{|I|}$ matrix multiplications as in Theorem~\ref{thrm:countdsbwalg}.
This requires $\bigO(n^2 (sk)^{2(s-2)} s^{|I|})$ multiplications of an $s^{|L|} \times s^{|F|}$ matrix and an $s^{|F|} \times s^{|R|}$ matrix.
The running time is maximal if $|I| = 0$ and $|L|=|R|=|F| = \frac{k}{2}$.
In this case, the total running time equals $\bigO(m n^2 (sk)^{2(s-2)} s^{\frac{\omega}{2}k} i_\times(n) )$ since we can do the computations using $n$-bit numbers.
\end{proof}

Similar to our results on the $[\rho,\sigma]$-domination problem on tree decompositions, we can improve the polynomial factors of the above algorithm in several ways.
The techniques involved are identical to those of Corollaries~\ref{cor:generalrstwalg}, \ref{cor:defluiterrstwalg}, and~\ref{cor:decisionrstwalg}.
Similar to Section~\ref{sec:rhosigmatw}, we define the value~$r$ associated with a $[\rho,\sigma]$-domination problems as follows:
\[ r = \left\{ \begin{array}{ll} \max\{p-1,q-1\} & \textrm{if $\rho$ and $\sigma$ are cofinite} \\ \max\{p,q-1\} & \textrm{if $\rho$ is finite and $\sigma$ is cofinite} \\ \max\{p-1,q\} & \textrm{if $\rho$ is confinite and $\sigma$ is finite} \\ \max\{p,q\} & \textrm{if $\rho$ and $\sigma$ are finite} \end{array} \right. \]

\begin{corollary}[General {$\boldsymbol{[\rho, \sigma]}$}-Domination Problems] \label{cor:generalrsbwalg}
Let $\rho, \sigma \subseteq \N$ be finite or cofinite, and let~$p$, $q$, $r$ and~$s$ be the values associated with the corresponding $[\rho,\sigma]$-domination problem.
There is an algorithm that, given a branch decomposition of a graph~$G$ of width~$k$, computes the number of $[\rho,\sigma]$-dominating sets in~$G$ of each size~$\kappa$, $0 \leq \kappa \leq n$, in $\bigO(m n^2 (rk)^{2r} s^{\frac{\omega}{2}k} i_\times(n))$ time.
Moreover, there is an algorithm that decides whether there exist a $[\rho,\sigma]$-dominating set of size~$\kappa$, for each individual value of $\kappa$, $0 \leq \kappa \leq n$, in $\bigO(m n^2 (rk)^{2r} s^{\frac{\omega}{2}k} i_\times(log(n)+k\log(r)))$ time.
\end{corollary}
\begin{proof}
Apply the modifications to the algorithm of Theorem~\ref{thrm:rstwalg} that we have used in the proof of Corollary~\ref{cor:generalrstwalg} for $[\rho,\sigma]$-domination problems on tree decompositions to the algorithm of Theorem~\ref{thrm:rsbwalg} for the same problems on branch decompositions.
\end{proof}

\begin{corollary}[{$\boldsymbol{[\rho, \sigma]}$}-Optimisation Problems with the de Fluiter Property] \label{cor:defluiterrsbwalg}
Let $\rho, \sigma \subseteq \N$ be finite or cofinite, and let~$p$, $q$, $r$ and~$s$ be the values associated with the corresponding $[\rho,\sigma]$-domination problem.
If the standard representation using State Set~I of the minimisation (or maximisation) variant of this $[\rho,\sigma]$-domination problem has the de Fluiter property for treewidth with function~$f$, then there is an algorithm that, given a branch decomposition of a graph~$G$ of width~$k$, computes the number of minimum (or maximum) $[\rho,\sigma]$-dominating sets in~$G$ in $\bigO(m [f(k)]^2 (rk)^{2r} s^{\frac{\omega}{2}k} i_\times(n))$ time.
Moreover, there is an algorithm that computes the minimum (or maximum) size of such a $[\rho,\sigma]$-dominating set in $\bigO(m [f(k)]^2 (rk)^{2r} s^{\frac{\omega}{2}k} i_\times(log(n)+k\log(r)))$ time.
\end{corollary}
\begin{proof}
Improve the result of Corollary~\ref{cor:generalrsbwalg} in the same way as Corollary~\ref{cor:defluiterrstwalg} improves Corollary~\ref{cor:generalrstwalg} on tree decompositions.
\end{proof}

\begin{corollary}[{$\boldsymbol{[\rho, \sigma]}$}-Decision Problems]
Let $\rho, \sigma \subseteq \N$ be finite or cofinite, and let~$p$, $q$, $r$ and~$s$ be the values associated with the corresponding $[\rho,\sigma]$-domination problem.
There is an algorithm that, given a branch decomposition of a graph~$G$ of width~$k$, counts the number of $[\rho,\sigma]$-dominating sets in~$G$ of a fixed $[\rho,\sigma]$-domination problem in $\bigO(m (rk)^{2r} s^{\frac{\omega}{2}k} i_\times(n))$ time.
Moreover, there is an algorithm that decides whether there exists a $[\rho,\!\sigma]$-dominating set in $\bigO(m (rk)^{2r} s^{\frac{\omega}{2}k} i_\times(log(n)+k\log(r)))$ time.
\end{corollary}
\begin{proof}
Improve the result of Corollary~\ref{cor:defluiterrsbwalg} in the same way as Corollary~\ref{cor:decisionrstwalg} improves upon Corollary~\ref{cor:defluiterrstwalg} on tree decompositions.
\end{proof}

\section{Dynamic Programming on Clique Decompositions} \label{sec:cliquewidth}
On graphs of bounded cliquewidth, we mainly consider the {\sc Dominating Set} problem.
We show how to improve the complex $\bigOs(8^{k})$-time algorithm, which computes a boolean decomposition of a graph of cliquewidth at most $k$ to solve the {\sc Dominating Set} problem~\cite{Bui-XuanTV09}, to an $\bigOs(4^{k})$-time algorithm.
Similar results for {\sc Independent Dominating Set} and {\sc Total Dominating Set} follow from the same approach.

\begin{theorem} \label{thrm:dscwalg}
There is an algorithm that, given a $k$-expression for a graph $G$, computes the number of dominating sets in $G$ of each size $0 \leq \kappa \leq n$ in $\bigO(n^{3} (k^{2} + i_{\times}(n))\, 4^{k})$ time.
\end{theorem}
\begin{proof}
An operation in a $k$-expression applies a procedure on zero, one, or two labelled graphs with labels $\{1,2,\ldots,k\}$ that transforms these labelled graphs into a new labelled graph with the same set of labels.
If~$H$ is such a labelled graph with vertex set~$V$, then we use $V(i)$ to denote the vertices of~$H$ with label~$i$.

For each labelled graph~$H$ obtained by applying an operation in a $k$-expression, we will compute a table~$A$ with entries $A(c,\kappa)$ that store the number of partial solutions of {\sc Dominating Set} of size~$\kappa$ that satisfy the constraints defined by the colouring~$c$.
In contrast to the algorithms on tree and branch decompositions, we do not use colourings that assign a state to each individual vertex, but colourings that assign states to the sets $V(1), V(2), \ldots, V(k)$.

The states that we use are similar to the ones used for {\sc Dominating Set} on tree decompositions and branch decomposition.
The states that we use are tuples representing two attributes: inclusion and domination. 
The first attribute determines whether at least one vertex in $V(i)$ is included in a partial solution.
We use states~$1$, $0$, and $?$ to indicate whether this is true, false, or any of both, respectively.
The second attribute determines whether all vertices of $V(i)$ are dominated in a partial solution.
Here, we also use states~$1$, $0$, and $?$ to indicate whether this is true, false, or any of both, respectively.
Thus, we get tuples of the form $(s,t)$, where the first components is related to inclusion and the second to domination, e.g., $(1,?)$ for vertex set $V(i)$ represents that the vertex set contains a vertex in the dominating set while we are indifferent about whether all vertices in $V(i)$ are dominated.

We will now show how to compute the table~$A$ for a $k$-expression obtained by using any of the four operations on smaller $k$-expressions that are given with similar tables for these smaller $k$-expressions.
This table~$A$ contains an entry for every colouring~$c$ of the series of vertex sets $\{V(1),V(2),\ldots,V(k)\}$ using the four states $(1,1)$, $(1,0)$, $(0,1)$, and $(0,0)$.
We note that the other states will be used to perform intermediate computations.
By a recursive evaluation, we can compute~$A$ for the $k$-expression that evaluates to~$G$.

\smallskip \noindent {\it Create a new graph}: In this operation, we create a new graph~$H$ with one vertex~$v$ with any label $j \in \{1,2,\ldots,k\}$.
We assume, without loss of generality by permuting the labels, that $j = k$.
We compute~$A$ by using the following formula where~$c$ is a colouring of the first $k-1$ vertex sets $V(i)$ and~$c_k$ is the state of $V(k)$:
\[ A( c \times \{c_k\} ,\kappa) = \left\{ \begin{array}{ll} 1 & \textrm{if $c_k=(1,1)$, $\kappa = 1$, and $c = \{(0,1)\}^{k-1}$} \\ 1 & \textrm{if $c_k=(0,0)$, $\kappa = 0$, and $c = \{(0,1)\}^{k-1}$} \\ 0 & \textrm{otherwise} \end{array} \right. \]
Since~$H$ has only one vertex and this vertex has label~$k$, the vertex sets for the other labels cannot have a dominating vertex, therefore the first attribute of their state must be~$0$.
Also, all vertices in these sets are dominated, hence the second attribute of their state must be~$1$.
The above formula counts the only two possibilities: either taking the vertex in the partial solution or not.

\smallskip \noindent {\it Relabel}: In this operation, all vertices with some label $i \in \{1,2,\ldots,k\}$ are relabelled such that they obtain the label $j \in \{1,2,\ldots,k\}$, $j \not= i$.
We assume, without loss of generality by permuting the labels, that $i = k$ and $j=k-1$.

Let~$A'$ be the table belonging to the $k$-expression before the relabelling and let~$A$ be the table we need to compute.
We compute~$A$ using the following formulas:
\begin{eqnarray*}
A( c \!\times\! \{(0,1)\} \!\times\! \{(i,d)\},\kappa) & \!\!=\!\! & \mathop{\sum_{i_1, i_2 \in \{0,1\}}}_{\max\{i_1, i_2\} = i} \; \mathop{\sum_{d_1, d_2 \in \{0,1\}}}_{\min\{d_1,d_2\} = d} \!\! A'(c \!\times\! \{(i_1,d_1)\} \!\times\! \{(i_2,d_2)\}, \kappa) \\
A( c \!\times\! \{(i^*,d^*)\} \!\times\! \{(i,d)\},\kappa) & \!\!=\!\! & 0 \qquad \qquad \textrm{if $(i^*,d^*) \not= (0,1)$}
\end{eqnarray*}
These formula correctly compute the table~$A$ because of the following observations.
For $V(i)$, the first attribute must be~$0$ and the second attribute must be~$1$ in any valid partial solution because $V(i) = \emptyset$ after the operations; this is similar to this requirement in the `create new graph' operation. 
If $V(j)$ must have a vertex in the dominating set, then this vertex must be in $V(i)$ or $V(j)$ originally.
And, if all vertices in $V(j)$ must be dominated, then all vertices in $V(i)$ and $V(j)$ must be dominated.
Note that the minimum and maximum under the summations correspond to `and' and `or' operations, respectively.

\smallskip \noindent {\it Add edges}: In this operation, all vertices with some label $i \in \{1,2,\ldots,k\}$ are connected to all vertices with another label $j \in \{1,2,\ldots,k\}$, $j \not= i$.
We again assume, without loss of generality by permuting the labels, that $i = k-1$ and $j=k$.

Let~$A'$ be the table belonging to the $k$-expression before adding the edges and let~$A$ be the table we need to compute.
We compute~$A$ using the following formula:
\[ A(c \times \{(i_1,d_1)\} \times \{(i_2,d_2)\}, \kappa) = \mathop{\sum_{d'_1 \in \{0,1\}}}_{\max\{d'_1,i_2\} = d_1} \; \mathop{\sum_{d'_2 \in \{0,1\}}}_{\max\{d'_2,i_1\} = d_2} \!\! A(c \times (i_1, d'_1) \times (i_2, d'_2), \kappa) \]
This formula is correct as a vertex sets $V(i)$ and $V(j)$ contain a dominating vertex if and only if they contained such a vertex before adding the edges.
For the property of domination, correctness follows because the vertex sets $V(i)$ and $V(j)$ are dominated if and only if they were either dominated before adding the edges, or if they become dominated by a vertex from the other vertex set because of the adding of the edges.

\smallskip \noindent {\it Join graphs}:
This operation joins two labelled graphs~$H_1$ and~$H_2$ with tables~$A_1$ and~$A_2$ to a labelled graph~$H$ with table~$A$.
To do this efficiently, we first apply state changes similar to those used in Sections~\ref{sec:treewidth} and~\ref{sec:branchwidth}.
We use states~$0$ and~$?$ for the first attribute (inclusion) and states~$1$ and~$?$ for the second attribute (domination).

Changing~$A_1$ and~$A_2$ to tables~$A^*_1$ and~$A^*_2$ that use this set of states can be done in a similar manner as in Lemmas~\ref{lem:dsstates} and~\ref{lem:asymdsstates}.
We first copy~$A_y$ into~$A^*_y$, for $y \in \{1,2\}$ and then iteratively use the following formulas in a coordinate-wise manner:
\begin{eqnarray}
A^*_y(c_1 \times (0,1) \times c_2,\kappa) & = & A^*_y(c_1 \times (0,1) \times c_2, \kappa) \\
A^*_y(c_1 \times (?,1) \times c_2,\kappa) & = & A^*_y(c_1 \times (1,1) \times c_2, \kappa) + A^*_y(c_1 \times (0,1) \times c_2, \kappa) \\
A^*_y(c_{1} \times (0,?) \times c_2,\kappa) & = & A^*_y(c_1 \times (0,1) \times c_2, \kappa) + A^*_y(c_1 \times (0,0) \times c_2, \kappa) \\
A^*_y(c_1 \times (?,?) \times c_2,\kappa) & = & A^*_y(c_1 \times (1,1) \times c_2, \kappa) + A^*_y(c_1 \times (1,0) \times c_2, \kappa) +\\
&& A^*_y(c_1 \times (0,1) \times c_2, \kappa) + A^*_y(c_1 \times (0,0) \times c_2, \kappa)
\end{eqnarray}

We have already seen many state changes similar to these in Sections~\ref{sec:treewidth} and~\ref{sec:branchwidth}.
Therefore, it is not surprising that we can now compute the table $A^*$ in the following way, where the table $A^*$ is the equivalent of the table $A$ we want to compute only using the different set of states:
\[ A^*(c, \kappa) = \sum_{\kappa_1 + \kappa_2 = \kappa} A^{*}_{1}(c,\kappa_1) \cdot A^{*}_{2}(c,\kappa_2) \]

Next, we apply state changes to obtain $A$ from $A^*$.
These state changes are the inverse of those given above.
Again, first copy~$A^*$ into~$A$ and then iteratively transform the states in a coordinate-wise manner using the following formulas:
\begin{eqnarray*}
A(c_1 \times (0,1) \times c_2, \kappa) & = & A(c_1 \times (0,1) \times c_2, \kappa) \\
A(c_1 \times (1,1) \times c_2, \kappa) & = & A(c_1 \times (?,1) \times c_2, \kappa) - A(c_1 \times (0,1) \times c_2, \kappa) \\
A(c_1 \times (0,0) \times c_2, \kappa) & = & A(c_1 \times (0,?) \times c_2, \kappa) - A(c_1 \times (0,1) \times c_2, \kappa) \\
A(c_1 \times (1,0) \times c_2, \kappa) & = & A(c_1 \times (?,?) \times c_2, \kappa) - A(c_1 \times (0,?) \times c_2, \kappa) \\
&&- A(c_1 \times (?,1) \times c_2, \kappa) + A(c_1 \times (0,1) \times c_2, \kappa)
\end{eqnarray*}
Correctness of the computed table~$A$ follows by exactly the same reasoning as used in Theorem~\ref{thrm:countingtwdsalg} and in Proposition~\ref{prop:secondbwdsalg}.
We note that the last of the above formulas is a nice example of an application of the principle of inclusion/exclusion: to find the number of sets corresponding to the $(1,0)$ -state, we take the number of sets corresponding to the $(?,?)$-state; then, we subtract what we counted to much, but because we subtract some sets twice, we need to add some number of sets again to obtain the required value.

\smallskip
The number of dominating sets in~$G$ of size~$\kappa$ can be computed from the table~$A$ related to the final operation of the $k$-expression for~$G$.
In this table, we consider only the entries in which the second property is~$1$, i.e., the entries corresponding to partial solutions in which all vertices in~$G$ are dominated.
Now, the number of dominating sets in~$G$ of size~$\kappa$ equals the sum over all entries $A(c,\kappa)$ with $c \in \{(0,1),(1,1)\}$.

\smallskip
For the running time, we observe that each of the $\bigO(n)$ join operations take $\bigO(n^{2} 4^{k} i_{\times}(n))$ time because we are multiplying $n$-bit numbers.
Each of the $\bigO(nk^{2})$ other operations take $\bigO(n^{2} 4^{k})$ time since we need $\bigO(n4^k)$ series of a constant number of additions using $n$-bit numbers, and $i_+(n)=\bigO(n)$.
The running time of $\bigO(n^{3} (k^{2} + i_{\times}(n)) \, 4^{k})$ follows.
\end{proof}

Similar to the algorithms for {\sc Dominating Set} on tree decompositions and branch decompositions in Section~\ref{sec:dstw} and~\ref{sec:bwds}, we can improve the polynomial factors in the running time if we are interested only in the size of a minimum dominating set, or the number of these sets.
To this end, we will introduce a notion of a de Fluiter property for cliquewidth.
This notion is defined similarly to the de Fluiter property for treewidth; see Definition~\ref{def:fluiterproptw}.
\begin{definition}[de Fluiter property for cliquewidth] \label{def:fluiterpropcw}
Consider a method to represent the different partial solutions used in an algorithm that performs dynamic programming on clique decompositions ($k$-expressions) for an optimisation problem~$\Pi$.
Such a representation has the \emph{de Fluiter property for cliquewidth} if the difference between the objective values of any two partial solutions of~$\Pi$ that are stored for a partially evaluated $k$-expression and that can both still lead to an optimal solution is at most $f(k)$, for some function $f$.
Here, the function~$f$ depends only on the cliquewidth~$k$.
\end{definition}

The definition of the de Fluiter property for cliquewidth is very similar to the same notion for treewidth.
However, the structure of a $k$-expression is different from tree decompositions and branch decompositions in such a way that the de Fluiter property for cliquewidth does not appear to be equivalent to the other two.
This in contrast to the same notion for branchwidth that is equivalent to this notion for treewidth; see Section~\ref{sec:bwframework}.
The main difference is that $k$-expressions deal with sets of equivalent vertices instead of the vertices themselves.

The representation used in the algorithm for the {\sc Dominating Set} problem above also has the de Fluiter property for cliquewidth.
\begin{lemma}
The representation of partial solutions for the {\sc Dominating Set} problem used in Theorem~\ref{thrm:dscwalg} has the de Fluiter property for cliquewidth with $f(k) = 2k$.
\end{lemma}
\begin{proof}
Consider any partially constructed graph~$H$ from a partial bottom-up evaluation of the $k$-expression for a graph~$G$, and let~$S$ be the set of vertices of the smallest remaining partial solution stored in the table for the subgraph~$H$.
We prove the lemma by showing that by adding at most~$2k$ vertices to~$S$, we can dominate all future neighbours of the vertices in~$H$ and all vertices in~$H$ that will receive future neighbours.
We can restrict ourselves to adding vertices to~$S$ that dominate these vertices and not vertices in~$H$ that do not receive future neighbours, because Definition~\ref{def:fluiterpropcw} considers only partial solutions on~$H$ that can still lead to an optimal solution on~$G$.
Namely, a vertex set $V(i)$ that contains undominated vertices that will not receive future neighbours when continuing the evaluation of the $k$-expression will not lead to an optimal solution on~$G$.
This is because the selection of the vertices that will be in a dominating set happens only in the `create a new graph' operations.

We now show that by adding at most~$k$ vertices to~$S$, we can dominate all vertices in~$H$, and by adding another set of at most~$k$ vertices to~$S$, we can dominated all future neighbours of the vertices in~$H$.
To dominate all future neighbours of the vertices in~$H$, we can pick one vertex from each set $V(i)$.
Next, consider dominating the vertices in each of the vertex sets $V(i)$ and are not yet dominated and that will receive future neighbours.
Since the `add edges' operations of a $k$-expression can only add edges between future neighbours and all vertices with the label~$i$, and since the `relabel' operation can only merges the sets $V(i)$ and not split them, we can add a single vertex to~$S$ that is a future neighbour of a vertex in $V(i)$ to dominate all vertices in $V(i)$.
\end{proof}

Using this property, we can easily improve the result of Theorem~\ref{thrm:dscwalg} for the case where we want to count only the number of minimum dominating sets.
This goes in a way similar to Corollaries~\ref{cor:countmdstwalg}, \ref{cor:solvedstwalg}, \ref{cor:countmdsbwalg}, and~\ref{cor:solvedsbwalg}.
\begin{corollary} \label{cor:countmdscwalg}
There is an algorithm that, given a $k$-expression for a graph~$G$, computes the number of minimum dominating sets in~$G$ in $\bigO(n k^2 4^k i_\times(n))$ time.
\end{corollary}
\begin{proof}
For each colouring~$c$, we maintain the size $B(c)$ of any minimum partial dominating set inducing~$c$, and the number $A(c)$ of such sets.
This can also be seen as a table $D(c)$ of tuples.
Define a new function $\oplus$ such that
$$(A(c), B(c)) \oplus (A(c'), B(c')) \ = \ \left\{\begin{array}{ll}
(A(c) + A(c'), B(c)) & \mbox{if}\ B(c) = B(c')\\
(A(c^{*}), B(c^{*})) & \mbox{otherwise}
\end{array}\right.$$
where $c^{*} = \arg\min\{B(c),B(c')\}$.
We will use this function to ensure that we count only dominating sets of minimum size.

We now modify the algorithm of Theorem~\ref{thrm:dscwalg} to use the tables~$D$. 
For the first three operations, simply omit the size parameter~$\kappa$ from the formula and replace any~$+$ by~$\oplus$.
For instance, the computation for the third operation that adds new edges connecting all vertices with label $V(i)$ to all vertices with label $V(j)$, becomes:
\[ D(c \times \{(i_1,d_1)\} \times \{(i_2,d_2)\}) = \mathop{\bigoplus_{d'_1 \in \{0,1\}}}_{\max\{d'_1,i_2\} = d_1} \; \mathop{\bigoplus_{d'_2 \in \{0,1\}}}_{\min\{d'_2, i_1\} = d_2} D(c \times (i_1, d'_1) \times (i_2, d'_2)) \]

For the fourth operation, where we take the union of two labelled graphs, we need to be more careful.
Here, we use that the given representation of partial solutions has the de Fluiter property for cliquewidth.
We first discard solutions that contain vertices that are undominated and will not receive new neighbours in the future, that is, we set the corresponding table entries to $D(c) = (\infty,0)$.
Then, we also discard any remaining solutions that are at least $2k$ larger than the minimum remaining solution.

Let $D_1(c) = (A_1(c),B_1(c))$ and $D_2(c) = (A_2(c),B_2(c))$ be the two resulting tables for the two labelled graphs~$H_1$ and~$H_2$ we need to join.
To perform the join operation, we construct tables $A_1(c,\kappa)$ and $A_2(c,\kappa)$ as follows, with $y \in \{1,2\}$:
\[ A_y(c, \kappa) = \left\{ \begin{array}{ll} A_y(c) & \mbox{if}\ B_y(c) = \kappa \\ 0 & \mbox{otherwise} \end{array} \right. \]
In these tables, $\kappa$ has a range of size~$2k$ and thus this table has size $\bigO(k\, 4^{k})$.

Now, we can apply the same algorithm for the join operations as described in Theorem~\ref{thrm:dscwalg}.
Afterwards, we retrieve the value of $D(c)$ by setting $A(c) = A(c,\kappa')$ and $B(c) = \kappa'$, where~$\kappa'$ is the smallest value of~$\kappa$ for which $A(c,\kappa)$ is non-zero.

\smallskip
For the running time, we observe that each of the $\bigO(k^2n)$ operations that create a new graph, relabel vertex sets, or add edges to the graph compute $\bigO(4^k)$ tuples that cost $\bigO(i_+(n))$ time each since we use a constant number of additions and comparisons of an $\log(n)$-bits number and an $n$-bits number.
Each of the $\bigO(n)$ join operations cost $\bigO(k^2 4^k i_\times(n))$ time because of the reduced table size.
In total, this gives a running time of $\bigO(n k^2 4^k i_\times(n))$.
\end{proof}

Finally, we show that one can use $\bigO(k)$-bit numbers when considering the decision version of this minimisation problem instead of the counting variant.

\begin{corollary} \label{cor:solvedscwalg}
There is an algorithm that, given a $k$-expression for a graph~$G$, computes the size of a minimum dominating sets in~$G$ in $\bigO(n k^2 4^k)$ time.
\end{corollary}
\begin{proof}
Maintain only the size $B(c)$ of any partial solution satisfying the requirements of the colouring~$c$ in the computations involved in any of the first three operations.
Store this table by maintaining the size~$\xi$ of the smallest solution in~$B$ that has no undominated vertices that will not get future neighbours, and let~$B$ contain $\bigO(\log(k))$-bit numbers giving the difference in size between the size of the partial solutions and the number~$\xi$; this is similar to, for example, Corollary~\ref{cor:solvedstwalg}.

For the fourth operation, follow the same algorithm as in Corollary~\ref{cor:countmdscwalg}, using $A(c,\kappa) = 1$ if $B(c)=\kappa$ and $A(c,\kappa) = 0$ otherwise.
Since the total sum of all entries in this table is $4^k$, the computations for the join operation can now be implemented using $\bigO(k)$-bit numbers.
See also, Corollaries~\ref{cor:solvedstwalg} and~\ref{cor:solvedsbwalg}.
In the computational model with $\bigO(k)$-bit word size that we use, the term in the running time for the arithmetic operations disappears since $i_\times(k) = \bigO(1)$.
\end{proof}

We conclude by noticing that $\bigOs(4^k)$ algorithms for {\sc Independent Dominating Set} and {\sc Total Dominating Set} follow from the same approach.
For {\sc Total Dominating Set}, we have to change only the fact that a vertex does not dominate itself at the 'create new graph' operations.
For {\sc Independent Dominating Set}, we have to incorporate a check such that no two vertices in the solution set become neighbours in the `add edges' operation.

\section{Relations Between the de Fluiter Properties and Finite Integer Index} \label{sec:fluiterprop}
In the previous sections, we have defined \emph{de Fluiter properties} for all three types of graph decompositions.
This property is highly related to the concept \emph{finite integer index} as defined in~\cite{BodlaenderA01}.
Finite integer index is a property used in reduction algorithms for optimisation problems on graphs of small treewidth~\cite{BodlaenderA01} and is also used in meta results in the theory of kernelisation~\cite{BodlaenderFLPST09}.
We will conclude by explaining the relation between the de Fluiter properties and finite integer index.

We start with a series of definitions.
Let a \emph{terminal graph} be a graph $G$ together with an ordered set of distinct vertices $X = \{x_1,x_2,\ldots,x_l\}$ with each $x_i \in V$.
The vertices $x_i \in X$ are called the \emph{terminals} of $G$.
For two terminal graphs $G_1$ and $G_2$ with the same number of terminals, the addition operation $G_1 + G_2$ is defined to be the operation that takes the disjoint union of both graphs, then identifies each pair of terminals with the same number $1,2,\ldots,t$, and finally removes any double edges created.

For a graph optimisation problem $\Pi$, Bodlaender and van Antwerpen-de Fluiter define an equivalence relation $\sim_{\Pi,l}$ on terminal graphs with $l$ terminals~\cite{BodlaenderA01}: $G_1 \sim_{\Pi,l} G_2$ if and only if there exists an $i \in \Z$ such that for all terminal graphs $H$ with $l$ terminals:
\[ \pi(G_1 + H) = \pi(G_2 + H) + i \]
Here, the function $\pi(G)$ assigns the objective value of an optimal solution of the optimisation problem $\Pi$ to the input graph $G$.

\begin{definition}[finite integer index]
An optimisation problem $\Pi$ is of \emph{finite integer index} if $\sim_{\Pi,l}$ has a finite number of equivalence classes for each fixed $l$.
\end{definition}

When one proves that a problem has finite integer index, one often gives a representation of partial solutions that has the de Fluiter property for treewidth; see for example~\cite{Fluiter97}.
That one can prove that a problem has finite integer index in this way can see from the following proposition.
\begin{proposition} \label{prop:fii}
If a problem $\Pi$ has a representation of its partial solutions of different characteristics that can be used in an algorithm that performs dynamic programming on tree decompositions and that has the de Fluiter property for treewidth, then $\Pi$ is of finite integer index.
\end{proposition}
\begin{proof}
Let $l$ be fixed, and consider an $l$-terminal graph $G$.
Construct a tree decomposition $T$ of $G$ such that bag associated the root of $T$ equals the set of terminals $X$ of $G$.
Note that this is always possible since we have not specified a bound on the treewidth of $T$.
For an $l$-terminal graph $H$, one can construct a tree decomposition of $G + H$ by making a similar tree decomposition of $H$ and identifying the roots, which both have the same vertex set $X$.

Let $G_1$, $G_2$ be two $l$-terminal graphs to which we both add another $l$-terminal graph $H$ through addition, i.e., $G_i + H$, and let $T_1$, $T_2$ be tree decompositions of these graphs obtained in the above way.
For both graph, consider the dynamic programming table constructed for the node $x_X$ associated with the vertex set $X$ by a dynamic programming algorithm for $\Pi$ that has the de Fluiter property for treewidth.
For these tables, we assume that the induced subgraph associated with $x_X$ of the decompositions equals $G_i$, that is, the bags of the nodes below $x_X$ contain all vertices in $G_i$, and vertices in $H$ only occur in bags associated with nodes that are not descendants of $x_X$ in $T_i$.

Clearly, $\pi(G_1 + H) = \pi(G_2 + H)$ if both dynamic programming tables are the same and $G_1[X] = G_2[X]$, that is, if the tables are equal and both graphs have the same edges between their terminals.
Let us now consider a more general case where we first normalise the dynamic programming tables such that the smallest valued entry equals zero, and all other entries contain the difference in value to this smallest entry.
In this case, it is not hard to see that if both normalised dynamic programming tables are equal and $G_1[X] = G_2[X]$, then there must exists an $i \in \Z$ such that $\pi(G_1 + H) = \pi(G_2 + H) + i$.

The dynamic programming algorithm for the problem $\Pi$ can compute only finite size tables.
Moreover, as the representation used by the algorithm has the de Fluiter property for treewidth, the normalised tables can only have values in the range $0,1,\ldots,f(k)$.
Therefore, there are only a finite number of different normalised tables and a finite number of possible induced subgraphs on $l$ vertices (terminals).
We conclude that the relation $\sim_{\Pi,l}$ has a finite number of equivalence classes.
\end{proof}

By the same Proposition it follows that problems that are not of finite integer index (e.g., {\sc Independent Dominating Set}) do not have a representation of partial solutions that has the de Fluiter property for treewidth.
We note that the converse of Proposition~\ref{prop:fii} is not necessarily true.

While we focused on the relation between finite integer index and the de Fluiter property for treewidth (or branchwidth), we do not know the relation between these concepts and the de Fluiter property for cliquewidth.
This property seems to be very different to the other two.
Finding the details on the relations between all these properties is beyond the scope of this paper as we only used the de Fluiter properties to improve the polynomial factors in the running times of the presented algorithms.

\section{Conclusion} \label{sec:conclusion}
We have presented faster algorithms for a broad range of problems on three different types of graph decompositions.
These algorithms were obtained by using generalisations of the fast subset convolution algorithm, sometimes combined with using fast multiplication of matrices.
On tree decompositions and clique decompositions the exponential factor in the running times equal the space requirement for such algorithms.
On branch decompositions, the running times of our algorithms come very close to this space bound.
Additionally, a further improvement of the exponential factor in the running time for some problems on tree decompositions would contradict the Strong Exponential-Time Hypothesis.

We like to mention that, very recently, $\bigOs(c^k)$-time algorithms for various problems on tree decompositions have been obtained, for some constant $c \geq 2$, for which previously only $\bigOs(k^k)$-time algorithms existed~\cite{CyganNPPRW11}.
This includes problems like {\sc Hamilton Cycle}, {\sc Feedback Vertex Set}, {\sc Steiner Tree}, {\sc Connected Dominating Set}.
Our techniques play an important role in this paper to make sure that the constants $c$ in these algorithms are small and equal the space requirement.
Here, $c$ is often also small enough such that no faster algorithms exist under the Strong Exponential-Time Hypothesis~\cite{CyganNPPRW11}.
It would be interesting to find a general result stating which properties a problem (or join operation on nice tree decompositions) must have to admit $\bigO(c^k)$-time algorithms, where $c$ is the space requirement of the dynamic programming algorithm.

To conclude, we note that, for some problems like counting perfect matchings, the running times of our algorithms on branch decompositions come close to the running times of the currently-fastest exact exponential-time algorithms for these problems.
For this we use that the branchwidth of any graph is at most $\frac{2}{3}n$, e.g., see~\cite{Hicks05}.
In this way, we directly obtain an $\bigO(2^{\frac{\omega}{2} \cdot \frac{2}{3} n}) = \bigO(1.7315^n)$-time algorithm for counting the number of perfect matchings.
This running time is identical to the fast-matrix-multiplication-based algorithm for this problem by Bj\"orklund et al.~\cite{BjorklundH08}.
We note that this result has recently been improved to $\bigO(1.6181^n)$ by Koivisto~\cite{Koivisto09}.
Our algorithm improves this result on graphs for which we can compute a branch decomposition of width at most $0.5844n$ in polynomial time; this is a very large family of graphs since this bound is not much smaller that the given upper bound of $\frac{2}{3}n$.

\subsubsection*{Acknowledgements}
The first author is grateful to Jan Arne Telle for introducing him to several problems solved in this paper at Dagstuhl seminar 08431, and for the opportunity to visit Bergen to work with Martin Vatshelle.
The first author also thanks Jesper Nederlof for several useful discussions.

\end{document}